\documentclass[pdflatex]{sn-jnl}

\usepackage{amsmath,amssymb,amsfonts,amsthm}
\usepackage{mathrsfs}%
\usepackage{xcolor,cite}
\usepackage{graphicx}
\usepackage{float,mathtools}
\usepackage{enumerate,enumitem}
\usepackage{makecell}
\usepackage[title]{appendix}%
\usepackage{textcomp}%
\usepackage{booktabs}%
\usepackage{algorithm}%
\usepackage{algorithmicx}%
\usepackage{algpseudocode}%
\usepackage{fullpage}

\newcommand{\myblue}[1]{\textcolor{black}{#1}}

\DeclareMathOperator\supp{supp}

\def\bpm{\begin{pmatrix}}
\def\epm{\end{pmatrix}}
\def\Real{\mathbb{R}}
\def\Complex{\mathbb{C}}

\def\rA{\mathrm{A}}
\def\A{\mathbf{A}}
\def\B{\mathbf{B}}
\def\conj{\mathbf{B}}
\def\c{\mathrm{c}}
\def\C{\mathbf{C}}
\def\sC{\mathscr{C}}
\def\d{\mathrm{d}}
\def\dsg{\mathrm{DSG}}

\def\E{\mathbf{E}}

\def\bGamma{\boldsymbol{\Gamma}}

\def\I{\mathbb{I}}

\def\M{\mathbf{M}}

\def\N{\mathbf{N}}

\def\O{\mathcal{O}}
\def\P{\mathbf{P}}
\newcommand\sP[1]{\mathop{\mathcal{P}_{#1}}}

\def\Q{\mathbf{Q}}
\def\s{\mathrm{s}}
\def\S{\mathbf{S}}

\def\u{\mathrm{u}}

\def\V{\mathbf{V}}

\def\W{\mathbf{W}}

\def\Y{\mathbf{Y}}
\def\Z{\mathbf{Z}}

\def\bphi{\boldsymbol{\phi}}

\def\brho{\boldsymbol{\rho}}

\def\vbone{\overrightarrow{\boldsymbol{1}}}

\def\diag{\mathop{\rm diag}\nolimits}

\def\Res{\mathop{\rm Res}}
\def\gl{\gtrless}

\def\sech{\mathrm{sech}}
\def\sg{\mathrm{sg}}

\def\conj#1{\overline{#1}}
\newcommand\set[2]{\{\,#1\mid#2\,\}}
\newcommand\Set[2]{\set{#1}{\text{#2}}}

\newcommand{\ee}{\mathrm{e}}
\newcommand{\ii}{\mathrm{i}}
\newcommand{\dd}{\,\mathrm{d}}

\def\mbe{\mathrm{MBE}}
\def\nls{\mathrm{NLS}}
\def\mkdv{\mathrm{mKdV}}

\newtheorem{rhp}{Riemann-Hilbert Problem}

\newtheorem{theorem}{Theorem}
\newtheorem{corollary}{Corollary}
\newtheorem{lemma}{Lemma}
\theoremstyle{definition}
\newtheorem{definition}{Definition}

\newtheorem{remark}{Remark}

\begin{document}

\title[On ZBG solitons of sharp-line MBEs]{\myblue{On zero-background solitons of the sharp-line Maxwell-Bloch equations}}

\author{\fnm{Sitai} \sur{Li}}
\email{sitaili@xmu.edu.cn}
\affil{
\orgdiv{School of Mathematical Sciences},
\orgname{Xiamen University},
\orgaddress{\city{Xiamen}, \postcode{361005}, \state{Fujian}, \\ \country{People's Republic of China}
}
}


\abstract{
This work is devoted to systematically study general $N$-soliton solutions
possibly containing multiple degenerate soliton groups (DSGs),
in the context of the sharp-line Maxwell-Bloch equations with a zero background.
We also show that results can be readily migrated to other integrable systems,
with the same non-self-adjoint Zakharov-Shabat scattering problem or alike.
Results for the focusing nonlinear Schr\"{o}dinger equation
and the complex modified Korteweg-De Vries equation
are obtained as explicit examples for demonstrative purposes.
A DSG is a localized coherent nonlinear traveling-wave structure,
comprised of inseparable solitons with identical velocities.
Hence,
DSGs are generalizations of single solitons
(considered as $1$-DSGs),
and form fundamental building blocks of solutions of many integrable systems.
We provide an explicit formula for an $N$-DSG and its center.
With the help of the Deift-Zhou's nonlinear steepest descent method,
we prove the localization of DSGs,
and calculate the long-time asymptotics for an arbitrary $N$-soliton solutions.
It is shown that
the solution becomes a linear combination of multiple DSGs in the distant past and future,
with explicit formul{\ae} for the asymptotic phase shift for each DSG.
Other generalizations of a single soliton are also discussed,
such as $N$th-order solitons and soliton gases.
\myblue{
We prove that
every $N$th-order soliton can be obtained by fusion of eigenvalues of $N$-soliton solutions,
with proper rescalings of norming constants,
and demonstrate that soliton-gas solution can be considered as limits of $N$-soliton solutions as $N\to+\infty$.
}
}




\maketitle

\section{Introduction}

This work studies the general $N$-soliton solutions
of the sharp-line Maxwell-Bloch equations (MBEs) with a zero background (ZBG),
particularly including degenerate solitons,
\myblue{in order to} present a framework of
systematical analysis on arbitrary multi-soliton solutions of integrable systems
with the same non-self-adjoint Zakharov-Shabat scattering problem or alike.
\myblue{
This work also demonstrates the relation
between multi-soliton solutions and high-order solitons
or between multi-soliton solutions and soliton gases.
Examples of generalizations to other integrable systems are presented as well,
which is achieved
by primarily analyzing solutions in the spectral space instead of the physical space,
with the help of the Deift-Zhou's nonlinear steepest descent method
for analyzing oscillatory Riemann-Hilbert problems (RHPs)~\cite{diz1993,dz1993}.
}

The MBEs are completely integrable~\cite{as1981,ae1975},
and have attracted great interest,
because of their importance in nonlinear optics~\cite{jrt2019},
successfully explaining self-induced transparency~\cite{mh1967,mh1969,mh1970}
and the closely-related phenomenon of superfluorescence~\cite{gzm1983,gzm1984,gzm1985,z1980}.
In particular,
the system describes light-matter interactions
in a semi-infinitely long one dimensional two-level optical medium
with $z\ge0$ for all time $t\in\Real$,
where the light pulse is injected at one end $z = 0$.
The general MBEs contain an integral with an arbitrary weight function,
describing the atom distribution of the optical medium according to the atoms' resonant frequencies.
This work is concerned with the ideal medium,
whose atom distribution is a Dirac delta,
corresponding to a sharp-line spectral shape.
Therefore,
the Cauchy problem for sharp-line MBEs are written as follows~\cite{akn1974,gzm1985},
\begin{equation}
\label{e:mbe}
\begin{aligned}
q_z(t,z) & = -P(t,z)\,,\quad
P_t(t,z) = - 2q(t,z)D(t,z)\,,\quad
D_t(t,z) = 2\Re(\conj{q(t,z)} P(t,z))\,,\\
q(t,0) & = q_0(t)\,,\quad
t\in\Real\,,\\
\lim_{t\to-\infty}q(t,z) & = 0\,,\quad
z\ge0\,,\\
D_- & \coloneqq \lim_{t\to-\infty}D(t,z) = \pm1\,,\quad
P_-\coloneqq \lim_{t\to-\infty}P(t,z) = 0\,,\quad
z\ge0\,.
\end{aligned}
\end{equation}
The variable $t$ and $z$ are the spatial and temporal variables in the comoving frame.
The quantity $q(t,z)\in\Complex$ describes the optical pulse,
$D(t,z)\in\Real$ denotes the population inversion,
and $P(t,z)\in\Complex$ is the polarization of the optical medium.
The subscript $t$ and $z$ denote partial derivatives and $\conj{q(t,z)}$ is the complex conjugate of $q(t,z)$.
Moreover, $q_0(t)$ is the input pulse,
and $\{D_-, P_-\}$ is the initial state of the medium in the distant past $t\to-\infty$.
It is easy to check that $\frac{\partial}{\partial t}(D^2 + |P|^2) = 0$ from the MBEs~\eqref{e:mbe}.
\myblue{
The values of $\{D_-, P_-\}$ are the results of the normalization $D^2 + |P|^2 = 1$ for $t\in\Real$ and $z\ge0$.
}
The trivial solutions $\{q\equiv0,D\equiv\pm1,P\equiv0\}$ are considered as the ZBGs,
denoting the static states of the system with absence of light.
Consequently,
the medium is in one of two pure states,
the ground state $D=-1$ or the excited state $D=1$.
Therefore,
the boundary condition $D_- = -1$ or $D_- = 1$ means that initially all atoms are in the ground state or excited state,
respectively.

The Lax pair for the system~\eqref{e:mbe} is given by
\begin{equation}
\label{e:laxpair}
\begin{aligned}
\bphi_t & = (\ii \lambda \sigma_3 + \Q)\,\bphi\,,\qquad
\bphi_z = -\frac{\ii}{2\lambda}\brho(t,z)\,\bphi\,,\\
\Q(t,z) & \coloneqq \bpm 0 & q(t,z)\\ -\conj{q(t,z)} &0\epm\,,\quad
\brho(t,z) \coloneqq \bpm D(t,z) & P(t,z)\\ \conj{P(t,z)} & -D(t,z) \epm\,,\quad
\sigma_3 \coloneqq \bpm 1 & 0 \\ 0 & -1 \epm\,,
\end{aligned}
\end{equation}
where $\bphi\coloneqq\bphi(\lambda;t,z)$ is the eigenfunction. Note that the scattering problem $\bphi_t = (\ii\lambda\sigma_3 + \Q)\bphi$ is the celebrated non-self-adjoint Zakharov-Shabat problem,
so MBEs~\eqref{e:mbe} can be considered as a member of negative flows in the AKNS hierarchy~\cite{akns1974},
due to the term proportional to $\lambda^{-1}$ in the second part of the Lax pair~\eqref{e:laxpair}.

It should be pointed out that,
assuming real solutions,
the MBEs~\eqref{e:mbe} with \myblue{the change of variables}
$P = \sin(\Theta)$,
$D = \cos(\Theta)$ and \myblue{$q = -\Theta_t/2$}
are related to the sine-Gordon equation
in characteristic coordinates
\myblue{(also known as light-cone coordinates)}
$\Theta_{tz} = 2\sin(\Theta)$ with $\Theta\in\Real$.
\myblue{
In some sense, MBEs are complex generalizations of the sine-Gordon equation.
Another two closely related systems of MBEs are the focusing nonlinear Schr\"{o}dinger (NLS) equation,
and the complex modified Korteweg-De Vries (mKdV) equation.
They are the second and the third positive flows in the same hierarchy,
and will be discussed later in this work.
}

\myblue{
It is well-known that all aforementioned systems have soliton solutions~\cite{zs1972,akns1974,akn1974}.
In particular, the MBEs solitons can be constructed on either stable or unstable backgrounds,
exhibiting consequential differences from a physical point of view.
}


Solitons are fascinating phenomena,
and have been extensively studied \myblue{for} years in different fields of nonlinear sciences.
The availability of the inverse scattering transform (IST)
and many other algebraic methods make it possible
to derive and study exact soliton solutions for completely integrable systems~\cite{zs1972}. In particular, the focus on one hand is studying the properties of single solitons, and on the other hand is analyzing the nonlinear interactions among solitons, including the proof of elastic collisions and the calculation of asymptotic phase shifts before and afterwards.

In order to discuss soliton interactions,
it is crucial to distinguish the cases whether solitons have different velocities or share the same \myblue{velocity}.
The latter case is referred to as a \textit{degenerate soliton group} (DSG) in this work, but may have other names in literature, such as a breather~\cite{akns1973}.
Of course,
some integrable systems by nature do not admit DSGs,
such as the Korteweg-De Vries (KdV) equation and the defocusing NLS equation~\cite{cj2016,dkpv2013,as1981,akns1974}.
On the other hand,
a large number of integrable systems do admit such complex structures,
for example,
the focusing NLS-type systems
(including continuous and discrete cases,
scalar and vector forms,
with zero or nonzero backgrounds)~\cite{zs1972,apt2004,m1974,lbs2017,lb2018,f2014},
the mKdV equation~\cite{zy2020,akns1974}, the sine-Gordon equation~\cite{akns1974},
and many more~\cite{xtaqgq2012}.
Therefore,
the soliton analysis on the latter set of systems
is more complex in general because of the existence of DSGs.


\myblue{
Assuming distinct soliton velocities inside multi-soliton solutions,
the common aftermath of collisions is
well-separated individual solitons with phase shifts~\cite{zs1972,apt2004,m1974,akns1974,km2023,gp2020,apbk2021}.
However,
it is unavoidable to include DSGs
when discussing the general $N$-soliton solutions for many integrable systems.
Because an $N$-DSG is
an $N$-soliton structure comprising of solitons sharing an identical velocity,
it forms a coherent localized particle-like wave train.
Contrary to the aforementioned non-degenerate multi-soliton solutions,
a DSG does not break apart into smaller components as time passes,
but is able to interact with other DSGs or solitons as a whole.
Due to the complexity,
in-depth analysis and even new tools become necessary,
in order to include DSGs in the discussion of soliton solutions in the most general settings.
}

We recall some existent results on DSGs for different integrable systems.
Though this list is by no means exhaustive.
The general $N$-soliton solutions of the focusing NLS equation with a ZBG
is systematically studied in~\cite{lbs2017},
where the center of a $N$-DSG is derived and soliton asymptotics are calculated.
However,
the authors use a determinant solution formula obtained from the operator formalism.
Thus,
the results cannot be easily generalized to other integrable systems,
because new solution {formul\ae} have to be derived using similar operator formalism.
The $N$-DSGs of the focusing NLS equation with a nonzero background (NZBG)
is considered in~\cite{lb2018},
but the long-time asymptotics for a general $N$-soliton solution is still missing
because of the complexity of the solution formula.
The degenerate breathers of the focusing NLS-type equation with NZBG
are also analyzed in a series of works
when discussing superregular breathers and modulational instability~\cite{gz2014,zwl2017,gazers2019,lca2024}.
\myblue{
Partial results on $N$-DSGs,
such as derivation without systematical analyses or studying with small values of $N$,
for various integrable systems are discussed in many early works,
such as the focusing NLS equation,
the sine-Gordon equation,
the mKdV equation and others~\cite{l1971,akns1973,akns1974,zs1972,hl2018,gtlzl2012,gtw2012,rmhz2023,ks1981}.
}

We should also mention other close topics of DSGs,
such as the high-order solitons and soliton gases.
High-order solitons correspond to eigenvalues of multiplicities more than one,
and can be considered as fusion of solitons~\cite{akns1974}.
There are a great number of results on this type of solutions,
for sine-Gordon equation~\cite{ccf2017},
the NLS equation~\cite{s2017},
Hirota equation~\cite{cf2019,zl2021}
and many more~\cite{lz2024,bm2022,blm2020}.
Other than merging eigenvalues,
one can also discuss the limit as $N\to+\infty$ in a $N$-soliton solution,
resulting in a phenomenon called soliton gas or integrable turbulence.
This class of solutions cannot be explicitly computed in general.
Though special cases are typically described by finite-gap solutions~\cite{z2009,ggjm2021,ggjmm2023,ekpz2011,ea2020,kmm2003}.


One of many ways to find exact multi-soliton solutions of MBEs~\eqref{e:mbe} is to formulate the corresponding IST,
which was first done in~\cite{akn1974},
but the authors only considered $D_- = -1$ case.
Ten years later,
the IST for the full Cauchy problem~\eqref{e:mbe} was finally established~\cite{gzm1983,gzm1984,gzm1985}.
However, in all works,
the inverse problems were formulated in terms of the \myblue{Gel'fand-Levitan-Marchenko} integral equations.
It is possible to obtain the soliton solutions from these works,
but the form of the inverse problem are not useful for the asymptotic calculations.
A RHP formulation is necessary for the Deift-Zhou's nonlinear steepest descent approach.
\myblue{One could reformulate the IST for the Cauchy problem~\eqref{e:mbe} using a RHP as the inverse problem.}
However,
as recently shown~\cite{lm2024} that,
the Cauchy problem for a general input light pulse $q_0(t)$ may be ill-posed,
and one of the often used assumption that $q(\cdot,z)\in L_1(\Real)$ may be violated.
Thus,
a considerable effort must be made in order to pursue this direction,
\myblue{such as reformulation of IST only for pure soliton solutions,}
which defeats the purpose of the current work on analyzing soliton solutions,
by making a huge detour.
By similar arguments,
algebraic methods,
such as the Darboux transformation,
to obtain exact soliton solutions are not feasible as well.
\myblue{
As such,
we adopt another approach similarly to the one used in~\cite{lm2024},
by starting with an appropriate exactly solvable RHP,
which is then proved to produce exact $N$-soliton solutions to MBEs~\eqref{e:mbe}.
This simple approach has the benefit providing a RHP that is readily available for the nonlinear steepest descent.
}


\section{Main results}
\label{s:results}

\myblue{
We first lay the groundwork and state necessary assumptions,
then we introduce results of various types of soliton solutions.
Finally, results for MBEs are generalized to the focusing NLS equation and the complex mKdV equation.
}

\subsection{Assumptions and notations}
\label{s:assumptions-notations}

Recall that IST works in three major steps:
\begin{enumerate}
\item
\textbf{Direct problem:}
One analyzes the scattering problem with the given initial condition $q_0(t)$ and obtains the scattering data at $z = 0$, including the reflection coefficient $r(\lambda,z=0)$, discrete eigenvalues and norming constants.
\item
\textbf{Propagation}\footnote{This step usually is called ``evolution". Because of the spatial-temporal-variable switch in optical systems, the ``evolution" variable becomes the spatial variable $z$, and we rename this step as ``propagation".}\textbf{:}
One uses the second half of the Lax pair~\eqref{e:laxpair} to calculate the $z$ dependence of the scattering data.
\item
\textbf{Inverse problem:}
One reconstructs the eigenfunction $\bphi(\lambda;t,z)$, by formulating and solving either a \myblue{Gel'fand-Levitan-Marchenko} integral equation or a RHP, and subsequently obtains the solutions of MBEs~\eqref{e:mbe}.
\end{enumerate}

\myblue{
Recall that
the past ISTs~\cite{akn1974,gzm1985} are ill-suited in this work,
because their inverse problems are formulated in terms of Gel'fand-Levitan-Marchenko integral equations.
}
Luckily,
the direct problem only involves the initial data $q_0(t)$ and the scattering problem, the latter of which is shared among the whole AKNS hierarchy. Thus, one could borrow the results from the general theory for the AKNS hierarchy and other well-studied integrable systems~\cite{akns1974,z1989,z1998,bc1984,bc1985}:
\begin{enumerate}
\item \textbf{Existence of scattering data:}
If $q_0(t)$ is sufficiently smooth and decays to zero fast enough as $t\to\pm\infty$,
for example,
if $q_0(t)$ is in the Schwartz class,
then the reflection coefficient $r(\lambda;0)$ and $N$ pairs of discrete eigenvalues
$\{\lambda_j,\conj{\lambda_j}\}_{j=1}^N\subset\Complex\setminus\Real$ exist,
where $\conj{\lambda_j}$ is the complex conjugate of $\lambda_j$.
For each eigenvalue pair,
there are associated norming constants with proper symmetries.
\item \textbf{Reflectionless solutions:}
Pure soliton solutions require that the reflection coefficient $r(\lambda;z)$ vanishes for all $z\ge0$.
\item \textbf{Anomalies:}
\begin{enumerate}
\item
Eigenvalues could be of higher orders.
\item
The number of eigenvalue pairs $N$ could be infinity.
\item
\myblue{
Eigenvalues may appear in the continuous spectrum $\Real$.
}
\end{enumerate}
\end{enumerate}
Because we are studying $N$-soliton solutions,
it is reasonable to assume that the first two results are true.
For anomaly (a),
most of the work addresses simple eigenvalues,
but in Section~\ref{s:N-order-soliton} we show that every high-order eigenvalue can be obtained by fusion of $N$ simple eigenvalues.
Correspondingly,
every high-order soliton can be obtained by merging $N$ simple solitons.
Then,
we present an explicit formula for such high-order solutions.
Similarly to anomaly (b),
we mostly discuss the case with $1\le N<+\infty$,
but in Section~\ref{s:soliton-gas} we show how to derive a RHP for soliton gases by taking limits $N\to+\infty$.
We do not discuss the phenomena related to anomalies (c),
\myblue{usually called embedded eigenvalues,}
because they are out of the scope of this work.
\begin{remark}[On notations]
\label{rmk:notations}
We use boldface letters to denote matrices and vectors, with the exception of the identity matrix $\I$ and the Pauli matrix $\sigma_3$ defined in Equation~\eqref{e:laxpair}.
The imaginary unit is denoted $\ii$.
Complex conjugation and conjugate transpose are indicated with a bar $\conj{\lambda}$ and a dagger $\A^\dagger$,
respectively.
\myblue{
The Schwarz reflection for a scalar function is define by $f^*(\lambda) \coloneqq \conj{f(\conj\lambda)}$,
and for a matrix function is denoted by
$\A(\conj{\lambda})^\dagger$.
Furthermore,
we use the shorthand notation
$\A(\conj\lambda)^{-\dagger} = [\A(\conj\lambda)^\dagger]^{-1}$ to denote the combination of Schwarz reflection and inversion of a matrix.
}
Let the notation $D_a^\epsilon \coloneqq \{\lambda\in\Complex: |\lambda - a| < \epsilon\}$ being the open disk centered at $\lambda = a$ with radius $\epsilon > 0$ in the complex plane, so $\partial D_a^\epsilon\coloneqq \{\lambda\in\Complex:|\lambda - a| = \epsilon\}$ is a circle.
\end{remark}

\subsection{Spectra and Riemann-Hilbert problems}
\label{s:intro-spectra-rhp}

\myblue{
This work deals with different types of soliton solutions.
With the assumptions in the previous section in mind,
we can properly introduce the spectra of multi-soliton solutions,
characterized by the following definition.
}

\begin{definition}[Eigenvalues and norming constants]
\label{def:LambdaOmega}
Let $J\ge 1$ and $N_j\ge1$ with $j = 1,2,\dots,J,$ be given integers. Suppose $0< r_1 < r_2 < \dots <r_J$ are positive real numbers.
We define two sets
\begin{equation}
\begin{aligned}
\Lambda
 & \coloneqq \bigcup_{j = 1}^J\Lambda_j\,,\qquad
 \Lambda_j \coloneqq \Set{\lambda_{j,k}}{$\lambda_{j,k} = r_j\ee^{\ii \alpha_{j,k}}$  with $0 < \alpha_{j,1} < \dots < \alpha_{j,N_j} < \pi$}\,,\\
\Omega
 & \coloneqq \bigcup_{j = 1}^J\Omega_j\,,\qquad
  \Omega_j \coloneqq  \Set{\omega_{j,k}}{$\omega_{j,k}\in\Complex\setminus\{0\}$ with $k = 1,2,\dots, N_j$}\,.
\end{aligned}
\end{equation}
We call $\Lambda$ \textit{the set of eigenvalues (in the upper half plane)}, $\Omega$ \textit{the set of norming constants}. The total number of eigenvalues is simply the cardinality of set $\Lambda$, denoted by $N \coloneqq |\Lambda| = N_1 + \dots + N_J$.
\end{definition}
Before moving on, we present a few remarks about Definition~\ref{def:LambdaOmega}:
\begin{itemize}
\item
From now on, the phrase ``eigenvalues" simply means ``discrete eigenvalues", because: (i) we are only discussing soliton solutions, rising from discrete spectra, and (ii) we assume reflectionless solutions, so that the continuous spectrum $\Real$ does not play any role.
\item
Eigenvalues from the non-self-adjoint Zakharov-Shabat problem come in conjugate pairs
$\{\lambda_{j,k},\conj{\lambda_{j,k}}\}$,
but for simplicity we call the ones in the upper half plane eigenvalues,
or $\Lambda$ the eigenvalue set.
The conjugates are implied.
\item
All eigenvalues from Definition~\ref{def:LambdaOmega} are categorized into $J$ groups $\Lambda_j$, each of which contains eigenvalues with identical moduli $r_j > 0$. The reason for the definition of each group $\Lambda_j$ is discussed in detail after Theorem~\ref{thm:N-DSG}.
\end{itemize}
\begin{figure}[tp]
\centering
\includegraphics[width=0.35\textwidth]{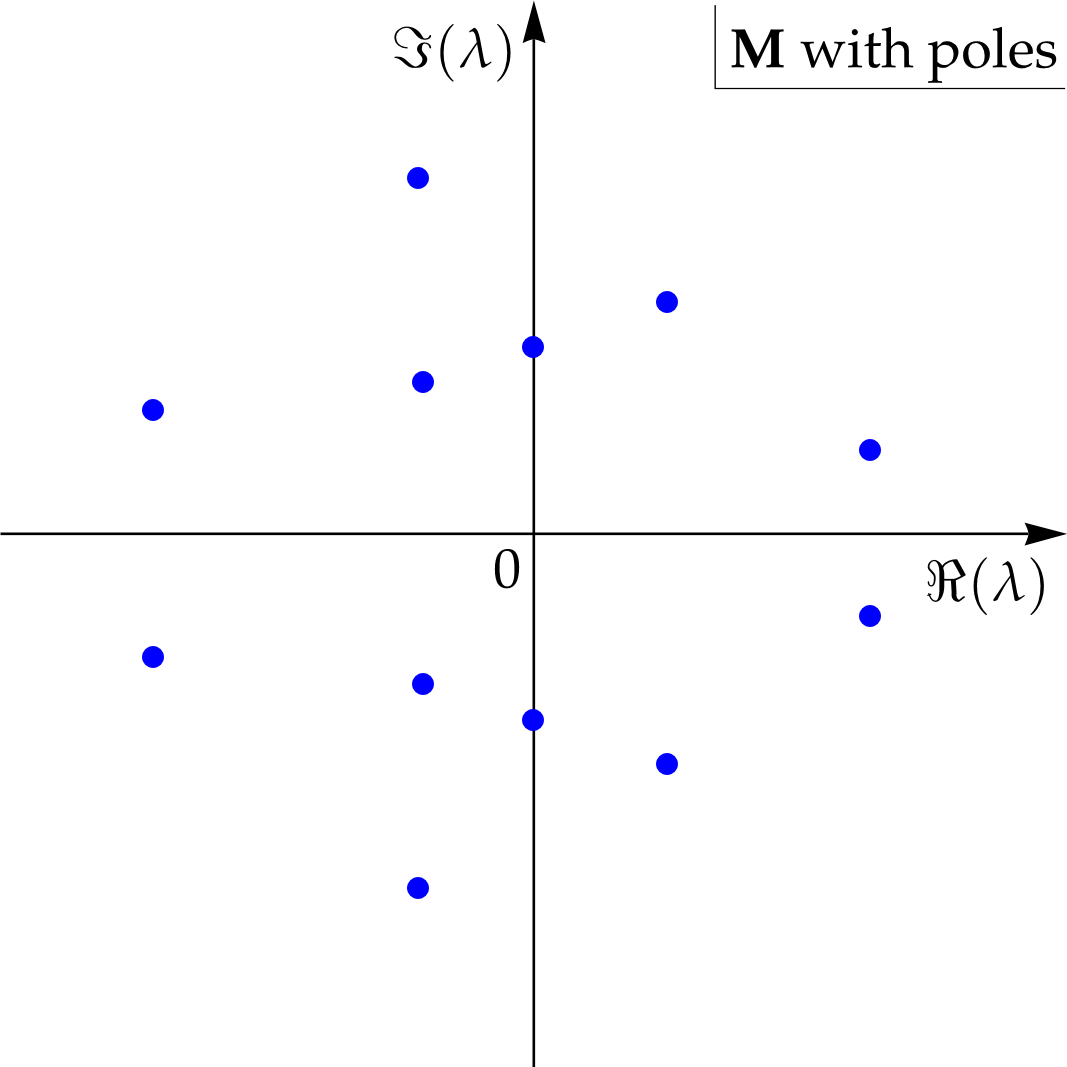}\qquad\qquad
\includegraphics[width=0.35\textwidth]{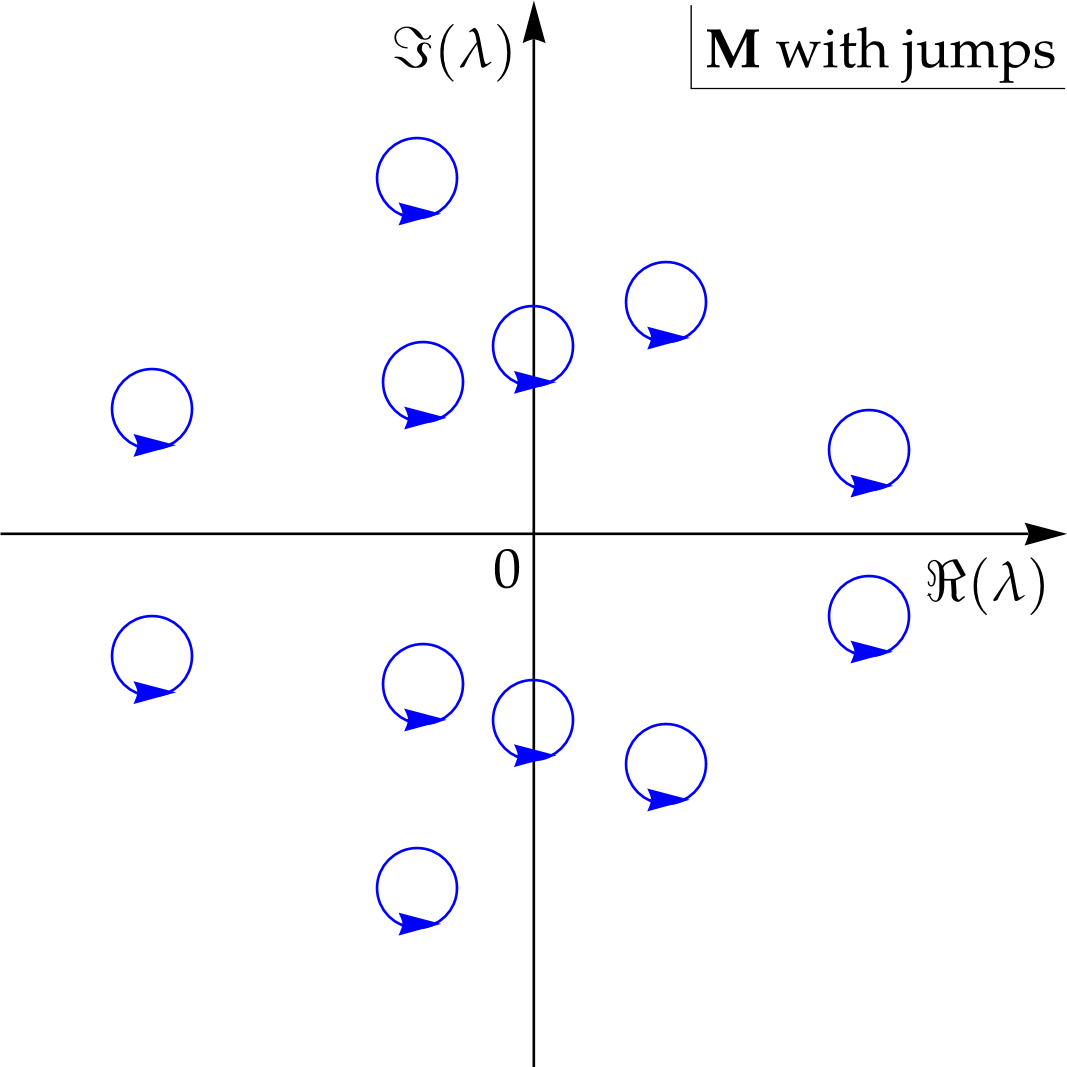}
\caption{
Left: an illustration of the pole distribution in RHP~\ref{rhp:N-soliton-residue-form}.
Right: an illustration of the jump configuration in RHP~\ref{rhp:N-soliton-jump-form} which is equivalent to the left plot.
}
\label{f:M}
\end{figure}
We first present RHPs composed of given eigenvalues and norming constants. The RHPs also satisfy necessary symmetries from the non-self-adjoint Zakharov-Shabat scattering. The relation between the RHPs and the solutions to MBEs~\eqref{e:mbe} will be discussed in Lemma~\ref{thm:reconstruction}.
\begin{rhp}[Residue form]
\label{rhp:N-soliton-residue-form}
Given sets $\Lambda$ and $\Omega$ from Definition~\ref{def:LambdaOmega}.
Seek a matrix meromorphic function $\lambda\mapsto\M(\lambda;t,z)$ on $\Complex$,
such that it has asymptotics $\M(\lambda;t,z)\to\I$ as $\lambda\to\infty$,
\myblue{is analytic for $\lambda\in\Complex\setminus\Lambda$,}
and satisfies the following residue conditions
\begin{equation}
\begin{aligned}
\Res_{\lambda = \lambda_{j,k}}\M(\lambda;t,z)
 &= \lim_{\lambda\to\lambda_{j,k}}\M(\lambda;t,z)\bpm 0 & 0 \\ \omega_{j,k}\ee^{-2\ii\theta(\lambda;t,z)} & 0 \epm\,,\\
\Res_{\lambda = \conj{\lambda_{j,k}}}\M(\lambda;t,z)
 & = \lim_{\lambda\to\conj{\lambda_{j,k}}}\M(\lambda;t,z)\bpm 0 & - \conj{\omega_{j,k}}\ee^{2\ii\theta(\conj{\lambda};t,z)t} \\ 0 & 0 \epm\,,
\end{aligned}
\end{equation}
for all $\lambda_{j,k}\in\Lambda$ and $\omega_{j,k}\in\Omega$,
and
\begin{equation}
\label{e:theta-def}
\theta(\lambda;t,z)\coloneqq \lambda t - \frac{D_-}{2\lambda}z\,.
\end{equation}
\end{rhp}
The quantity $\theta(\lambda;t,z)$ is the (linear) dispersion relation of the MBEs~\eqref{e:mbe} deduced from the Lax pair~\eqref{e:laxpair}.
An illustration of the pole distribution of RHP~\ref{rhp:N-soliton-residue-form} is shown in Figure~\ref{f:M}(left). The existence and uniqueness of solutions of RHP~\ref{rhp:N-soliton-residue-form} will be addressed in Lemma~\ref{thm:reconstruction} later.
Clearly, by \myblue{Liouville's theorem} the solution $\M(\lambda;t,z)$ can be formally written as
\begin{equation}
\M(\lambda;t,z)
 = \I + \sum_{\lambda_{j,k}\in\Lambda}\bigg(\frac{\Res_{\lambda = \lambda_{j,k}}\M(\lambda;t,z)}{\lambda-\lambda_{j,k}} + \frac{\Res_{\lambda = \conj{\lambda_{j,k}}}\M(\lambda;t,z)}{\lambda-\conj{\lambda_{j,k}}}\bigg)\,.
\end{equation}
Note that this is not an explicit formula,
but instead is a linear system,
which is solved in Section~\ref{s:derivation-Nsoliton}.
It is also worth pointing out that
the residue conditions in RHP~\ref{rhp:N-soliton-residue-form} are not \myblue{well-suited}
when applying Deift-Zhou's nonlinear steepest descent method
in order to calculate asymptotics,
which is the main part of this work.
Thus,
it is necessary to convert the residue conditions to jumps.

Recall the notation $D_a^\epsilon$ from Remark~\ref{rmk:notations}.
Obviously,
there always exists a constant $0<\epsilon \ll1$
such that all \myblue{closed} disks
$\{\conj{D_{\lambda_{j,k}}^\epsilon},\, \conj{D_{\conj{\lambda_{j,k}}}^\epsilon}\}_{\lambda_{j,k}\in\Lambda}$
do not intersect for a given set $\Lambda$,
and they do not cross the real line.
Without repeating this argument in the rest of this work,
we always \myblue{assume that} it is true and such $\epsilon$ is used.
\begin{rhp}[Jump form]
\label{rhp:N-soliton-jump-form}
\myblue{
Given sets $\Lambda$ and $\Omega$.
Seek an analytic matrix function $\lambda\mapsto\M(\lambda;t,z)$
on $\Complex\setminus \partial D^\epsilon$,
where $D^\epsilon \coloneqq \bigcup_{\lambda_{j,k}\in\Lambda}\big(D_{\lambda_{j,k}}^\epsilon\bigcup D_{\conj{\lambda_{j,k}}}^\epsilon\big)$,
and that $\M(\lambda; t, z)$ has continuous boundary values on $\partial D^\epsilon$.
such that it has asymptotics $\M(\lambda)\to\I$ as $\lambda\to\infty$,
and satisfies the jumps
}
\begin{equation}
\begin{aligned}
\M^+(\lambda;t,z) & = \M^-(\lambda;t,z)\V_{j,k}(\lambda;t,z)^{-1}\,,\qquad &\lambda\in \partial D_{\lambda_{j,k}}^\epsilon\,,\\
\M^+(\lambda;t,z) & = \M^-(\lambda;t,z)\V_{j,k}(\conj{\lambda};t,z)^\dagger\,,\qquad &\lambda\in \partial D_{\conj{\lambda_{j,k}}}^\epsilon\,,
\end{aligned}
\end{equation}
for all $\lambda_{j,k}\in\Lambda$ and $\omega_{j,k}\in\Omega$,
all circles $\{\partial D_{\lambda_{j,k}}^\epsilon,\, \partial D_{\conj{\lambda_{j,k}}}^\epsilon\}_{\lambda_{j,k}\in\Lambda}$ are oriented counterclockwise and disjoint with a proper $0<\epsilon\ll1$, and the jump matrices are defined by
\begin{equation}
\everymath{\displaystyle}
\V_{j,k}(\lambda;t,z)
    \coloneqq \bpm
        1 & 0 \\ \frac{\omega_{j,k}\ee^{-2\ii\theta(\lambda;t,z)}}{\lambda-\lambda_{j,k}} & 1
    \epm\,,\qquad
\V_{j,k}(\conj\lambda;t,z)^\dagger
    = \bpm
        1 & \frac{\conj{\omega_{j,k}}\ee^{2\ii\theta(\lambda;t,z)}}{\lambda - \conj{\lambda_{j,k}}} \\ 0 & 1
    \epm\,.
\end{equation}
\end{rhp}
An illustration of the jump configuration of RHP~\ref{rhp:N-soliton-jump-form} is shown in Figure~\ref{f:M}(right).
\begin{remark}
\label{thm:M-symmetry}
Clearly,
all jump matrices in RHP~\ref{rhp:N-soliton-jump-form} satisfy the Schwarz symmetry
$\V(\lambda) = \V(\conj\lambda)^\dagger$ for all $t$ and $z$,
\myblue{after the jump contours in the upper-half plane are oriented clockwise,}
so by Zhou's lemma the solution to RHP~\ref{rhp:N-soliton-jump-form} uniquely exists~\cite{z1989}.
Moreover, it is easy to verify that the solution $\M(\lambda;t,z)$ has the symmetry
$\M(\lambda;t,z)^{-1} = \M(\conj\lambda;t,z)^\dagger$
\myblue{and the property $\det\M(\lambda;t,z) = 1$, which imply that}
\begin{equation}
\M_{2,2}(\lambda;t,z) = \M_{1,1}^*(\lambda;t,z)\,,\qquad
\M_{2,1}(\lambda;t,z) = - \M_{1,2}^*(\lambda;t,z)\,.
\end{equation}
Hence, one only needs to solve for the first row in order to reconstruct the whole matrix $\M(\lambda;t,z)$.
\end{remark}
\begin{lemma}[Reconstruction formula]
\label{thm:reconstruction}
For given sets $\Lambda$ and $\Omega$,
and a given number $D_- = \pm1$,
there is a unique solution \myblue{$\{q(t,z),D(t,z),P(t,z)\}$}
to the MBEs~\eqref{e:mbe} with boundary condition $D \to D_-$ as $t\to-\infty$,
which can be reconstructed from the solutions to the RHP~\ref{rhp:N-soliton-residue-form}
or RHP~\ref{rhp:N-soliton-jump-form} via the following formula
\begin{equation}
\label{e:reconstructon-matrix}
\Q(t,z;\Lambda,\Omega)
 = \lim_{\lambda\to\infty}\ii\lambda[\M(\lambda;t,z),\sigma_3]\,,\quad
\brho(t,z;\Lambda,\Omega) = D_-\M(0;t,z)\sigma_3\M(0;t,z)^{-1}\,.
\end{equation}
Recall $N = |\Lambda|$.
The solution is call the $N$-soliton solution of MBEs.
\end{lemma}
Lemma~\ref{thm:reconstruction} can be proved by applying the dressing method to
RHP~\ref{rhp:N-soliton-jump-form}, similarly to the one in~\cite{lm2024}.
The proof is shown in Section~\ref{s:proof-reconstruction-formula} for completeness of this work.
\begin{remark}
In this work, the phrase \textit{equivalent forms} or \textit{equivalent RHPs} indicates that multiple RHPs yield identical MBEs solutions. Hence, RHPs~\ref{rhp:N-soliton-jump-form} and~\ref{rhp:N-soliton-residue-form} are equivalent by Lemma~\ref{thm:reconstruction}. Usually there are two equivalent RHPs for soliton solutions. One contains residue conditions, as usually done in ISTs, and the other one contains jumps. The former one is useful when calculating the solution formula, whereas the latter one is useful when analyzing properties or applying the nonlinear steepest descent method. Both forms are used extensively in this work.
\end{remark}
\begin{remark}
We append the spectral data $\Lambda$ and $\Omega$ as explicit parameter dependence in MBEs solutions as $(t,z;\Lambda,\Omega)$ in Lemma~\ref{thm:reconstruction} or later, in order to explicitly show the spectral dependence. This will be helpful when later discussing soliton asymptotics, where $\Lambda$ or $\Omega$ may change.
\end{remark}
%

\subsection{General soliton solution {formul\ae}}
\label{s:intro-solutoin-formula}

\myblue{
We are now ready to derive general $N$-soliton solutions,
which are the foundation of later results.
With the spectrum set $\Lambda$ characterized by the most general parameters
$J \ge 1$ and $N_1,\dots, N_J \ge 1$,
the general $N$-soliton solution is computed in Section~\ref{s:derivation-Nsoliton},
via the reconstruction formula in Lemma~\ref{thm:reconstruction}
from RHP~\ref{rhp:N-soliton-residue-form}.
}

\begin{theorem}[General $N$-soliton solution formula]
\label{thm:N-soliton-formula}
For given sets $\Lambda$ and $\Omega$ from Definition~\ref{def:LambdaOmega},
and $N = |\Lambda|$.
\myblue{The} corresponding $N$-soliton solution to MBEs is given by
\begin{equation}
\label{e:Nsoliton-formula}
\begin{aligned}
q(t,z;\Lambda,\Omega)
 & = 2\ii - 2\ii\frac{\det(\I - \myblue{\vbone\conj{\C_\infty}} + \bGamma\conj{\bGamma})}{\det(\I + \bGamma\conj{\bGamma})}\,,
\end{aligned}
\end{equation}
where $D(t,z;\Lambda,\Omega)$ and $P(t,z;\Lambda,\Omega)$ are reconstructed by Lemma~\ref{thm:reconstruction}, with entries of $\M(\lambda;t,z)$ given by
\begin{equation}
\begin{aligned}
M_{1,1}(\lambda;t,z)
 & = \frac{\det(\I - \myblue{\vbone\C\conj{\bGamma}} + \bGamma\conj{\bGamma})}{\det(\I + \bGamma\conj{\bGamma})}\,, \quad
M_{1,2}(\lambda;t,z)
 = \frac{\det(\I - \myblue{\vbone\conj{\C}} + \bGamma\conj{\bGamma})}{\det(\I + \bGamma\conj{\bGamma})} - 1\,,
\end{aligned}
\end{equation}
where all quantities are defined by
\begin{equation}
\label{e:C-Gamma-definition}
\begin{aligned}
\C(\lambda)
 & \coloneqq \bpm c_{1,1}(\lambda) & \dots & c_{1,N_1}(\lambda) & c_{2,1}(\lambda) & \dots & c_{2,N_2}(\lambda) & \dots & c_{J,1}(\lambda) & \dots & c_{J,N_J}(\lambda) \epm\,,\\
\C_\infty
 & = \lim_{\lambda\to\infty}\lambda \C(\lambda)
 = \bpm \omega_{1,1}\ee^{-2\ii\theta_{1,1}} & \omega_{1,2}\ee^{-2\ii\theta_{1,2}} & \dots & \omega_{J,N_J}\ee^{-2\ii\theta_{J,N_J}} \epm\,,\\
\vbone & \coloneqq \bpm 1 & \dots & 1\epm_{1\times N}^\top\,,\\
\bGamma
 & \coloneqq
  \bpm
   \C(\conj{\lambda_{1,1}})^\top & \dots & \C(\conj{\lambda_{1,N_1}})^\top &
   \dots & \C(\conj{\lambda_{J,1}})^\top & \dots & \C(\conj{\lambda_{J,N_J}})^\top
  \epm_{N\times N}^\top\,,\\
\theta_{j,k}
 & \coloneqq \theta(\lambda_{j,k};t,z)\,,\qquad
c_{j,k}(\lambda)\coloneqq \frac{\omega_{j,k}}{\lambda - \lambda_{j,k}}\ee^{-2\ii\theta_{j,k}}\,,\qquad
\lambda_{j,k}\in\Lambda\,,\qquad \omega_{j,k}\in\Omega\,.
\end{aligned}
\end{equation}
\myblue{
Moreover, $\det(\I+\bGamma\conj{\bGamma}) > 1$ for all $(t,z)\in\Real^2$,
so that the soliton solution is always regular.
}
\end{theorem}

\myblue{
The spectrum set $\Lambda$ and $\Omega$ are arbitrarily given in Theorem~\ref{thm:N-soliton-formula},
producing $N$ solitons.
Because we do not impose additional constraint,
}
some of the solitons may travel with the same speed, i.e.,
form a DSG\footnote{We will make these statements concrete in Theorems~\ref{thm:N-DSG} and~\ref{thm:soliton-asymptotics}.}. Hence, the setup of this work is capable of producing DSGs, and becomes a perfect tool for studying multi-soliton solutions containing such complex structures.

\myblue{
Before diving into the long-time asymptotics for general $N$-soliton solutions,
it is necessary to characterize the fundamental building blocks --- $N$-DSGs,
with a single soliton being a special case as a $1$-DSG.
}

\subsection{Degenerate soliton groups}
\label{s:intro-dsg}

\myblue{
In this section,
we consider a special case of the general $N$-soliton solutions from Theorem~\ref{thm:N-soliton-formula},
by requiring only one eigenvalue group $J = 1$ from Definition~\ref{def:LambdaOmega}.
We show that such solutions form DSGs and explore their properties.
Discussion on similarities between DSGs and single solitons
implies that an $N$-DSG is a generalization of a single soliton with $N\ge2$.
Recall that solitons are localized traveling waves
and maintain their velocities and shapes after nonlinear interactions~\cite{zk1965}.
}
The following theorem establishes the localization and traveling-wave nature of \myblue{a} DSG with an arbitrary size $N\ge1$.
An explicit formula for the DSG position is also provided.
\begin{theorem}[Degenerate $N$-soliton solution]
\label{thm:N-DSG}
Suppose $J = 1$. Let $\Lambda$ and $\Omega$ be given according to Definition~\ref{def:LambdaOmega}, so that $N = N_1 \ge 1$ be a positive integer and $\Lambda = \Lambda_1$. Namely, all eigenvalues have identical moduli $|\lambda_{1,k}| = r_1 > 0$ for $1\le k\le N_1$. Then, regardless of the initial state of the medium $D_- = \pm1$, Lemma~\ref{thm:reconstruction} produces an $N$-DSG of the MBEs~\eqref{e:mbe}, which has the following properties
\begin{enumerate}
\item
All $N$ solitons in the solution travel coherently with an identical speed $V \coloneqq -2D_- r_1^2$.
\item
The $N$-DSG is localized along the direction $z = Vt$.
Along other directions $z = \xi t$ with $\xi\ne V$
\myblue{and a positive constant $\aleph$,
as $t\to\pm\infty$ the solution is}
\begin{equation}
\myblue{
q(t,z;\Lambda,\Omega) = \O(\ee^{-\aleph |t|})\,,\qquad
D(t,z;\Lambda,\Omega) = D_- + \O(\ee^{-\aleph |t|})\,,\qquad
P(t,z;\Lambda,\Omega) = \O(\ee^{-\aleph |t|})\,.
}
\end{equation}
\item
The center of the $N$-DSG denoted by $z_\c(t)$ is given as
\begin{equation}
\label{e:N-soliton-solution:zc-zd}
z_\c(t)
 \coloneqq Vt + z_\d\,,\qquad
z_\d
 \coloneqq -\frac{D_- r_1}{\sum_{k=1}^{N_1}\sin(\alpha_{1,k})}\ln\bigg(\prod_{k=1}^{N_1} |\omega_{1,k}|\cdot \frac{\prod_{k = 2}^{N_1}\prod_{l = 1}^{k-1}|\lambda_{1,k} - \lambda_{1,l}|^2}{\prod_{k=1}^{N_1}\prod_{l=1}^{N_1}|\conj{\lambda_{1,k}} - \lambda_{1,l}|}\bigg)\,,
\end{equation}
where $z_\d$ is the displacement.
\item With $J = 1$, the solution $\M(\lambda;t,z)$ to RHPs~\ref{rhp:N-soliton-residue-form} and~\ref{rhp:N-soliton-jump-form} are bounded as $t\to\pm\infty$ in the direction $z = V t$.
\end{enumerate}
\end{theorem}
Theorem~\ref{thm:N-DSG} is proved in Section~\ref{s:N-DSG}. The last statement (4) of Theorem~\ref{thm:N-DSG} is used later in the proof of Theorem~\ref{thm:soliton-asymptotics}, in Section~\ref{s:N-soliton-asymptotics-stable-tpos-case4}.
\begin{remark}
The $N$-DSG's velocity $V$ can be derived from the equation $\Im(\theta(\lambda_{1,k};t,V t)) = 0$ and is solely determined by $r_1$. Conversely, we use the equation $\Im(\theta(\lambda_{j,k};t,V t) = const$ to define the eigenvalue groups $\Lambda_j$ in Definition~\ref{def:LambdaOmega}. The velocity is positive in a stable medium ($D_- = -1)$ and is negative in an unstable medium ($D_- = 1$). Because the light cone is the first quadrant of the $(t,z)$ plane, the $N$-DSG travels subluminally in a stable medium, and superluminally in an unstable medium. Consequently, the stable $N$-DSG stays inside the light cone as $t\to+\infty$, whereas the unstable $N$-DSG eventually travels out. Clearly, the unstable case seems unphysical, and deserves further analysis.
In fact,
solitons in unstable \myblue{media} with non-vanishing reflection coefficient \myblue{have been} studied before and shown to be related to superfluorescence~\cite{gzm1983,gzm1984,gzm1985,z1980}.
\end{remark}
Due to the complexity of the solution formula from Theorem~\ref{thm:N-soliton-formula}, it is hard to discuss the shape of $N$-DSGs, and how the shape depends on the eigenvalues and norming constants. Therefore, it is necessary to analyze special cases when $N$ is small. The simplest DSG is of course a single soliton corresponding to $N = 1$, and the breathers correspond to $N = 2$. Let us first introduce a shorthand notation to simplify solution {formul\ae} before we investigate these special cases.
\begin{equation}
\label{e:spm-def}
s_{j,\pm}(t,z) \coloneqq \pm 2r_j t + \frac{D_-}{r_j}z\,.
\end{equation}
One first discusses what happens of Theorem~\ref{thm:N-DSG} when $N = 1$, namely, a single soliton.
\begin{corollary}[One-soliton solution]
\label{thm:1-soliton}
Suppose that $N = J = N_1 = 1$ in Theorem~\ref{thm:N-soliton-formula}, and sets $\Lambda = \{\lambda_{1,1}\}$ and $\Omega = \{\omega_{1,1}\}$ are given. Using the parameterization for eigenvalue and norming constants in Definition~\ref{def:LambdaOmega}. Theorem~\ref{thm:N-soliton-formula} produces a one-soliton solution, which can be written in a compact form
\begin{equation}
\begin{aligned}
q(t,z;\lambda_{1,1},\omega_{1,1})
 & = 2 \ii r_1 \sin(\alpha_{1,1})\ee^{-\ii \tau_{1,1}(t,z)}\sech(\chi_{1,1}(t,z))\,,\\
D(t,z;\lambda_{1,1},\omega_{1,1})
 & = D_- \big[\cos(2\alpha_{1,1}) \sech^2(\chi_{1,1}(t,z)) + \tanh^2(\chi_{1,1}(t,z))\big]\,,\\
P(t,z;\lambda_{1,1},\omega_{1,1})
 & = -2 D_- \ee^{-\ii \tau_{1,1}(t,z)} \sin(\alpha_{1,1}) \sech^2(\chi_{1,1}(t,z)) \cosh (\chi_{1,1}(t,z) - \ii\alpha_{1,1})\,,
\end{aligned}
\end{equation}
with
\begin{equation}
\label{e:1soliton-chis-def}
\begin{aligned}
\chi_{1,1}(t,z) & \coloneqq \sin(\alpha_{1,1})s_{1,+}(t,z) + \xi_{1,1}\,,\qquad
\tau_{1,1}(t,z) \coloneqq \cos(\alpha_{1,1})s_{1,-}(t,z) + \phi_{1,1}\,,\\
\omega_{1,1} & = 2 r_1\sin(\alpha_{1,1})\ee^{\xi_{1,1} + \ii\phi_{1,1}}\,,\qquad
\xi_{1,1}\in\Real\,,\qquad
\phi_{1,1}\in[0,2\pi)\,.
\end{aligned}
\end{equation}
The soliton velocity is $V = -2D_-r_1^2$, the soliton amplitude is $2r_1\sin(\alpha_{1,1})$, and the soliton center is $z_\c(t) = Vt + r_1\xi_{1,1}/\sin(\alpha_{1,1})$.
\end{corollary}
Clearly, the modulus parameter $r_1$ of the eigenvalue $\lambda_{1,1}$ solely determines the soliton velocity $V$, and together with the phase $\alpha_{1,1}$ determines the soliton amplitude. The phase parameter $\alpha_{1,1}$ also dominates the complex oscillation of the soliton, via the functions $\tau_{1,1}(t,z)$ in Equation~\eqref{e:1soliton-chis-def}. In particular, the oscillation frequency is proportional to $\cos(\alpha_{1,1})$,
so smaller value of $\alpha_{1,1}$ means \myblue{faster} oscillations, as shown in Figure~\ref{f:1soliton}. Of course, the oscillatory behavior is hidden when considering the modulus $|q(t,z)|$ in the one-soliton solution, but it becomes more prominent in consideration of $N$-DSG with $N\ge2$.
\begin{remark}
As can be easily seen in Corollary~\ref{thm:1-soliton},
the one-soliton solution decays to the ZBG as $t\to\pm\infty$ with a fixed value of $z$. In particular, $D(t,z)\to D_-$ as $t\to\pm\infty$. Therefore, if the medium is initially in the stable state, it falls back to the stable state after a long time, as one expects physically. However, if the medium is initially in the unstable state, it also returns, which is \textit{not} physical. Contrary to the pure soliton solution discussed here, it was proved that the solution to MBEs without any solitons fell back to the stable state from the unstable one~\cite{lm2024}. Hence, it suggests that for an unstable medium, the nonlinear interactions between solitons and radiation play a crucial role, and deserve further analysis. Again, this case is related to superfluorescence.
\end{remark}
\begin{figure}[tp]
\centering
\begin{minipage}[b]{.68\textwidth}
\includegraphics[width=0.47\textwidth]{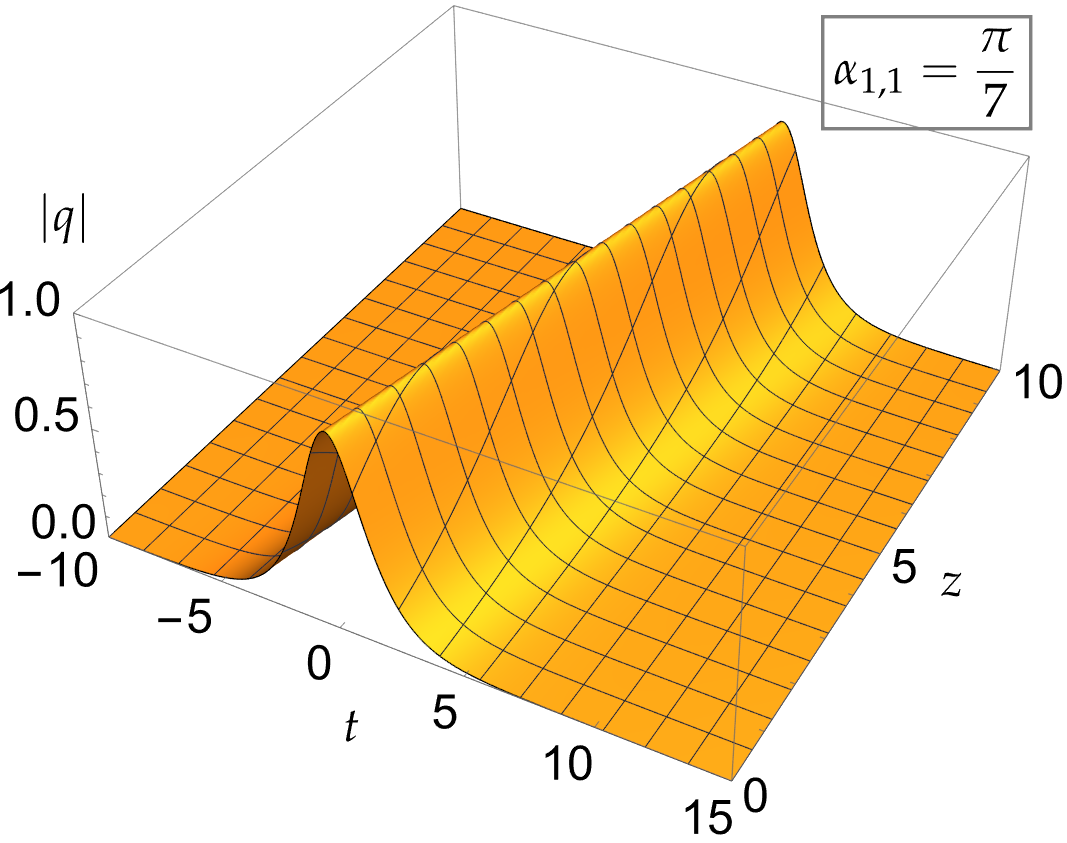}\quad
\includegraphics[width=0.47\textwidth]{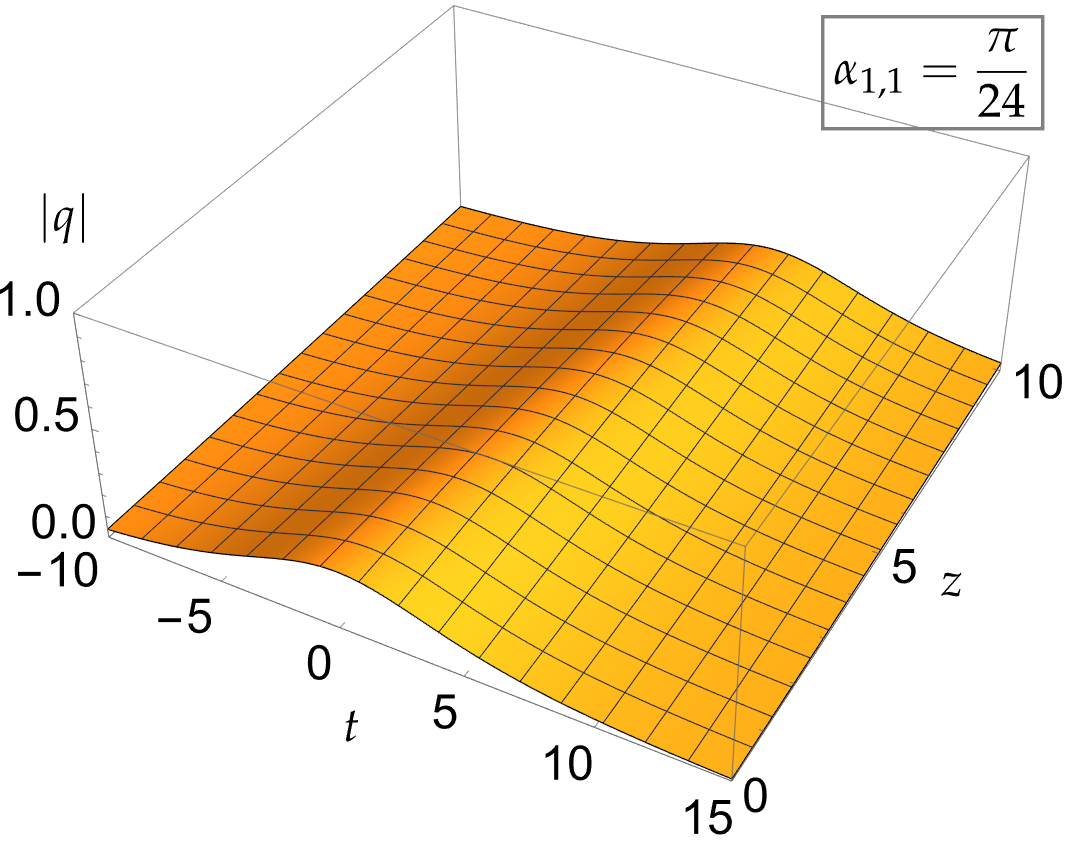}
\end{minipage}
\begin{minipage}[b]{.3\textwidth}
\includegraphics[width=1\textwidth]{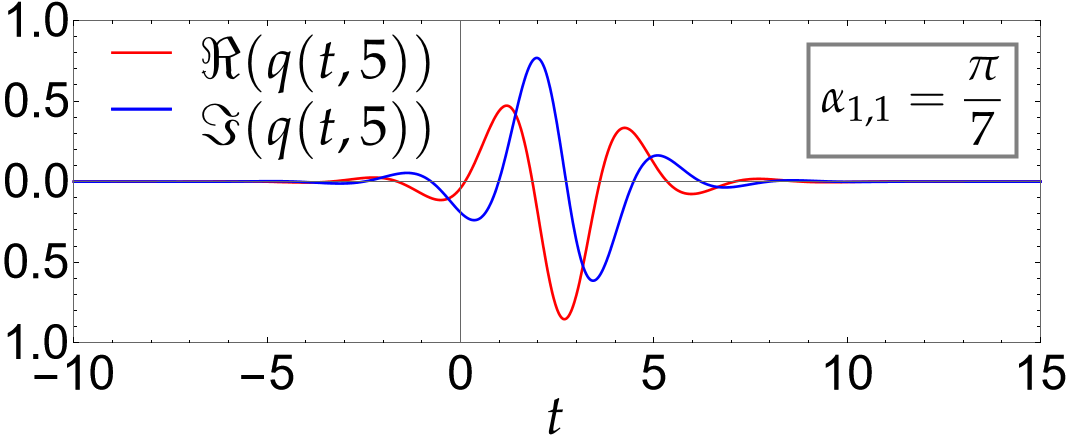}\\[0.5ex]
\includegraphics[width=1\textwidth]{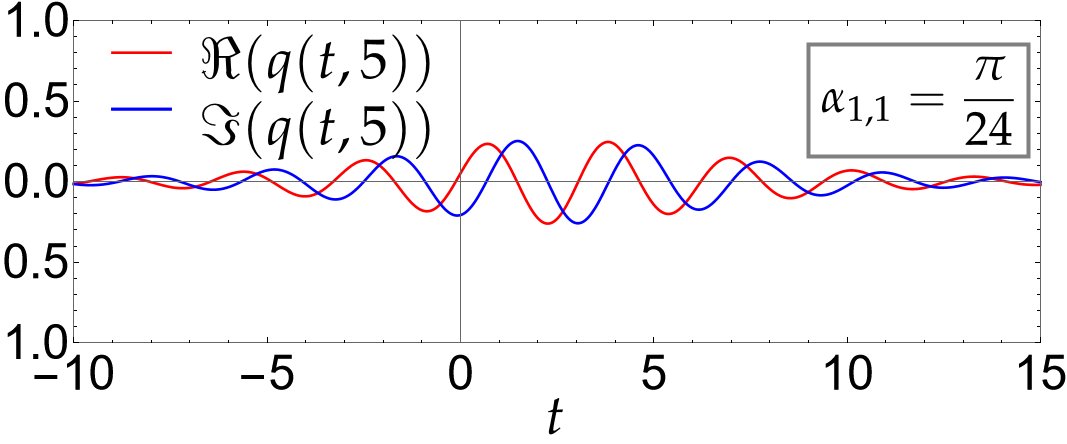}
\end{minipage}
\caption{Plots of exact $1$-soliton solutions in an initially stable medium with $r_1 = 1$, $\xi_{1,1} = 0$ and $\phi_{1,1} = \pi/2$, and $\alpha_{1,1} = \pi/7$ or $\alpha_{1,1} = \pi/24$ from Corollary~\ref{thm:1-soliton}. The two solutions have identical velocities, but different amplitudes determined by $\alpha_{1,1}$.}
\label{f:1soliton}
\end{figure}
Now,
we move on to analyzing a more complex yet interesting coherent structure than a single soliton.
However,
\myblue{the expression for a generic $2$-DSG (a bound state or a breather) is too complex to present.}
We therefore discuss a symmetric version,
by imposing $\lambda_{1,1} = -\conj{\lambda_{1,2}}$. This spectral setup was discussed in~\cite[Equation (5.9)]{akns1974}, but for the sine-Gordon equation.

\begin{corollary}[Symmetric $2$-DSG]
\label{thm:2-DSG-symmetric}
Suppose that $J = 1$ and $N = N_1 = 2$. Theorem~\ref{thm:N-DSG} reduces to a general $2$-DSG, which can be simplified further by taking  $\alpha_{1,1} = \alpha$ and $\alpha_{1,2} = \pi - \alpha$, $\omega_{1,1} = 2 r_1\sin(\alpha)\ee^{\xi + \ii\phi}$ and $\omega_{1,2} = \conj{\omega_{1,1}}$, implying that the eigenvalues are tied with an additional symmetry $\lambda_{1,2} = -\conj{\lambda_{1,1}}$.
\begin{equation}
\begin{aligned}
q(t,z;\Lambda,\Omega)
 & = 4\ii r_1 \tan(\alpha )\frac{\sin(\alpha) \sinh(L) \sin(\tau) + \cos(\alpha) \cosh(L) \cos(\tau)}{\cosh(L)^2 +  \tan(\alpha)^2\sin(\tau)^2}\,,\\
D(t,z;\Lambda,\Omega)
 & = D_-\frac{\cosh(L) ^4 - 6 \tan(\alpha )^2 \cosh(L) ^2 \sin(\tau) ^2 + \tan(\alpha ) ^4 \sin(\tau) ^4}{[\cosh(L)^2 + \tan(\alpha)^2 \sin (\tau)^2]^2}\,,\\
P(t,z;\Lambda,\Omega)
 & = 4\ii D_- \tan(\alpha) \cosh(L) \sin(\tau) \frac{\cosh(L)^2 -  \tan(\alpha)^2 \sin(\tau)^2}{[\cosh(L)^2 + \tan(\alpha)^2 \sin(\tau)^2]^2}\,,
\end{aligned}
\end{equation}
where the two real functions $L(t,z)$ and $\tau(t,z)$ are defined below
\begin{equation}
\begin{aligned}
L(t,z)
 & \coloneqq \sin(\alpha)s_{1,+}(t,z) + \xi + \ln(|\cos(\alpha)|)\,,\\
\tau(t,z)
 & \coloneqq \cos(\alpha)s_{1,-}(t,z) + \phi - \alpha\,.
\end{aligned}
\end{equation}
\end{corollary}
In the symmetric $2$-DSG in Corollary~\ref{thm:2-DSG-symmetric}, the optical pulse $q(t,z)$ and polarization $P(t,z)$ are purely imaginary, and $D(t,z)$ is of course still real. Note that this symmetric $2$-DSG is fundamentally different from the one of the sine-Gordon equation~\cite[Equation (5.9)]{akns1974}, the latter of which is a real function. The overall form of $2$-DSG is remarkably similar to the soliton solution with NZBG~\cite[Equation (17)]{lbkg2018}. This can be explained from the spectral point of view, because: (i) both cases are reflectionless; and (ii) the RHP for $2$-DSG with ZBG contains four poles (two pairs of eigenvalues), whereas the RHP for the $1$-soliton solution with NZBG contains a quartet of eigenvalues using a uniformization variable. Thus, the overall structure of RHPs are alike, so are the corresponding solutions. Two such examples are shown in Figure~\ref{f:2soliton-symmetric}. Clearly, the shape of DSGs depends on the phase parameter $\alpha_{j,k}$.
\begin{figure}[tp]
\centering
\begin{minipage}[b]{.68\textwidth}
\includegraphics[width=0.47\textwidth]{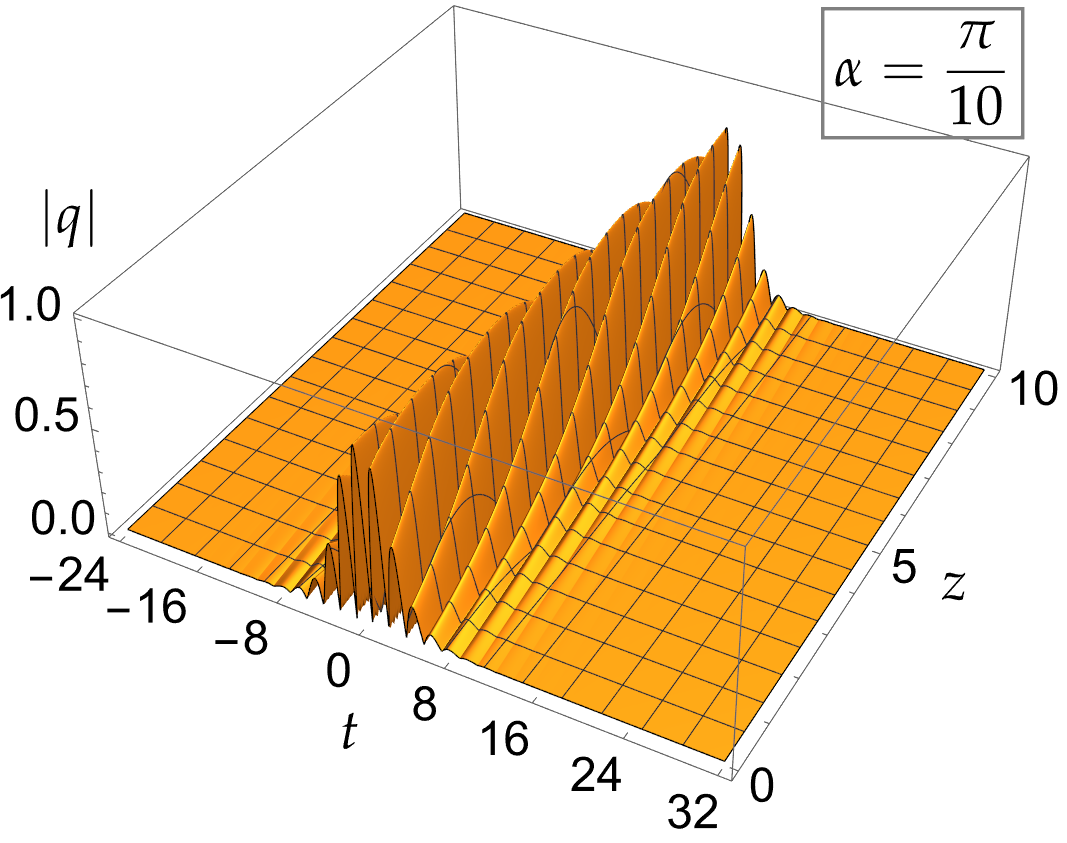}\quad
\includegraphics[width=0.47\textwidth]{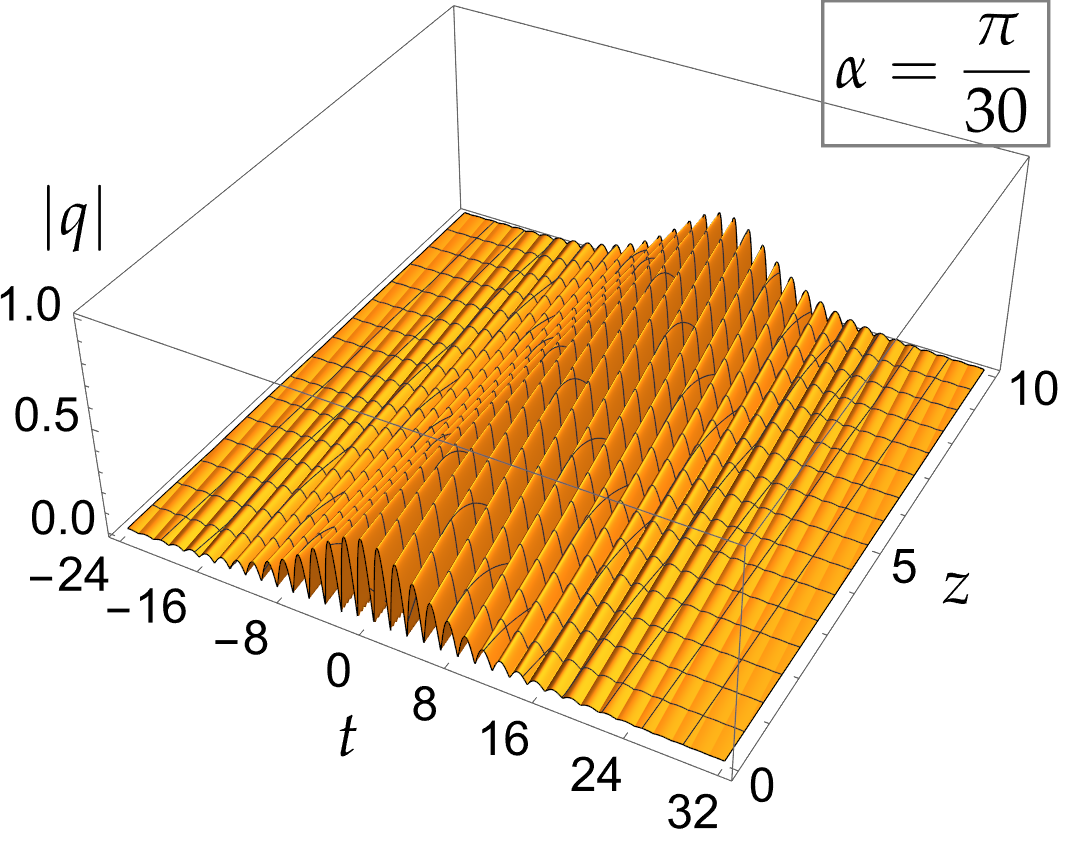}
\end{minipage}
\begin{minipage}[b]{.3\textwidth}
\includegraphics[width=1\textwidth]{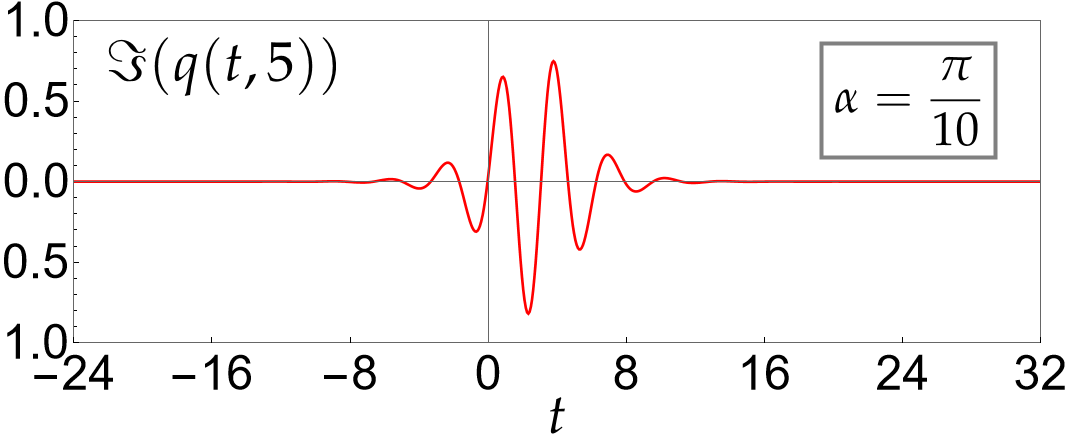}\\[0.5ex]
\includegraphics[width=1\textwidth]{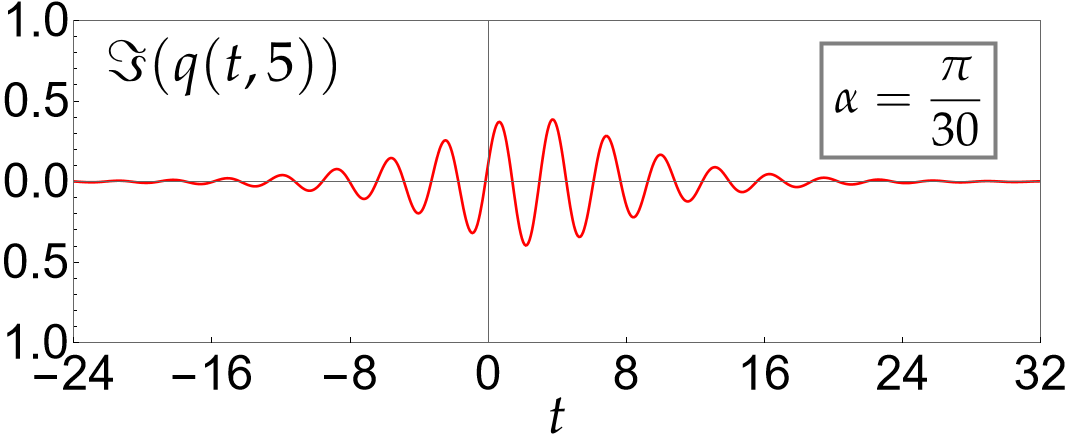}
\end{minipage}
\caption{Plots of exact symmetric $2$-DSG in an initially stable medium with $r_1 = 1$, $\xi_{1,1} = \phi_{1,1} = 0$, and $\alpha = \pi/10$ or $\alpha = \pi/30$ from Corollary~\ref{thm:2-DSG-symmetric}. Since $q(t,z)$ is purely imaginary, the right column only shows the imaginary parts. The two cases have identical velocity. The parameter $\alpha$ affects amplitudes and oscillatory behavior.}
\label{f:2soliton-symmetric}
\end{figure}

A more general case is shown in Figure~\ref{f:MBEs-DSG-soliton-asymptotics}(left), where the density plot of $|q(t,z)|$ of a $4$-DSG in an initially stable medium $D_- = -1$ is presented from Theorem~\ref{thm:N-soliton-formula} with parameters $\lambda_{1,k} = \ee^{\ii\pi/5},\ee^{\ii\pi/3},\ii\,,\ee^{2\ii\pi/3}$ and $\omega_{1,k} = 1$ for $1\le k\le 4$. On top of the exact solution, a red dashed line is drawn determined by the center $z_\c(t)$ from Theorem~\ref{thm:N-DSG}. A perfect agreement can be observed.

\subsection{Soliton asymptotics}
\label{s:intro-soliton-asymptotics}

With the general characterization of $N$-DSGs obtained in Section~\ref{s:intro-dsg},
we are ready to consider the nonlinear interactions among several DSGs with different sizes.
Thus,
in this subsection,
we consider more eigenvalue groups,
meaning $J >1$ from Definition~\ref{def:LambdaOmega}.
We would like to calculate the long-time asymptotics of general $N$-soliton solution as $t\to\pm\infty$,
corresponding to the investigation of the light-matter interactions inside the optical medium from the distant past to far future,
in different directions $z = \xi t$ with $\xi\in\Real$.
Hence,
we apply the Deift-Zhou's nonlinear steepest descent method to the oscillatory RHP~\ref{rhp:N-soliton-jump-form}
and calculate the leading terms and perform error estimates
\myblue{in Section~\ref{s:N-soliton-asymptotics}}.
\begin{theorem}[Soliton asymptotics]
\label{thm:soliton-asymptotics}
Suppose that $J \ge 2$.
Let $\Lambda$ and $\Omega$ be given according to Definition~\ref{def:LambdaOmega}.
\myblue{Let us define
\begin{equation}
\label{e:Vj-def}
V_j\coloneqq -2D_-r_j^2\,,\qquad
r_j = |\lambda_{j,k}|\,,\qquad
1\le j\le J\,.
\end{equation}
}
The general $N$-soliton solution from Theorem~\ref{thm:N-soliton-formula} has the long-time asymptotics with $z = \xi t$:
\begin{enumerate}
\item
\myblue{
If $\xi = V_j$ for each $1\le j\le J$,
then in both stable and unstable \myblue{media} as $t\to\pm\infty$ one obtains the asymptotic expansion
}
\begin{equation}
\myblue{
\begin{aligned}
q(t,z;\Lambda,\Omega)
 & = q(t,z;\Lambda_j,\Omega_{j}^{(\pm)}) + \O(\ee^{-\aleph |t|})\,,\\
D(t,z;\Lambda,\Omega)
 & = D(t,z;\Lambda_j,\Omega_{j}^{(\pm)}) + \O(\ee^{-\aleph |t|})\,,\qquad
P(t,z;\Lambda,\Omega)
 = P(t,z;\Lambda_j,\Omega_{j}^{(\pm)}) + \O(\ee^{-\aleph |t|})\,.
\end{aligned}
}
\end{equation}
\myblue{
with a positive constant $\aleph$.
The leading-order term
$\{q(t,z;\Lambda_j,\Omega_{j}^{(\pm)}), D(t,z;\Lambda_j,\Omega_{j}^{(\pm)}), P(t,z;\Lambda_j,\Omega_{j}^{(\pm)})\}$
denotes an $N_j$-DSG of the MBEs in the corresponding medium with the eigenvalue set $\Lambda_j$ and a modified norming constant set $\Omega_j^{(\pm)} \coloneqq \{\omega_{j,k}^{(\pm)}\}_{k = 1}^{N_j}$ given by
}
\begin{equation}
\omega_{j,k}^{(+)}
 \coloneqq \omega_{j,k}\delta_{j+1}(\lambda_{j,k})^{-2}\,,\qquad
\omega_{j,k}^{(-)}
 \coloneqq \omega_{j,k}\delta_j(\lambda_{j,k})^2\delta_1(\lambda_{j,k})^{-2}\,,
\end{equation}
with
\begin{equation}
\label{e:deltan-def}
\delta_j(\lambda)
 \coloneqq \frac{p_j^*(\lambda)}{p_j(\lambda)}\,,\qquad
p_j(\lambda)
 \coloneqq \prod_{s = j}^J\prod_{l = 1}^{N_s}(\lambda - \lambda_{s,l})\,,\qquad
j = 1,2,\dots,J\,.
\end{equation}
\item If $\xi \ne V_j$ for all $1\le j \le J$, then as $t\to\pm\infty$
\begin{equation}
\myblue{
q(t,z;\Lambda,\Omega) = \O(\ee^{-\aleph |t|})\,,\qquad
D(t,z;\Lambda,\Omega) = D_- + \O(\ee^{-\aleph |t|})\,,\qquad
P(t,z;\Lambda,\Omega) = \O(\ee^{-\aleph |t|})\,,
}
\end{equation}
\myblue{for a positive constant $\aleph$}.
\end{enumerate}
\end{theorem}
Theorem~\ref{thm:soliton-asymptotics} immediately implies that, asymptotically, the general $N$-soliton solution is a linear combination of multiple DSGs with various sizes,
\begin{equation}
\myblue{
\begin{aligned}
q(t,z;\Lambda,\Omega)
    & = \sum_{j=1}^J q(t,z;\Lambda_j,\Omega_{j}^{(\pm)}) + \O(\ee^{-\aleph |t|})\,,\\
D(t,z;\Lambda,\Omega)
    & = \sum_{j=1}^J D(t,z;\Lambda_j,\Omega_{j}^{(\pm)}) + \O(\ee^{-\aleph |t|})\,,\qquad
P(t,z;\Lambda,\Omega)
    = \sum_{j=1}^J P(t,z;\Lambda_j,\Omega_{j}^{(\pm)}) + \O(\ee^{-\aleph |t|})\,.
\end{aligned}
}
\end{equation}
Applying the DSG center formula from Theorem~\ref{thm:N-DSG} to Theorem~\ref{thm:soliton-asymptotics} immediately yields the asymptotic phase shifts induced by nonlinear interactions among DSGs.
\begin{corollary}[Asymptotic shifts]
\label{thm:asymptotic-shifts}
The asymptotic centers for each $N_j$-DSG with parameters $\Lambda_j$, $\Omega_j^{(\pm)}$ and $j = 1,2,\dots,J$, as $t\to\pm\infty$ from Theorem~\ref{thm:soliton-asymptotics}, in both stable and unstable \myblue{media}, are given by
\begin{equation}
\begin{aligned}
z_{\c,j}^{+}(t)
 & \coloneqq V_j t + z_{\d,j}^{(+)}\,,\quad
z_{\d,j}^{(+)}
 \coloneqq -\frac{D_- r_j}{\sum_{k=1}^{N_j}\sin(\alpha_{j,k})}\ln\bigg(\prod_{k=1}^{N_j}\big|\omega_{j,k}^{(+)}\big| \,\frac{\prod_{k = 2}^{N_j}\prod_{l = 1}^{k-1}|\lambda_{j,k} - \lambda_{j,l}|^2}{\prod_{k=1}^{N_j}\prod_{l=1}^{N_j}|\conj{\lambda_{j,k}} - \lambda_{j,l}|}\bigg)\,,\quad
t\to+\infty\,,\\
z_{\c,j}^{-}(t)
 & \coloneqq V_j t + z_{\d,j}^{(-)}\,,\quad
z_{\d,j}^{(-)}
 \coloneqq -\frac{D_- r_j}{\sum_{k=1}^{N_j}\sin(\alpha_{j,k})}\ln\bigg(\prod_{k=1}^{N_j}\big|\omega_{j,k}^{(-)}\big| \,\frac{\prod_{k = 2}^{N_j}\prod_{l = 1}^{k-1}|\lambda_{j,k} - \lambda_{j,l}|^2}{\prod_{k=1}^{N_j}\prod_{l=1}^{N_j}|\conj{\lambda_{j,k}} - \lambda_{j,l}|}\bigg)\,,\quad
t\to-\infty\,.
\end{aligned}
\end{equation}
Comparing with the unshifted displacement $z_{\d, j}$ from Theorem~\ref{thm:N-DSG} yields the asymptotic shifts
\begin{equation}
\begin{aligned}
\Delta_j^{(+)}
 & \coloneqq z_{\d,j}^{(+)} - z_{\d,j}
 = \frac{2D_- r_j}{\sum_{k=1}^{N_j}\sin(\alpha_{j,k})}\sum_{k=1}^{N_j}\sum_{s = j+1}^{J}\sum_{l = 1}^{N_s}\ln\bigg|\frac{\lambda_{j,k} - \conj{\lambda_{s,l}}}{\lambda_{j,k} - \lambda_{s,l}}\bigg|\,,\\
\Delta_j^{(-)}
 & \coloneqq z_{\d,j}^{(-)} - z_{\d,j}
 = \frac{2D_- r_j}{\sum_{k = 1}^{N_j}\sin(\alpha_{j,k})}\sum_{k = 1}^{N_j}\sum_{s = 1}^{j-1}\sum_{l = 1}^{N_s}\ln\bigg|\frac{\lambda_{j,k} - \conj{\lambda_{s,l}}}{\lambda_{j,k} - \lambda_{s,l}}\bigg|\,,
\end{aligned}
\end{equation}
The overall asymptotic phase shift is given by
\begin{equation}
\Delta_j
 \coloneqq z_{\d,j}^{(+)} - z_{\d,j}^{(-)}
 = \frac{2D_- r_j}{\sum_{k = 1}^{N_j}\sin(\alpha_{j,k})}\sum_{k = 1}^{N_j}\ln\bigg(\prod_{s = j+1}^J\prod_{l = 1}^{N_s}\bigg|\frac{\lambda_{j,k} - \conj{\lambda_{s,l}}}{\lambda_{j,k} - \lambda_{s,l}}\bigg|\cdot\prod_{s = 1}^{j-1}\prod_{l = 1}^{N_s}\bigg|\frac{\lambda_{j,k} - \lambda_{s,l}}{\lambda_{j,k} - \conj{\lambda_{s,l}}}\bigg|\bigg)\,.
\end{equation}
\end{corollary}
Two examples of general soliton asymptotics are shown in Figure~\ref{f:MBEs-DSG-soliton-asymptotics}, where the medium is chosen as initially stable $D_- = -1$. Only the optical pulse $|q(t,z)|$ are shown in the two examples. The medium functions $\{D(t,z)\,, P(t,z)\}$ can also be obtained from Theorem~\ref{thm:N-soliton-formula} and can be compared with Theorem~\ref{thm:soliton-asymptotics}, but are omitted to avoid repetition. The center panel contains the density plot of a $3$-soliton solution with two DSGs, with parameters given in the caption. Four red dashed lines are drawn describing the asymptotic center $z_{\c,j}^\pm(t)$ for each DSG, obtained from Theorem~\ref{thm:soliton-asymptotics}. Similarly, the right panel contains a $5$-soliton solution with three DSGs. In both plots, phase shifts resulting from nonlinear interactions can be observed. Theorem~\ref{thm:soliton-asymptotics} perfectly describes the asymptotic behavior for each DSG, before and after nonlinear collisions.

\begin{figure}[tp]
\centering
\includegraphics[width=0.31\textwidth]{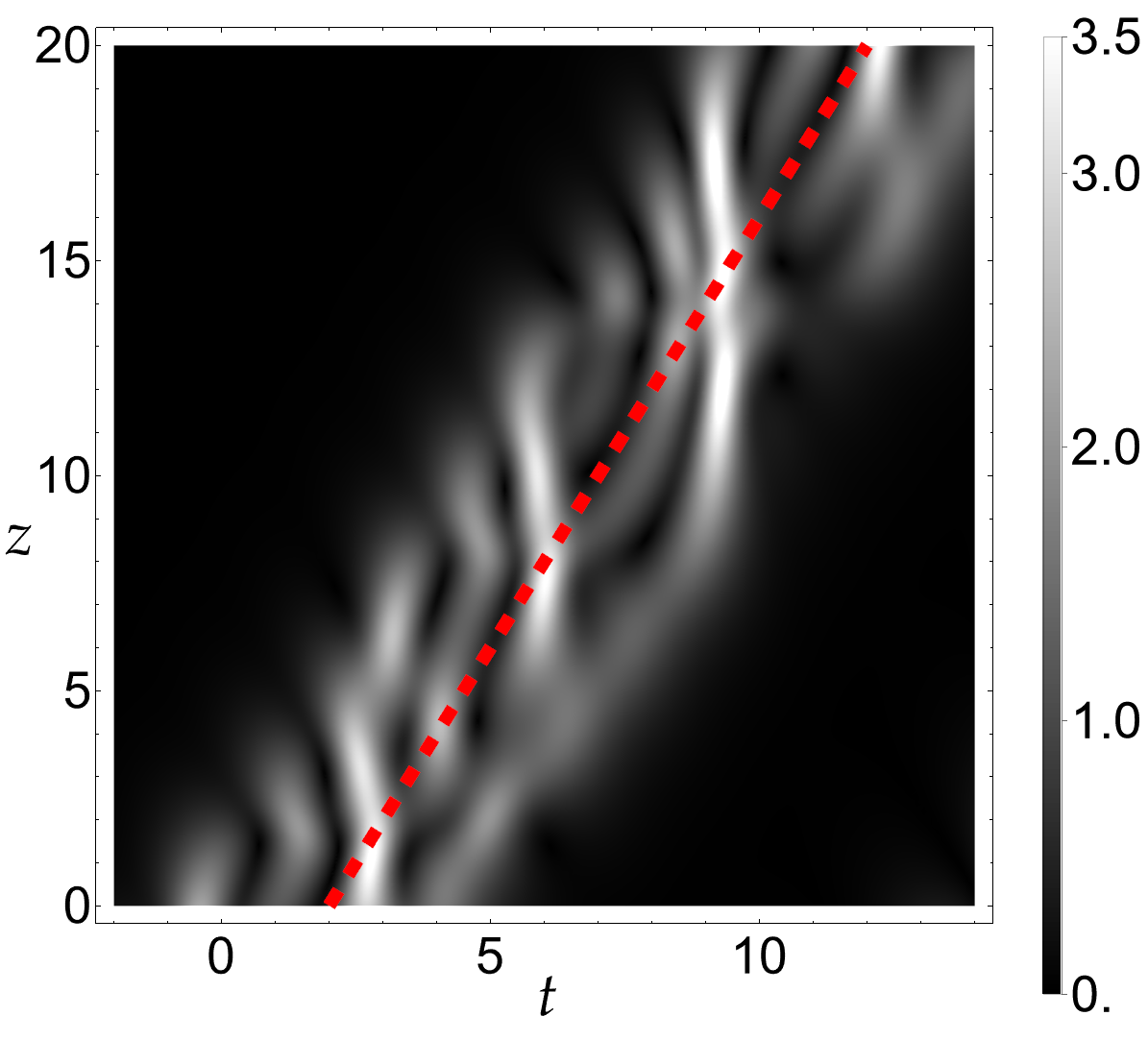}\quad
\includegraphics[width=0.31\textwidth]{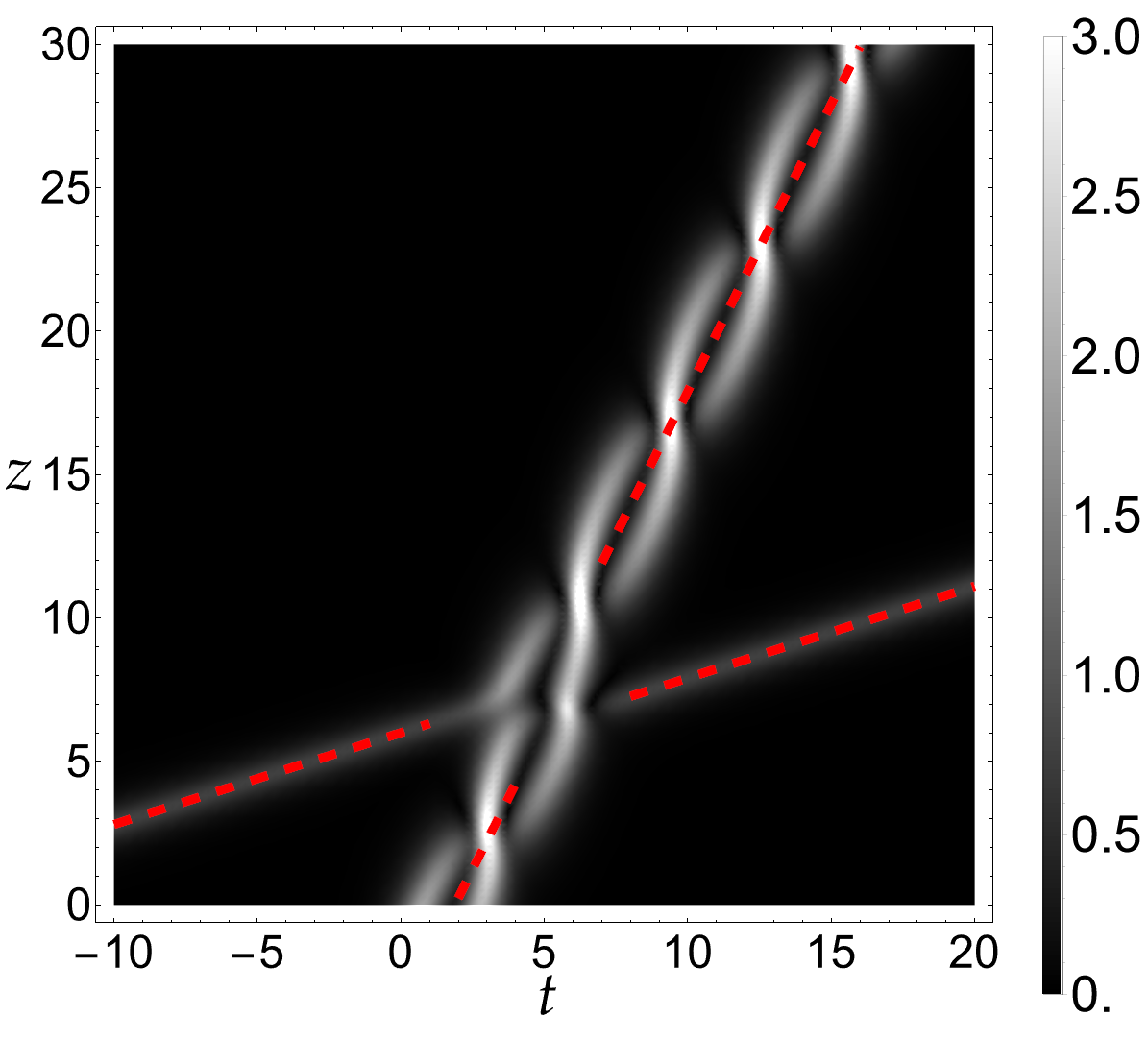}\quad
\includegraphics[width=0.31\textwidth]{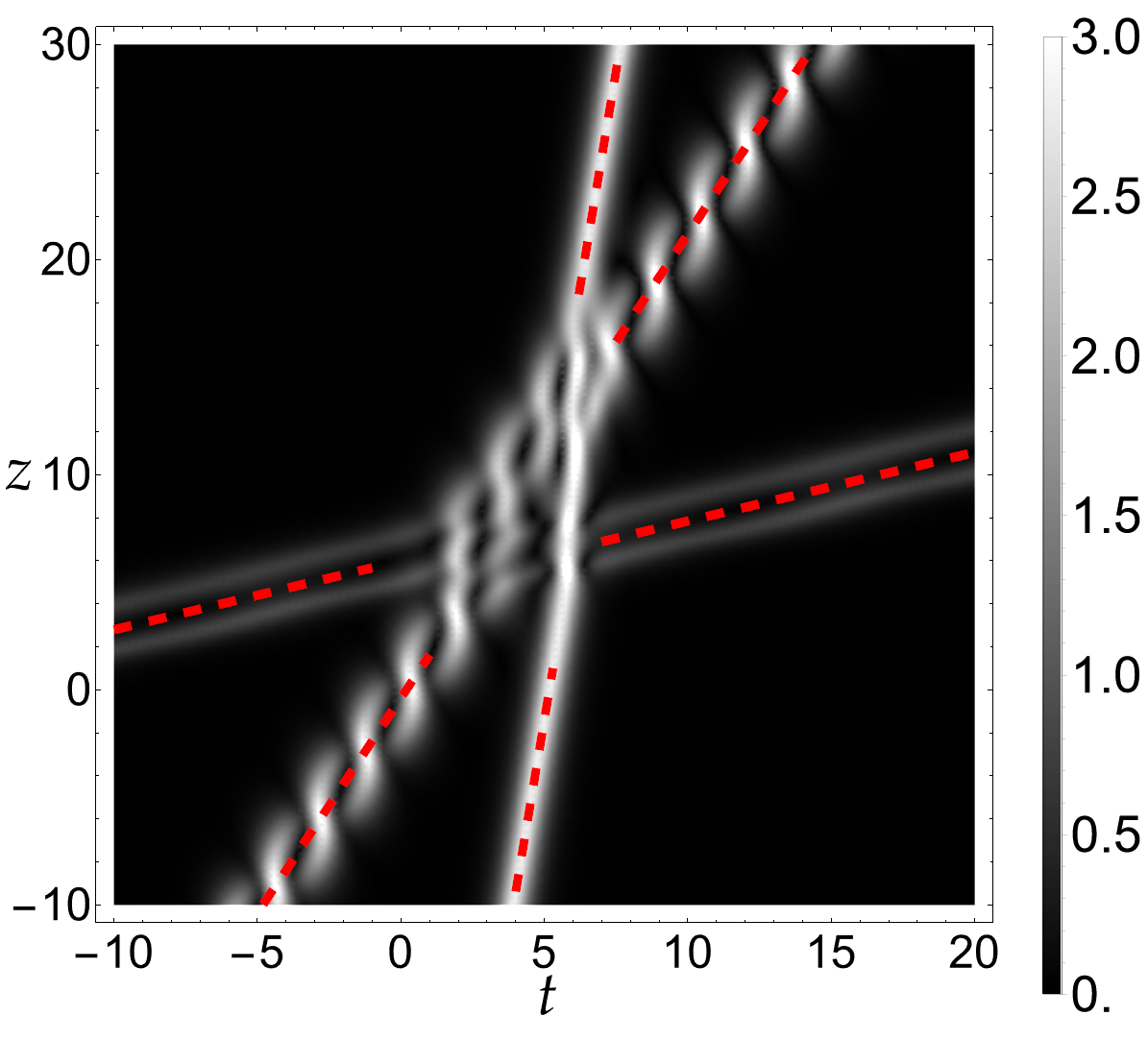}
\caption{
    Density plots of \myblue{$|q(t,z)|$} of exact $N$-soliton solutions in stable \myblue{media} from Theorem~\ref{thm:N-soliton-formula} with different setups.
    Left: A $4$-DSG with parameters $r_{1} = 1$, $\alpha_{1,k}= \pi/5,\pi/3,\pi/2,2\pi/3$, and $\omega_{1,k} = 1$ for all $1\le k \le 4$. The red dashed line corresponds to the center in Theorem~\ref{thm:N-DSG}.
    Center: A $3$-soliton solution with two DSGs. The parameters are $r_1 = 0.4$, $r_2 = 1$, $\alpha_{1,1} = 2\pi/5$, $\alpha_{2,1} = \pi/3$, $\alpha_{2,2} = \pi/2$, $\omega_{1,1} = \ee^{14}$, $\omega_{2,1} = 1$, and $\omega_{2,2} = \ii$. The red dashed lines correspond to the asymptotic centers in Corollary~\ref{thm:asymptotic-shifts}.
    Right: A $5$-soliton solution with three DSGs. The parameters are $r_1 = 0.4$, $r_2 = 1$, $r_3 = 2$, $\alpha_{1,k} = \pi/3,\pi/2$, $\alpha_{2,k} = \pi/3,2\pi/3$, $\alpha_{3,1} = \pi/4$, $\omega_{1,1} = \omega_{1,2} = \ee^{15}$, $\omega_{2,1} = \ee^3$, $\omega_{2,2} = \ee^4$, and $\omega_{3,1} = \ee^{-10}$. Again, the red dashed lines correspond to the asymptotic centers in Corollary~\ref{thm:asymptotic-shifts}.
}
\label{f:MBEs-DSG-soliton-asymptotics}
\end{figure}
%

\subsection{High-order solitons}

Merging two eigenvalues with proper rescaling of the norming constants yields
a double-pole soliton of the NLS equation~\cite{zs1972,akns1974}.
In this paper,
we make the calculation more general
and present the special limiting procedure in a \myblue{systematical} way,
by showing that merging $N>1$ eigenvalues into one,
with rescaling of the norming constants $\{\omega_{j,k}\}$,
produces an $N$th-order soliton, i.e., $N$-pole soliton of the MBEs~\eqref{e:mbe}.

By definition,
the $N$th-order soliton is the solution corresponding to an $N$th-order eigenvalue,
i.e.,
an $N$th pole of the scattering data in the inverse scattering.
So,
let us denote this eigenvalue in the upper half plane by
$\lambda_\circ\in\Complex^+$.
Naturally,
the other eigenvalue $\conj{\lambda_\circ}$ is in the lower half plane.
Then,
the corresponding $N$th pole in the upper half plane can be described by
$\frac{p_\circ(\lambda)}{(\lambda - \lambda_\circ)^N}$ with $N\ge2$ and
\begin{equation}
\label{e:p-circ-def}
p_\circ(\lambda)\coloneqq \sum_{k=0}^{N-1}\omega_{\circ,k}(\lambda - \lambda_\circ)^k\,,\qquad
\omega_{\circ,k}\in\Complex\,,\qquad
\omega_{\circ,0} \ne 0\,.
\end{equation}
By definition,
$\omega_{\circ,0}$ cannot be zero,
otherwise this is not an $N$th pole anymore.
However,
all other constants can be arbitrary,
including zeros.

Then,
the most general $N$th-order soliton can be obtained via the following RHP.
\begin{rhp}[Jump form of the $N$th-order soliton]
\label{rhp:N-order-soliton:jump-form}
\myblue{
Let $\M_\circ(\lambda;t,z)$ be a two-by-two matrix analytic function on
$\Complex\setminus\big(\partial D_{\lambda_\circ}^\epsilon\bigcup \partial D_{\conj{\lambda_\circ}}^\epsilon\big)$
admitting continuous boundary values.
It has the asymptotics $\M_\circ(\lambda;t,z)\to\I$ as $\lambda\to\infty$ and satisfies jumps
}
\begin{equation}
\begin{aligned}
\M_\circ^+(\lambda;t,z) & = \M_\circ^-(\lambda;t,z)\V_\circ(\lambda;t,z)^{-1}\,,\qquad && \lambda\in\partial D_{\lambda_\circ}^\epsilon\,,\\
\M_\circ^+(\lambda;t,z) & = \M_\circ^-(\lambda;t,z)\V_\circ(\conj\lambda;t,z)^\dagger\,,\qquad && \lambda\in\partial D_{\conj{\lambda_\circ}}^\epsilon\,,
\end{aligned}
\end{equation}
where the jump contours are oriented counterclockwise,
and the jump matrices are given by
\begin{equation}
\everymath{\displaystyle}
\V_\circ(\lambda;t,z)
    \coloneqq \bpm
        1 & 0 \\ \frac{p_\circ(\lambda)\ee^{-2\ii\theta(\lambda;t,z)}}{(\lambda - \lambda_\circ)^N} & 1
    \epm\,,\qquad
\V_\circ(\conj\lambda;t,z)^\dagger
    = \bpm
        1 & \frac{p_\circ^*(\lambda)\ee^{2\ii\theta(\lambda;t,z)}}{(\lambda - \conj{\lambda_\circ})^N} \\ 0 & 1
    \epm\,.
\end{equation}
\end{rhp}
Clearly,
both jump matrices in RHP~\ref{rhp:N-order-soliton:jump-form} satisfy the Schwarz symmetry,
\myblue{after the jump contour in the upper-half plane is re-oriented clockwise,}
so by Zhou's lemma the solution to RHP~\ref{rhp:N-order-soliton:jump-form} exists and is unique~\cite{z1989}.
It is easy to apply the dressing method on RHP~\ref{rhp:N-order-soliton:jump-form} in order to show that it, together with the reconstruction formula in Lemma~\ref{thm:reconstruction},
produces a unique solution of the MBEs with boundary conditions $\lim_{t\to-\infty}D = D_-$ for all $z\ge0$. The procedure is identical to the one in Section~\ref{s:proof-reconstruction-formula}, so it is omitted for brevity.

Now, a few questions rise:
\begin{enumerate}
\item What is relation between RHPs~\ref{rhp:N-soliton-jump-form} and~\ref{rhp:N-order-soliton:jump-form}?
\item What is the pole-form of RHP~\ref{rhp:N-order-soliton:jump-form}?
\item How can the $N$th-pole soliton be solved from RHP~\ref{rhp:N-order-soliton:jump-form}?
\end{enumerate}
We answer all three questions below.

First, we present the following theorem in order to show the relation between RHPs~\ref{rhp:N-soliton-jump-form} and~\ref{rhp:N-order-soliton:jump-form}.
\begin{theorem}[Fusion of solitons]
\label{thm:N-order-soliton:fusion}
For every given RHP~\ref{rhp:N-order-soliton:jump-form} with the order $N\ge2$, the eigenvalue $\lambda_\circ\in\Complex^+$ and norming constants $\{\omega_{\circ,k}\}_{k=0}^{N-1}$ and $\omega_{\circ,0}\ne0$, there exists a fusion process of merging the (simple) eigenvalue set $\Lambda$ of RHP~\ref{rhp:N-soliton-jump-form} while rescaling norming constant set $\Omega$.
\end{theorem}
Theorem~\ref{thm:N-order-soliton:fusion} is proved in Section~\ref{s:N-order-soliton:fusion}, where the explicit rescaling of the norming constants $\Omega$ is also provided in Equation~\eqref{e:N-order-soliton:omega-rescaling}. Moreover, we also show that the rescaling of norming constants is essential when obtaining the $N$th-order soliton from the $N$-soliton solution. Without it, the fusion of $\Lambda$ just produces an one-soliton solution, which is the trivial result.

It is established that it is definitely possible to obtain the $N$th-order soliton from $N$-soliton solutions. The next task is to solve for such solution from RHP~\ref{rhp:N-order-soliton:jump-form}. It turns out that while the jump form a RHP is ideal for the Deift-Zhou's nonlinear steepest descent method, the equivalent pole form becomes useful when solving for the exact soliton solutions. As such, we address the second question, which is to rewrite RHP~\ref{rhp:N-order-soliton:jump-form} into its pole form. Different from the previous simple pole cases, where the residue conditions are used, the high-order poles require more general conditions. Therefore, we first define the pole operator $\sP{n}$ which picks out the coefficient of the $n$th-order term in a Laurent expansion of a function.
\begin{definition}
\label{def:pole-operator}
Let us define a operator $\sP{n}$, such that for every meromorphic function $f(\lambda)$,
\begin{equation}
\sP{n}\limits_{\lambda = \lambda_\circ} f = f_n\,,\qquad\mbox{ where }
f = \sum_{n = -N}^\infty (\lambda - \lambda_\circ)^n f_n\,,
\end{equation}
\end{definition}
Clearly, the operator $\sP{n}$ is a generalization of $\Res$, with the following property
\begin{equation}
\sP{-1}\limits_{\lambda = \lambda_\circ}f
    = \Res\limits_{\lambda = \lambda_\circ}f\,.
\end{equation}
With the help of $\sP{n}$, the pole form of an $N$th-order soliton RHP is given below.
\begin{rhp}[Pole form of the $N$th-order soliton]
\label{rhp:N-order-soliton:pole-form}
Given the $N$th-order eigenvalue $\lambda_\circ$
in the upper half plane and corresponding norming constant polynomial
$p_\circ(\lambda)$.
Let $\M_\circ(\lambda;t,z)$ be a $2\times2$
\myblue{analytic matrix function on $\Complex\setminus\{\lambda_\circ,\conj{\lambda_\circ}\}$},
with asymptotics $\M_\circ(\lambda;t,z)\to\I$ as $\lambda\to\infty$
and pole conditions
\begin{equation}
\begin{aligned}
\sP{-n}_{\lambda = \lambda_\circ}\M_\circ(\lambda;t,z)
    & = \lim_{\lambda\to\lambda_\circ}\sum_{k=0}^{N-n}\sum_{i=0}^{N-n-k}\frac{\omega_{\circ,i}}{(N-n-i-k)!\,k!}\frac{\partial^{N-n-i-k}\ee^{-2\ii\theta}}{\partial\lambda^{N-n-i-k}}\frac{\partial^k\M_\circ}{\partial\lambda^k}\bpm 0 & 0 \\ 1 & 0 \epm\,,\\
\sP{-n}_{\lambda = \conj{\lambda_\circ}}\M_\circ(\lambda;t,z)
    & = -\lim_{\lambda\to\lambda_\circ}\sum_{k=0}^{N-n}\sum_{i=0}^{N-n-k}\frac{\conj{\omega_{\circ,i}}}{(N-n-i-k)!\,k!}\frac{\partial^{N-n-i-k}\ee^{2\ii\theta}}{\partial\lambda^{N-n-i-k}}\frac{\partial^k\M_\circ}{\partial\lambda^k}\bpm 0 & 1 \\ 0 & 0 \epm\,,
\end{aligned}
\end{equation}
where $n = 1,2,\dots, N$.
\end{rhp}
The derivation RHP~\ref{rhp:N-order-soliton:pole-form} from RHP~\ref{rhp:N-order-soliton:jump-form} is shown in Section~\ref{s:N-order-soliton:pole-form}. Similarly to the $N$-soliton solutions. Both RHPs~\ref{rhp:N-order-soliton:jump-form} and~\ref{rhp:N-order-soliton:pole-form} are equivalent and useful in different ways. RHP~\ref{rhp:N-order-soliton:jump-form} is suitable for applying Deift-Zhou's nonlinear steepest descent method when computing asymptotics, whereas RHP~\ref{rhp:N-order-soliton:pole-form} can be used to find the exact solution {formul\ae} for the matrix $\M_\circ(\lambda;t,z)$ and subsequently for the MBEs. As shown in Section~\ref{s:N-order-soliton:solution-formula}, the exact formula for the $N$th-order soliton is given below.
\begin{theorem}[$N$th-order soliton solution formula]
\label{thm:N-order-soliton:solution-formula}
Given an $N$th-order eigenvalue \myblue{$\lambda_\circ\in\Complex^+$}
and norming constants $\{\omega_{\circ,k}\}_{k=0}^{N-1}$,
the corresponding $N$th-order soliton of MBEs is given by
\begin{equation}
\begin{aligned}
q(t,z) = 2\ii - 2\ii\frac{\det(\I - \myblue{\vec{e_1}\conj{\vec\gamma_\infty}} + \bGamma_\circ\conj{\bGamma_\circ})}{\det(\I + \bGamma_\circ\conj{\bGamma_\circ})}\,,
\end{aligned}
\end{equation}
where $D(t,z)$ and $P(t,z)$ are reconstructed by Lemma~\ref{thm:reconstruction}, with entries of $\M(\lambda;t,z)$ given by
\begin{equation}
\begin{aligned}
M_{1,1}(\lambda;t,z)
    = \frac{\det(\I - \myblue{\vec{e_1}\vec{\gamma}\conj{\bGamma_\circ}} + \bGamma_\circ \conj{\bGamma_\circ})}{\det(\I + \bGamma_\circ\conj{\bGamma_\circ})}\,,\qquad
M_{1,2}(\lambda;t,z)
    = \frac{\det(\I - \myblue{\vec{e_1}\conj{\vec\gamma}} + \bGamma_\circ\conj{\bGamma_\circ})}{\det(\I + \bGamma_\circ\conj{\bGamma_\circ})}-1\,,
\end{aligned}
\end{equation}
and other quantities
\begin{equation}
\everymath{\displaystyle}
\begin{aligned}
\vec{\gamma}(\lambda)
    & \coloneqq \bpm \gamma_0(\lambda) & \cdots & \gamma_{N-1}(\lambda)\epm\,,\qquad
\vec{\gamma}_\infty
    \coloneqq \lim_{\lambda\to\infty}\lambda\vec{\gamma}(\lambda)\,,\\
\bGamma_\circ
    & \coloneqq \bpm \vec{\gamma}(\conj{\lambda_\circ})^\top & \cdots & \vec{\gamma}^{(N-1)}(\conj{\lambda_\circ})^\top \epm^\top\,,\qquad
\vec{e_1}\coloneqq \bpm 1 & 0 & \dots & 0 \epm^\top\,,\\
\gamma_k(\lambda)
    & \coloneqq \sum_{n=1}^{N-k}\sum_{i=0}^{N-n-k}\frac{1}{(\lambda - \lambda_\circ)^n}\frac{\omega_{\circ,i}}{(N-n-i-k)!\,k!}\frac{\partial^{N-n-i-k} \ee^{-2\ii\theta}}{\partial\lambda^{N-n-i-k}}(\lambda_\circ)\,,
\end{aligned}
\end{equation}
\myblue{
where $\vec{\gamma}^{(n)}(\lambda)$ denotes the $n$th derivative of $\vec{\gamma}(\lambda)$.
}
\end{theorem}
Since the $N$th-order soliton is a special limit of the $N$-DSG with rescaled norming constants given in Equation~\eqref{e:N-order-soliton:omega-rescaling} as shown in Section~\ref{s:N-order-soliton:fusion}, one immediately obtains the following result.
\begin{corollary}[Localization for the $N$th-order soliton]
\label{e:thm:N-order-soliton:localization}
As a whole,
the $N$th-order soliton is localized along the line $z = V_\circ t$ with its velocity $V_\circ \coloneqq -2D_- |\lambda_\circ|^2$.
\end{corollary}
\myblue{
However,
we currently are not able to provide a simple formula for the displacement.
We are also not able to provide the location for each peak in a high-order soliton solution.
Note that the peak locations in other integrable systems have been derived~\cite{s2017,zl2021}.
}

\begin{figure}[tp]
\centering
\includegraphics[width=0.3\textwidth]{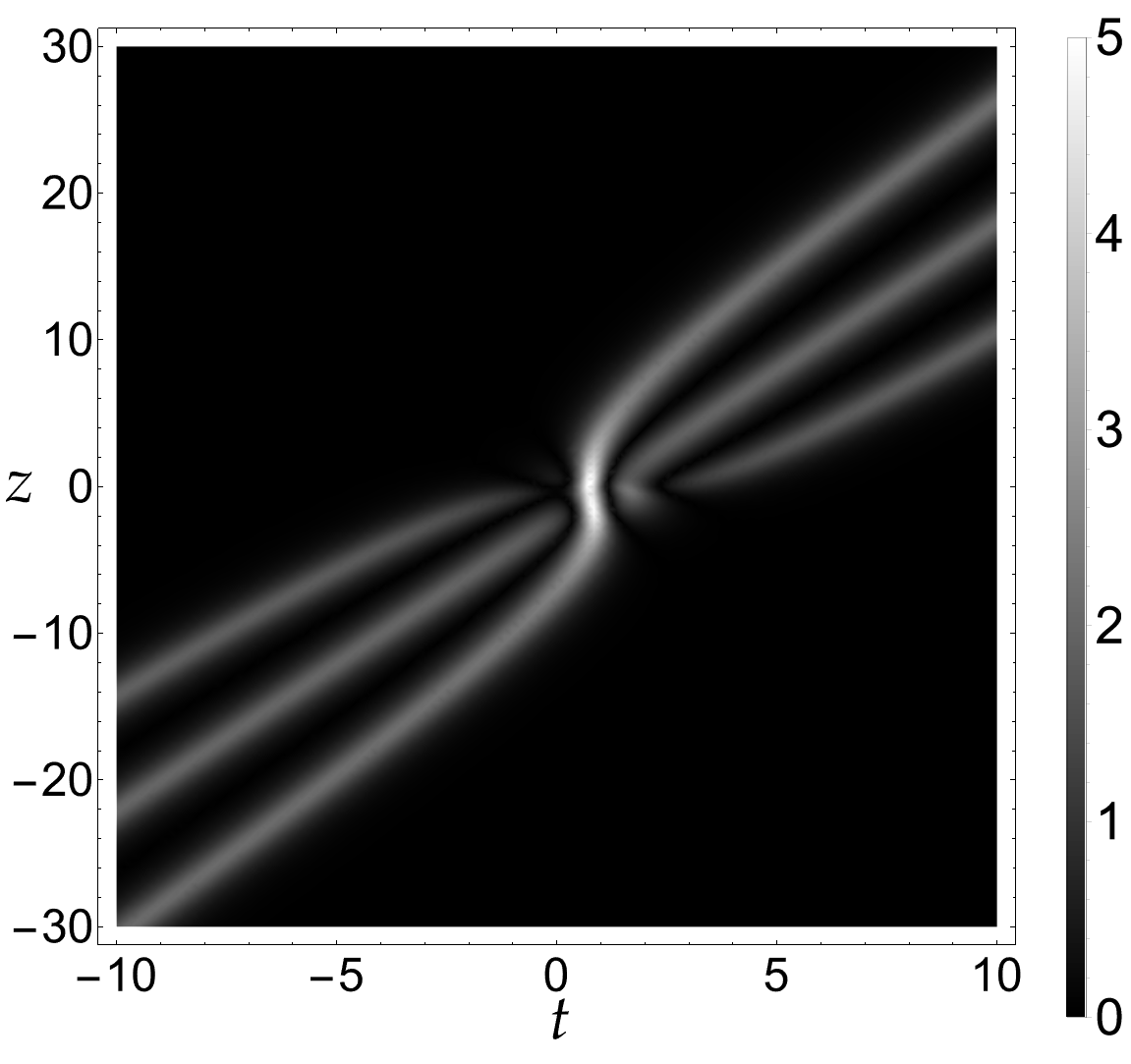}\quad
\includegraphics[width=0.3\textwidth]{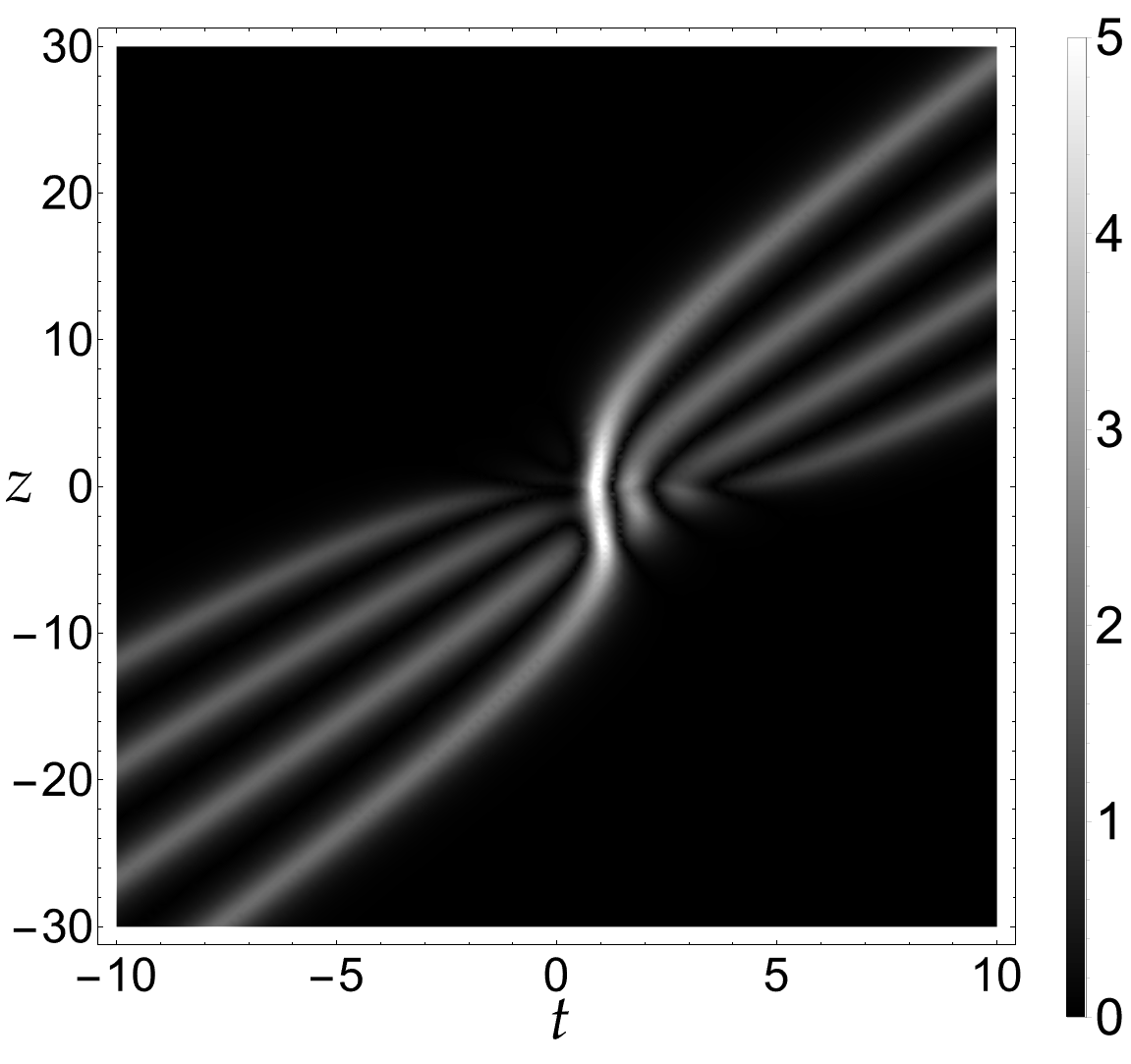}\quad
\includegraphics[width=0.3\textwidth]{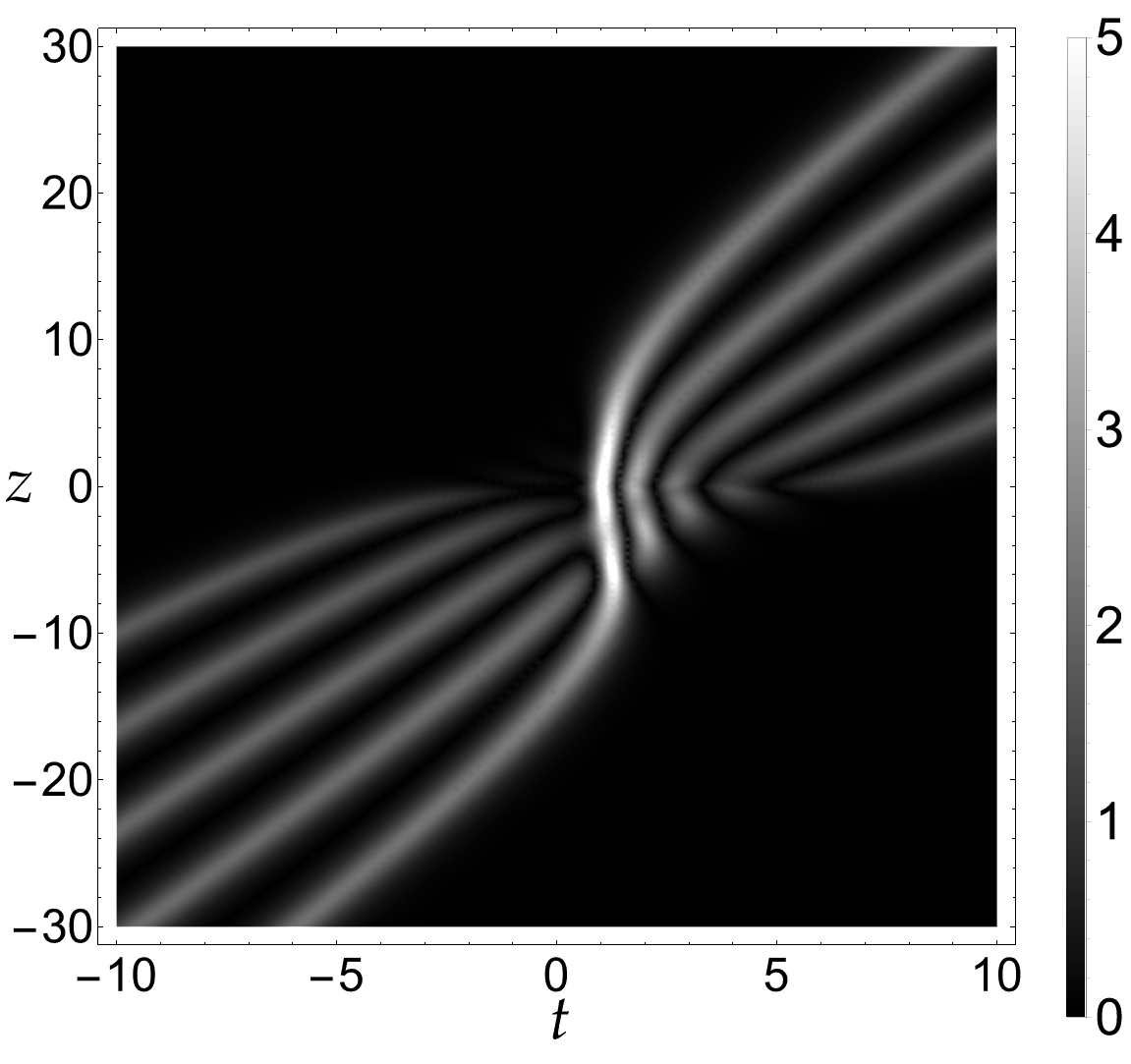}
\caption{
    Density plots of $|q(t,z)|$ of three exact $N$th-order soliton solutions in stable \myblue{media} from Theorem~\ref{thm:N-order-soliton:solution-formula}. All three plots have similar parameters $\lambda_\circ = \ii$, $\omega_{\circ,0} = 1$ and $\omega_{\circ,k} = 0$ for $k\ge1$, but with different values of $N$. Left: $N = 3$. Center: $N = 4$. Right: $N = 5$.
}
\label{f:Nth-order-solions}
\end{figure}

Three high-order solitons from Theorem~\ref{thm:N-order-soliton:solution-formula} are shown in Figure~\ref{f:Nth-order-solions} in an initially stable medium $D_- = -1$, with $N = 3,4,5$, respectively. In all three panels, the eigenvalues and the norming constant polynomials remain unchanged as $\lambda_\circ = \ii$ and $p_\circ(\lambda) = 1$, respectively. Hence, from left to right, the solutions correspond to upper-half-plane poles $(\lambda - \ii)^{-3}$, $(\lambda - \ii)^{-4}$ and $(\lambda - \ii)^{-5}$. They travel inside the light cone with identical velocity $V = 2$.

\subsection{Soliton gas}

Beside the $N$th-order soliton,
another generalization of the general $N$-soliton solutions,
corresponding to anomaly (b) in the spectra of the non-self-adjoint Zakharov-Shabat scattering problem,
produces the so-called soliton gas phenomenon, or integrable turbulence. In particular, one considers the limit of an arbitrary $N$-soliton solution as $N\to+\infty$. In this work, however, we only present results of a special case, in which one considers the limit of an $N$-DSGs as $N\to+\infty$, in order to keep discussions simple.
\begin{definition}[Generalized eigenvalues and norming constants]
\label{def:generalized-eigenvalue-norming-constant}
We define an arc $\rA$ generalizing the eigenvalue set $\Lambda$
\begin{equation}
\label{e:soliton-gas:A_def}
\rA \coloneqq \Set{\lambda}{$|\lambda| = r_1 > 0$ with $\arg(\lambda)\in(\epsilon,\pi - \epsilon)$}\,,\qquad
\mbox{with $0 < \epsilon \ll1$ fixed}\,,
\end{equation}
and a complex function $\omega(\lambda)\in L_2(\rA)$ generalizing the norming constant set $\Omega$.
\end{definition}
\myblue{The bound state soliton gas solution} of MBEs~\eqref{e:mbe}
is described by the following RHP together with the reconstruction formula Theorem~\ref{thm:reconstruction}.
This RHP is derived in Appendix~\ref{s:soliton-gas} as a limit of $N$-DSGs with $N\to+\infty$.
\begin{rhp}[Jump-form of soliton gas]
\label{rhp:soliton-gas:inftyDSG_jump_form}
\myblue{
Given an arc $\rA$ and a function $\omega(\lambda)$ according to Definition~\ref{def:generalized-eigenvalue-norming-constant}.
Seek a $2\times2$ matrix function $\lambda\mapsto\M_\sg(\lambda;t,z)$ analytic on $\Complex\setminus\{\rA\bigcup\conj{\rA}\}$,
admitting continuous boundary values on $\rA\bigcup\conj{\rA}$,
with asymptotics $\M_\sg(\lambda;t,z)\to\I$ as $\lambda\to\infty$ and jumps
}
\begin{equation}
\begin{aligned}
\M_\sg^{+}(\lambda;t,z)
 & = \M_\sg^{-}(\lambda;t,z)\V_\sg(\lambda;t,z)\,,\qquad&& \lambda\in \rA\,,\\
\M_\sg^{+}(\lambda;t,z)
 & = \M_\sg^{-}(\lambda;t,z)\V_\sg(\conj\lambda;t,z)^\dagger\,,\qquad&& \lambda\in \conj{\rA}\,,
\end{aligned}
\end{equation}
where the jump contours are oriented left-to-right,
and the jump matrix $\V_\sg(\lambda;t,z)$ is
\begin{equation}
\V_\sg(\lambda;t,z)
 \coloneqq \bpm1 & 0 \\ \ee^{-2\ii\theta(\lambda;t,z)}\omega(\lambda) & 1 \epm\,.
\end{equation}
\end{rhp}
\begin{remark}
Jump matrices in RHP~\ref{rhp:soliton-gas:inftyDSG_jump_form} satisfy the Schwarz symmetry,
so by Zhou's lemma the solution exists and is unique~\cite{z1989}.
RHP~\ref{rhp:soliton-gas:inftyDSG_jump_form}
can be easily generalized following the \myblue{discussion} in Appendix~\ref{s:soliton-gas}
by modifying the jump contour to more complex configurations.
It is easy to prove that
RHP~\ref{rhp:soliton-gas:inftyDSG_jump_form} \myblue{generates} a unique solution of MBEs~\eqref{e:mbe} with boundary condition
$D \to D_-$ as $t\to-\infty$.
The proof is similar to Section~\ref{s:proof-reconstruction-formula}
and is omitted for brevity.
\end{remark}
We would also like to obtain properties of the soliton gas solution,
such as the localization property,
by mimicking the proofs for Theorem~\ref{thm:N-DSG} and Corollary~\ref{e:thm:N-order-soliton:localization}.
However,
it turns out that certain steps require additional knowledge of the function
$\omega(\lambda)$.
The soliton gas solutions exhibit diverse properties depending on the choice of
$\omega(\lambda)$,
and deserve own studies.
Thus,
we leave the analysis of RHP~\ref{rhp:soliton-gas:inftyDSG_jump_form} for future studies.

\subsection{Results for the focusing NLS and complex mKdV equations}

We demonstrate here how to easily generalize obtained results for MBEs~\eqref{e:mbe} to other integrable systems. For simplicity, we only consider two systems in the same non-self-adjoint Zakharov-Shabat hierarchy, the focusing NLS equation
\begin{equation}
\label{e:nls}
iq_t(x,t) + q_{xx}(x,t) + 2|q(x,t)|^2q(x,t) = 0\,,
\end{equation}
with its Lax pair
\begin{equation}
\label{e:nls-lax-pair}
\begin{aligned}
\bphi_x & = (\ii\lambda\sigma_3 + \Q)\bphi\,,\\
\bphi_t & = [-2\ii\lambda^2\sigma_3 - 2\lambda \Q + \ii\sigma_3(\Q_x - \Q^2)]\bphi\,,
\end{aligned}
\end{equation}
and the complex mKdV equation
\begin{equation}
\label{e:mkdv}
q_t(x,t) + q_{xxx}(x,t) + 6|q(x,t)|^2q_x(x,t) = 0\,,
\end{equation}
with its Lax pair
\begin{equation}
\label{e:mkdv-lax-pair}
\begin{aligned}
\bphi_x & = (\ii\lambda\sigma_3 + \Q)\bphi\,,\\
\bphi_t & = [4\ii\lambda^3\sigma_3 + 4\lambda^2\Q + 2\ii\lambda(\Q^2 + \Q_x)\sigma_3 + 2\Q^3 + \Q_x\Q - \Q\Q_x - \Q_{xx}]\bphi\,.
\end{aligned}
\end{equation}
In both Lax pairs~\eqref{e:nls-lax-pair} and~\eqref{e:mkdv-lax-pair}, the matrix $\Q(x,t)$ is identical to that in Equation~\eqref{e:laxpair}, but with $q = q(x,t)$. Clearly, all three systems [MBEs~\eqref{e:mbe}, NLS~\eqref{e:nls} and mKdV~\eqref{e:mkdv}] have identical scattering problem.

Upon reviewing all soliton results for the MBEs in previous subsections, one notices that there are only two system-specific components, the structure of the $N$-DSG eigenvalue set $\Lambda_j$ from Definition~\ref{def:LambdaOmega} and the dispersion relation $\theta(\lambda;t,z)$ given in Equation~\eqref{e:theta-def}. Let us rewrite them
\begin{equation}
\theta_{\mbe}(\lambda;t,z) \coloneqq \lambda t - \frac{D_-}{2\lambda}z\,,\qquad
\Lambda_{\mbe,j} \coloneqq \set{\lambda_{j,k}}{|\lambda_{j,k}| = \mathrm{const}}\,.
\end{equation}
Note that eigenvalues with identical soliton velocities are determined by the equation $\Im(\theta) = 0$ for all $t\in\Real$ with $x = V t$ and $V$ being the fixed velocity (cf. the remark below Theorem~\ref{thm:N-DSG}). The same condition holds true for other integrable systems, so one gets
\begin{equation}
\begin{aligned}
\theta_\nls(\lambda;x,t) & \coloneqq \lambda x - 2\lambda^2 t\,,\quad&
\Lambda_{\nls,j} & \coloneqq \set{\lambda_{j,k}}{\Re(\lambda_{j,k}) = \mathrm{const}}\,,\\
\theta_\mkdv(\lambda;x,t) & \coloneqq \lambda x + 4\lambda^3 t\,,\quad&
\Lambda_{\mkdv,j} & \coloneqq \set{\lambda_{j,k}}{3\Re(\lambda_{j,k})^2 - \Im(\lambda_{j,k})^2 = \mathrm{const}}\,.
\end{aligned}
\end{equation}
Consequently, the DSG velocities are given by
\begin{equation}
V_\nls \coloneqq 4\Re(\lambda_{j,k})\,,\qquad
V_\mkdv \coloneqq 4\Im(\lambda_{j,k})^2 - 12\Re(\lambda_{j,k})^2\,.
\end{equation}
Also, the eigenvalue groups $\Lambda_j$ are ordered according to DSG velocities in the increasing order.

Now, simply substituting $\theta_\nls$ and $\Lambda_{\nls,j}$ into all results for the MBEs, one gets the corresponding soliton result for the focusing NLS equation. For example, replacing $\theta_\mbe$ by $\theta_\nls$ in Theorem~\ref{thm:N-soliton-formula}, one obtains the general $N$-soliton solution formula for the focusing NLS equation. Replacing $\Lambda_\mbe$ by $\Lambda_\nls$ in Theorem~\ref{thm:N-DSG} and substituting $\theta_\nls$ into Equation~\eqref{e:DSG-center-equation} yield the localization and center formula of $N$-DSGs for the NLS equation.
Same is true for the complex mKdV equation as well. The results are
\begin{equation}
\label{e:nls-mkdv-center}
\begin{aligned}
x_{\nls,\c}(t)
 & \coloneqq V_\nls t + x_{\d}\,,\\
x_{\mkdv,\c}(t)
 & \coloneqq V_\mkdv t + x_{\d}\,,\\
x_{\d}
 & \coloneqq -\frac{1}{2\sum_{k=1}^{N_1}\Im(\lambda_{1,k})}\ln\bigg(\prod_{k=1}^{N_1} |\omega_{1,k}|\cdot \frac{\prod_{k = 2}^{N_1}\prod_{l = 1}^{k-1}|\lambda_{1,k} - \lambda_{1,l}|^2}{\prod_{k=1}^{N_1}\prod_{l=1}^{N_1}|\conj{\lambda_{1,k}} - \lambda_{1,l}|}\bigg)\,.
\end{aligned}
\end{equation}
Note that the displacement $x_\d$ for both systems are formally identical.

One can obtain the general soliton asymptotics for the NLS equation as well,
by replacing the eigenvalue sets in Theorem~\ref{thm:soliton-asymptotics} and Corollary~\ref{thm:asymptotic-shifts}. Again, the same procedure yields soliton asymptotics for the complex mKdV equation. The final asymptotics for these two systems are
\begin{equation}
\label{e:nls-mkdv-asymptotics}
\begin{aligned}
q_\nls(x,t;\Lambda,\Omega)
 & = \sum_{j = 1}^J q_\nls(x,t;\Lambda_j,\Omega_j^{(\pm)}) + \myblue{\O(\ee^{-\aleph|t|})}\,,\qquad t\to\pm\infty\,,\\
q_\mkdv(x,t;\Lambda,\Omega)
 & = \sum_{j = 1}^J q_\mkdv(x,t;\Lambda_j,\Omega_j^{(\pm)}) + \myblue{\O(\ee^{-\aleph|t|})}\,,\qquad t\to\pm\infty\,,
\end{aligned}
\end{equation}
where \myblue{$\aleph$ is a positive constant},
and for both systems the modified norming constants are given by
\begin{equation}
\omega_{j,k}^{(+)}
 \coloneqq \omega_{j,k}\delta_j(\lambda_{j,k})^2\delta_1(\lambda_{j,k})^{-2}\,,\qquad
\omega_{j,k}^{(-)}
 \coloneqq \omega_{j,k}\delta_{j+1}(\lambda_{j,k})^{-2}\,,\qquad
 k = 1,2,\dots, N_j\,.
\end{equation}
\begin{remark}
It is important to point out that the results for the focusing NLS equation and the mKdV equation are formally identical, but are slightly different from the ones for the MBEs as shown in Theorems~\ref{thm:N-DSG} and~\ref{thm:soliton-asymptotics}. The reason is that the spatial and temporal variables are swapped between the MBEs and NLS/mKdV equations.
\end{remark}
One is able to obtain the $N$th-order solitons and soliton gases for the focusing NLS equation and the complex mKdV equation, again, by replacing $\theta$ and $\Lambda$ in Theorem~\ref{thm:N-order-soliton:solution-formula} and RHP~\ref{rhp:soliton-gas:inftyDSG_jump_form}. The explicit results are omitted for brevity.

Finally, for demonstrative purposes, we present figures of exact solutions and our theoretical results for the focusing NLS and mKdV equations. Figure~\ref{f:nls} contains results for the focusing NLS equation~\eqref{e:nls}: (Left) an exact $3$-DSG solution by replacing $\theta_\mbe$ with $\theta_\nls$ and proper $\Lambda$ in Theorem~\ref{thm:N-soliton-formula}, with its center shown as the red dashed line from Equation~\eqref{e:nls-mkdv-center}; (Center) an exact $5$-soliton solution with three DSGs by replacing $\theta_\mbe$ with $\theta_\nls$ and proper $\Lambda$ in Theorem~\ref{thm:N-soliton-formula}, with the asymptotic center shown as the red dashed lines from Equations~\eqref{e:nls-mkdv-center} and~\eqref{e:nls-mkdv-asymptotics}; (Right) an exact $4$th-order soliton solution by replacing $\theta_\mbe$ with $\theta_\nls$ in Theorem~\ref{thm:N-order-soliton:solution-formula}. Similarly, Figure~\ref{f:mkdv} contains results for the complex mKdV equation~\eqref{e:mkdv}: (Left) an exact $3$-DSG solution by replacing $\theta_\mbe$ with $\theta_\mkdv$ and proper $\Lambda$ in Theorem~\ref{thm:N-soliton-formula}, with its center shown as the red dashed line from Equation~\eqref{e:nls-mkdv-center}; (Center) an exact $4$-soliton solution with two DSGs by replacing $\theta_\mbe$ with $\theta_\mkdv$ and proper $\Lambda$ in Theorem~\ref{thm:N-soliton-formula}, with the asymptotic center shown as the red dashed lines from Equations~\eqref{e:nls-mkdv-center} and~\eqref{e:nls-mkdv-asymptotics}; (Right) an exact $4$th-order soliton by substituting $\theta_\mkdv$ in Theorem~\ref{thm:N-order-soliton:solution-formula}.

\begin{figure}[tp]
\centering
\includegraphics[width=0.31\textwidth]{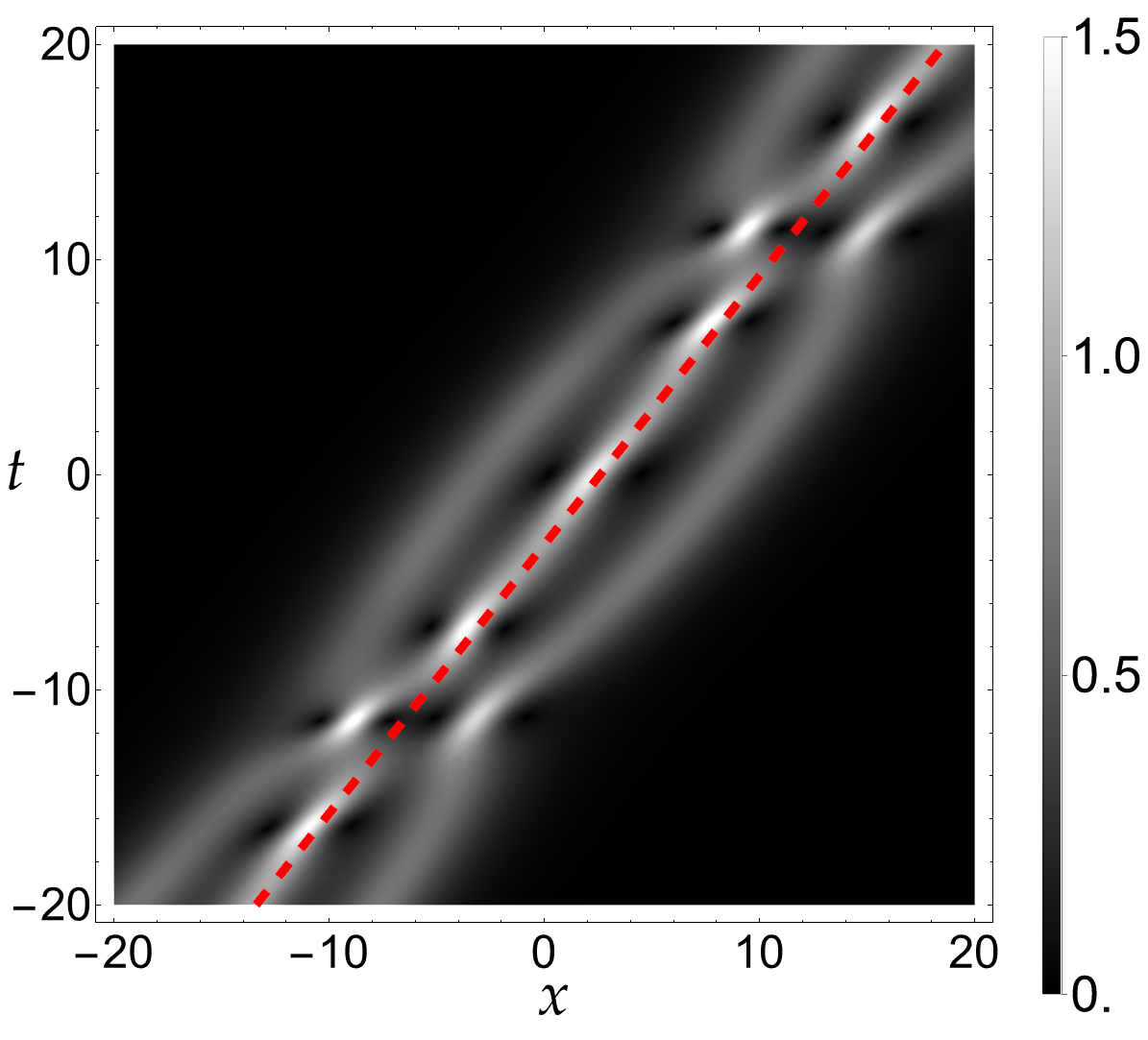}\quad
\includegraphics[width=0.31\textwidth]{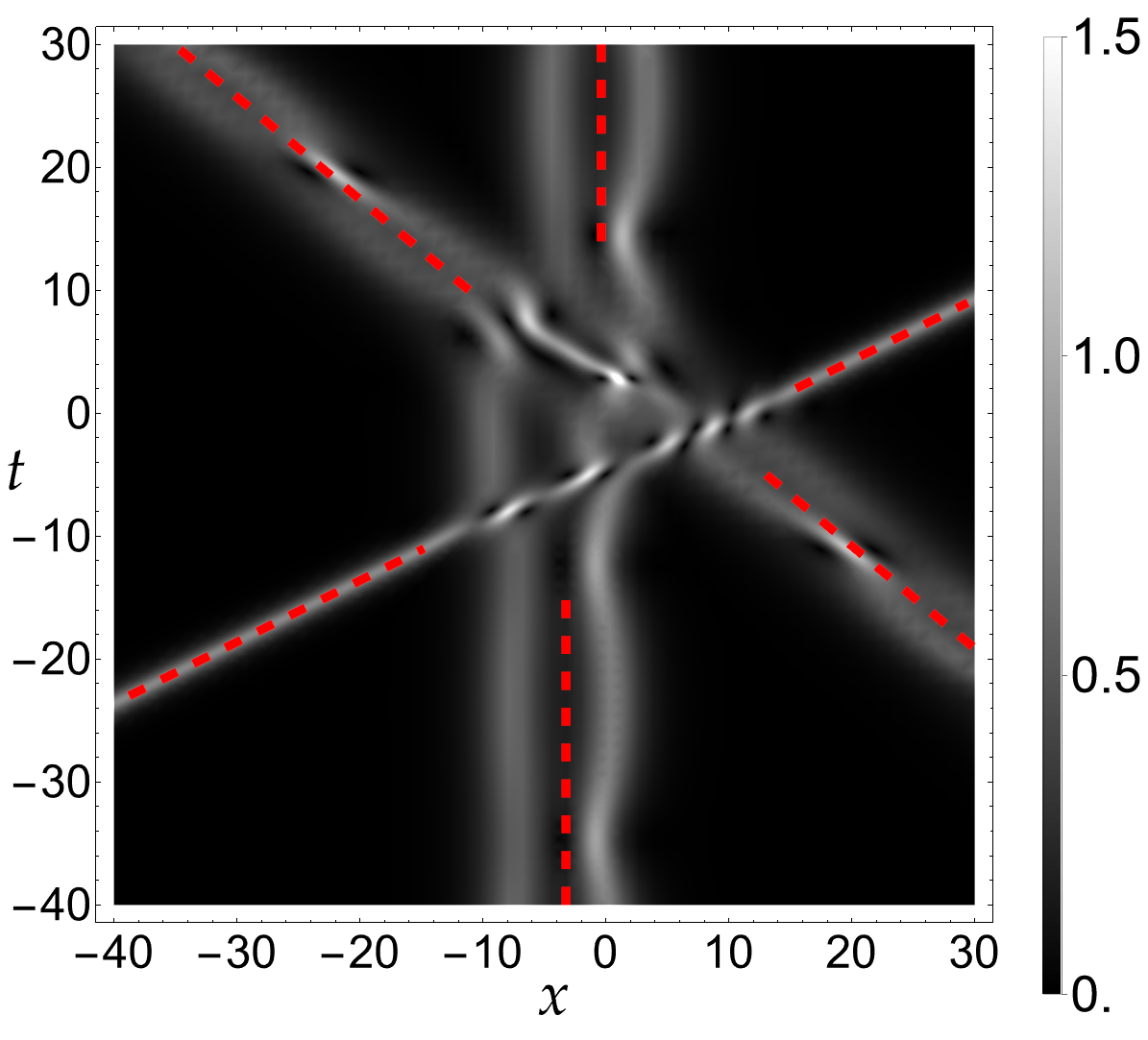}\quad
\includegraphics[width=0.3\textwidth]{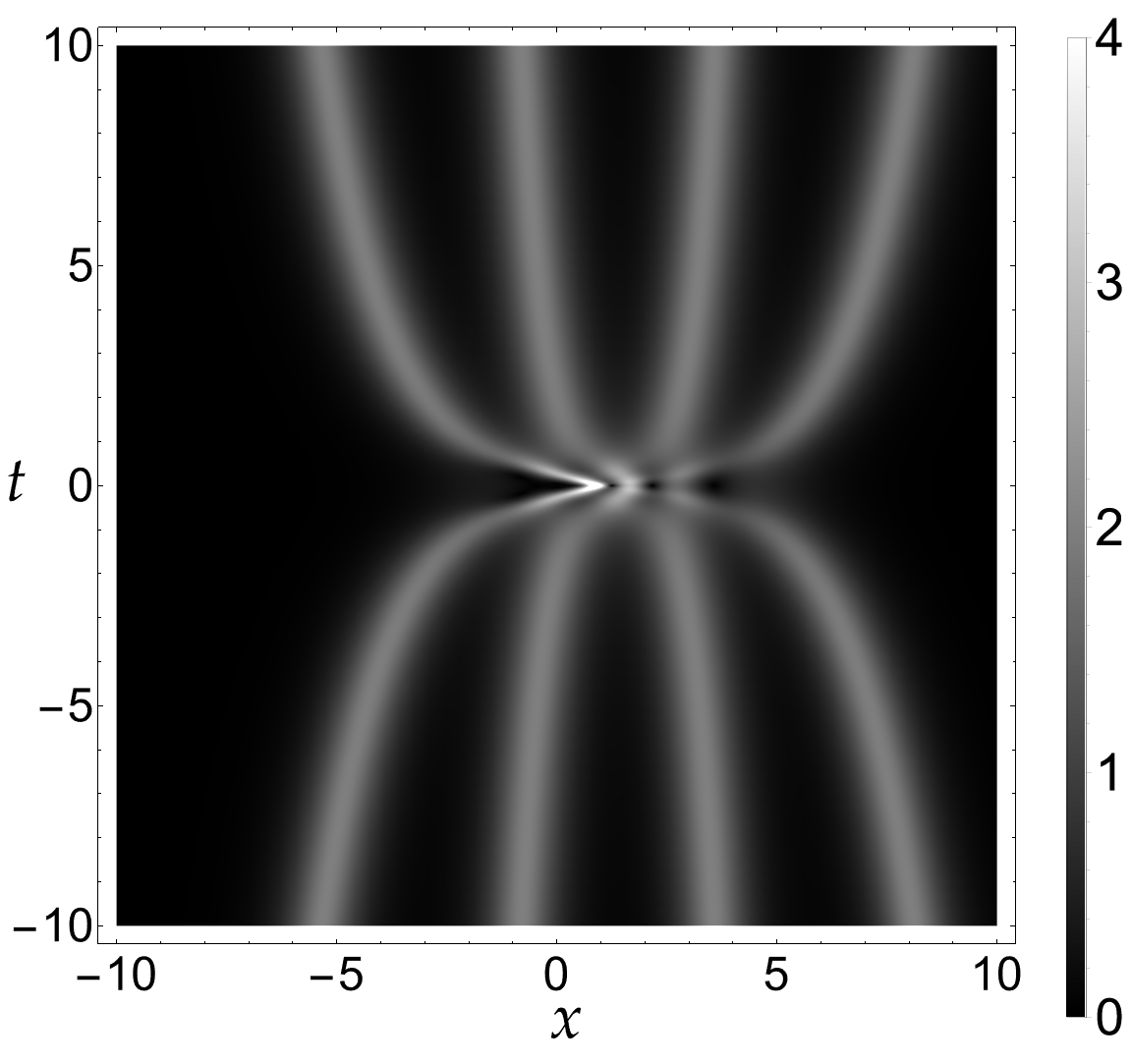}
\caption{
    Density plots of exact solutions $|q(x,t)|$ of the focusing NLS equation~\eqref{e:nls} from Theorems~\ref{thm:N-soliton-formula} and~\ref{thm:N-order-soliton:solution-formula}.
    Left: A $3$-DSG with parameters $\lambda_{1,1} = 0.2+0.6\ii$, $\lambda_{1,2} = 0.2+0.4\ii$, $\lambda_{1,3} = 0.2+0.3\ii$, $\omega_{1,1} = -1$, and $\omega_{1,2} = \lambda_{1,3} = \ee$. The red dashed line is the DSG center from Equation~\eqref{e:nls-mkdv-center}.
    Center: A $5$-soliton solution with three DSGs and parameters $\lambda_{1,1} = -0.3+0.2\ii$, $\lambda_{1,2} = -0.3+0.4\ii$, $\lambda_{2,1} = 0.3\ii$, $\lambda_{2,2} = 0.4\ii$, $\lambda_{3,1} = 0.5+0.4\ii$, $\omega_{1,1} = \omega_{1,2} = 1$, $\omega_{2,1} = \omega_{2,2} = \ee^5$, and $\omega_{3,1} = \ee^{-6}$. The red dashed lines are the asymptotic DSG centers from Equations~\eqref{e:nls-mkdv-center} and~\eqref{e:nls-mkdv-asymptotics}.
    Right: A $4$th-order soliton solution with $\lambda_\circ = \ii$,  $\omega_{\circ,0} = 1$, and $\omega_{\circ,k} = 0$ for $1\le k\le 3$.
}
\label{f:nls}
\vspace*{2ex}
\centering
\includegraphics[width=0.31\textwidth]{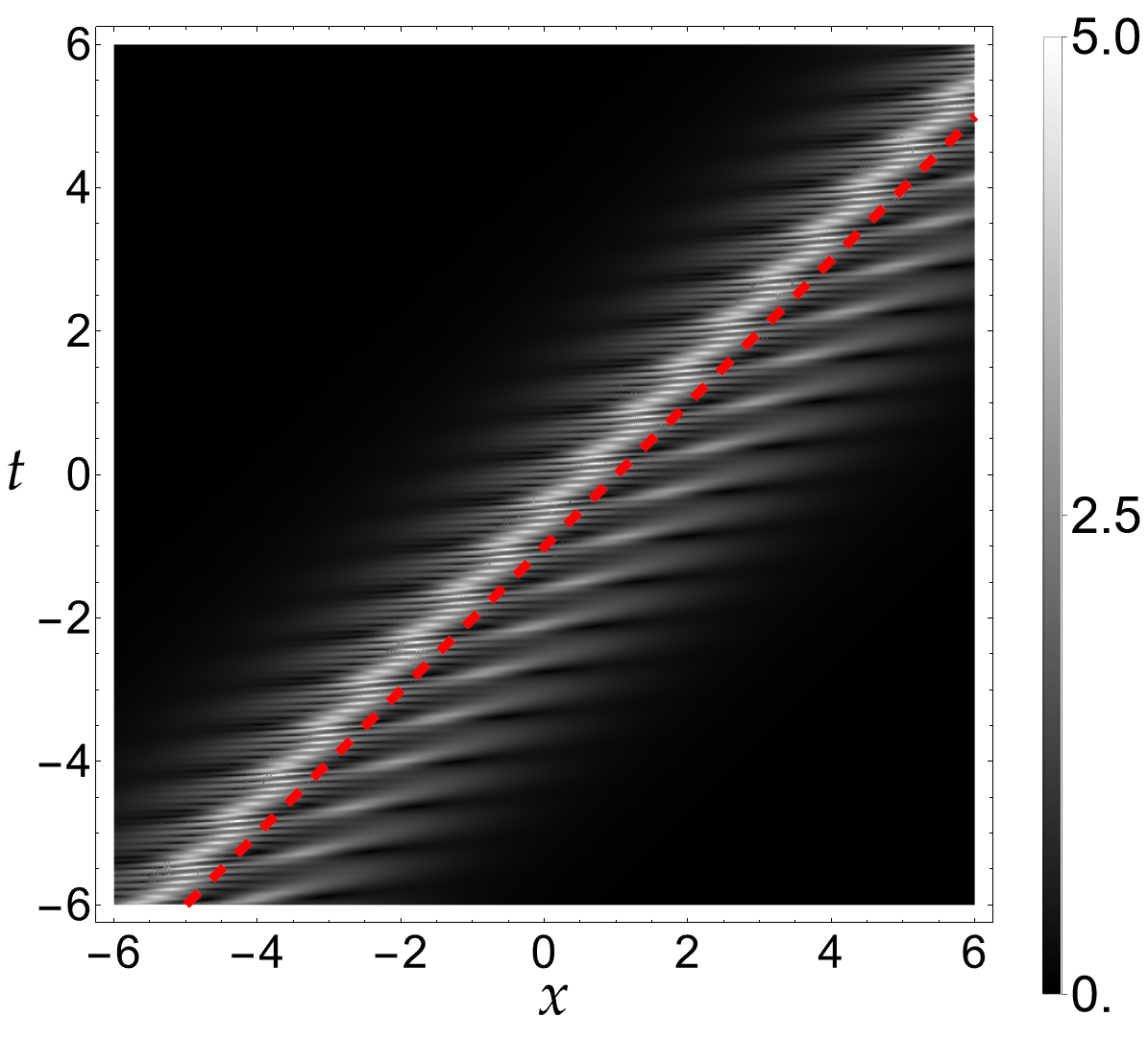}\quad
\includegraphics[width=0.3\textwidth]{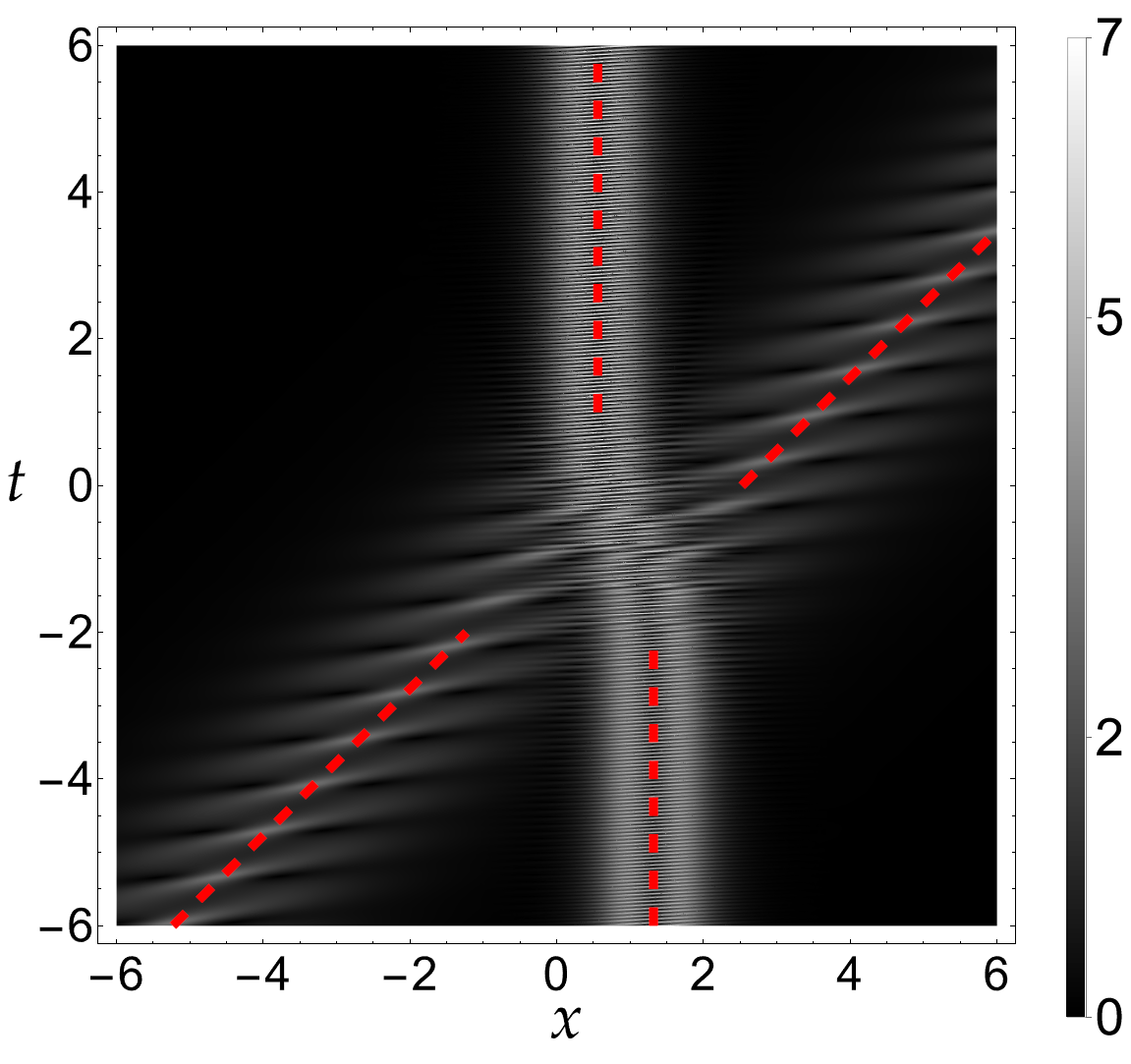}\quad
\includegraphics[width=0.3\textwidth]{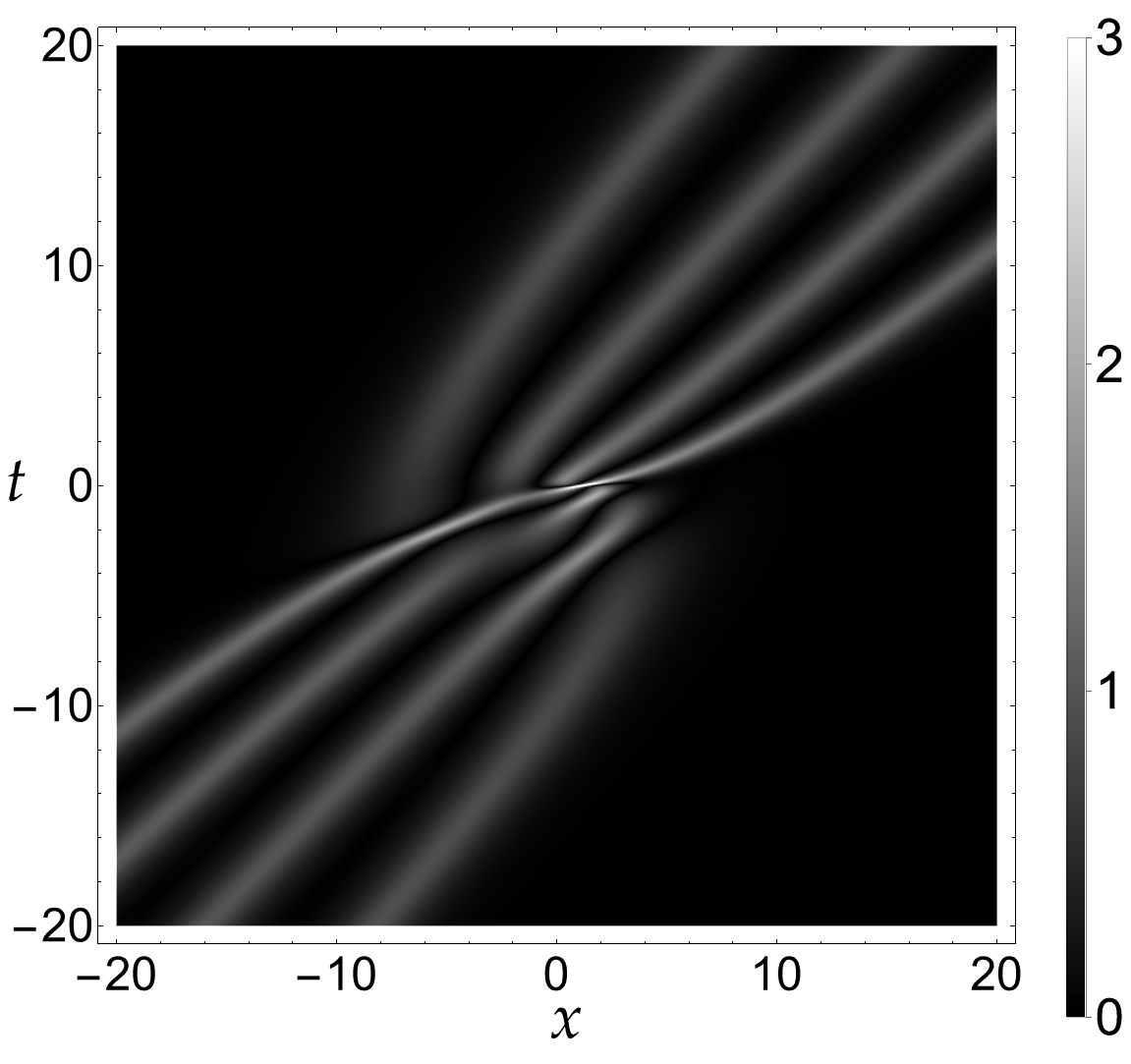}
\caption{
    Density plots of exact solutions $|q(x,t)|$ of the complex mKdV equation~\eqref{e:mkdv} from Theorems~\ref{thm:N-soliton-formula} and~\ref{thm:N-order-soliton:solution-formula}.
    Left: A $3$-DSG with parameters $\lambda_{1,1} = \ii/2$, $\lambda_{1,2} = 1/2+\ii$, $\lambda_{1,3} = 1 + \ii\sqrt{13}/2$, and $\omega_{1,1} = \omega_{1,2} = \omega_{1,3} = 0$. The red dashed line denotes the DSG center from Equation~\eqref{e:nls-mkdv-center}.
    Center: A $4$-soliton solution with two DSGs and parameters $\lambda_{1,1} = -1 + \ii\sqrt3$, $\lambda_{1,2} = 1 + \ii\sqrt3$, $\lambda_{2,1} = -1/2 + \ii$, $\lambda_{2,2} = \ii/2$, and $\omega_{j,k} = 2$ for all $j,k = 1,2$. The red dashed lines denote the asymptotic DSG centers from Equations~\eqref{e:nls-mkdv-center} and~\eqref{e:nls-mkdv-asymptotics}.
    Right: A $4$th-order soliton solution with $\lambda_\circ = \ii/2$, $\omega_{\circ,0} = 1$, and $\omega_{\circ,k} = 0$ for $1\le k\le 3$.
}
\label{f:mkdv}
\end{figure}

\begin{remark}
It is demonstrated that one is able to expand the discussion in this section without much effort, and obtain exact solutions and asymptotic results to even higher order systems in the focusing Zakharov-Shabat AKNS hierarchy as far as one wishes. The generalization to other hierarchies on DSGs/soliton asymptotics/high-order solitons/soliton gases can also be achieved following the framework.
\end{remark}

\section{Reconstruction and soliton solution {formul\ae}}
\label{s:reconstruction-soiton-solution-formula}

\subsection{Proof of Lemma~\ref{thm:reconstruction}}
\label{s:proof-reconstruction-formula}

We prove Lemma~\ref{thm:reconstruction} in three steps:
\begin{enumerate}
\item
We first derive the reconstruction formula from RHP~\ref{rhp:N-soliton-jump-form} via the dressing method.
\item
We then show the equivalence between RHPs~\ref{rhp:N-soliton-jump-form} and~\ref{rhp:N-soliton-residue-form}, i.e., both RHPs generate identical solutions of MBEs.
\item
We finally show that the MBEs solutions obtained from RHP~\ref{rhp:N-soliton-jump-form} exhibit the correct boundary condition $q(t,z)\to0$, $P(t,z)\to0$ and $D(t,z)\to D_-$ as $t\to-\infty$.
\end{enumerate}

For statement (1):
given RHP~\ref{rhp:N-soliton-jump-form},
one defines a new matrix function as
$\N(\lambda;t,z) \coloneqq \M(\lambda;t,z)\ee^{\ii\theta(\lambda;t,z)\sigma_3}$,
which is analytic for
\myblue{
$\lambda \in \Complex \setminus (\{0\}\bigcup\{\partial D_{\lambda_{j,k}}^\epsilon\bigcup\partial D_{\conj{\lambda_{j,k}}}^\epsilon\}_{\lambda_{j,k}\in\Lambda})$,
}
and which has the asymptotic behavior that
\myblue{
$\N(\lambda;t,z)\ee^{-\ii\theta(\lambda;t,z) \sigma_3} = \I + \lambda^{-1}\M_1(t,z) + o(\lambda^{-1})$
}
as $\lambda\to\infty$ with
\myblue{
$\M(\lambda;t,z) = \I + \lambda^{-1}\M_1(t,z) + \O(\lambda^{-2})$.
}
It is easy to check that the jump condition for $\N(\lambda;t,z)$ on all small circles are independent of $(t,z)\in\Real^2$.
Therefore, the matrix function $\N_t(\lambda;t,z)\N(\lambda;t,z)^{-1}$
is analytic in the whole complex plane with the possible exception at the origin.
One easily checks that the singularity at zero is removable.
Using the asymptotic behavior as $\lambda\to\infty$,
\myblue{Liouville's theorem}
shows that $\N_t(\lambda;t,z)\N(\lambda;t,z)^{-1}$
is a linear function of the form
\myblue{
$\ii \lambda\sigma_3 + \ii[\M_1(t,z),\sigma_3]$.
Defining a matrix
$\Q(t,z) \coloneqq \lim_{\lambda\to\infty}\ii\lambda[\M(\lambda;t,z),\sigma_3]$,
one obtains the first relation
$\N_t(\lambda;t,z) = (\ii\lambda\sigma_3 + \Q(t,z))\N(\lambda;t,z)$.
}
Similarly,
the matrix function $\N_z(\lambda;t,z)\N(\lambda;t,z)^{-1}$
is analytic except for a simple pole at the origin,
and it vanishes as $\lambda\to\infty$.
Therefore,
\myblue{Liouville's theorem}
yields that the product has the form
\myblue{
$-\ii D_-\M(0;t,z)\sigma_3\M(0;t,z)^{-1}/(2\lambda)$.
Defining a matrix $\brho(t,z) \coloneqq D_-\M(0;t,z)\sigma_3\M(0;t,z)^{-1}$,
one obtains the second relation
$\N_z(\lambda;t,z) = -\ii\brho(t,z)/(2\lambda)\N(\lambda;t,z)$.
Note that the overdetermined system of both relations is nothing else but the Lax pair in Equation~\eqref{e:laxpair}.
The compatibility condition $\N_{tz}(\lambda;t,z) = \N_{zt}(\lambda;t,z)$ requires that
$\Q(t,z)$ and $\brho(t,z)$ solve the matrix system
$\Q_z(t,z) = \frac12[\brho(t,z),\sigma_3]$ and
$\brho_t(t,z) = [\Q(t,z),\brho(t,z)]$,
whose component form is identically the MBEs~\eqref{e:mbe}.
Thus, we have shown that
the solution $\M(\lambda;t,z)$ of a given RHP~\ref{rhp:N-soliton-jump-form}
can produce a solution to MBEs~\eqref{e:mbe}.
}

For statement (2): Let $\M(\lambda;t,z)$ be the solution of RHP~\ref{rhp:N-soliton-jump-form}. One can define a new matrix function as \begin{equation}
\label{e:reconstruction-soiton-solution-formula:equivlent-RHPs-definition}
\N(\lambda;t,z) \coloneqq \begin{cases}
\M(\lambda;t,z)\V_{j,k}(\lambda;t,z)\,, & \lambda\in D_{\lambda_{j,k}}^\epsilon\,,\\
\M(\lambda;t,z)\V_{j,k}(\conj\lambda;t,z)^{-\dagger}\,, & \lambda\in D_{\conj{\lambda_{j,k}}}^\epsilon\,,\\
\M(\lambda;t,z)\,, & \mbox{otherwise}\,.
\end{cases}
\end{equation}
It is easily checked that $\N(\lambda;t,z)$ does not admit any jumps, and has simple poles at each eigenvalue point and its complex conjugate. The pole condition of $\N(\lambda;t,z)$ is \myblue{exactly} the ones in RHP~\ref{rhp:N-soliton-residue-form}. Therefore, $\N(\lambda;t,z)$ is a solution of RHP~\ref{rhp:N-soliton-residue-form}. Similarly, by taking a solution RHP~\ref{rhp:N-soliton-residue-form} and reverse the definition of Equation~\eqref{e:reconstruction-soiton-solution-formula:equivlent-RHPs-definition}, one gets a solution of RHP~\ref{rhp:N-soliton-jump-form}. In other words, the solutions of the two RHPs have one-to-one correspondence.
The uniqueness of solutions of RHP~\ref{rhp:N-soliton-jump-form} yields the uniqueness of RHP~\ref{rhp:N-soliton-residue-form}. Because the two RHP solutions are identical for $|\lambda|\gg1$ and for $|\lambda|\ll1$ in the definition~\eqref{e:reconstruction-soiton-solution-formula:equivlent-RHPs-definition} with sufficiently small $\epsilon > 0$, both yield identical solutions to MBEs~\eqref{e:mbe} via the reconstruction formula from Lemma~\ref{thm:reconstruction}. Thus, RHPs~\ref{rhp:N-soliton-residue-form} and~\ref{rhp:N-soliton-jump-form} are equivalent.

For statement (3):
We apply the nonlinear steepest descent method to RHP~\ref{rhp:N-soliton-jump-form} as
$t\to-\infty$
to calculate the asymptotics for the solutions of MBEs for all $z\ge0$.
The inequality $t < 0$ yields $\Re(\ii\lambda t) = -\Im(\lambda)t > 0$ for $\lambda$ in the upper half plane.
\myblue{
Consequently,
one obtains that
$\ee^{-2\ii\theta(\lambda;t,z)} = \O(\ee^{-\aleph |t|}) \to 0$
uniformly as $t\to-\infty$ with a positive constant $\aleph$,
for all $\lambda\in\partial D_{\lambda_{j,k}}^\epsilon$.
Therefore,
all jumps in the upper half plane of RHP~\ref{rhp:N-soliton-jump-form}
admit uniform limit $\|\V_{j,k}(\lambda;t,z)^{-1}-\I\| = \O(\ee^{-\aleph |t|})\to0$ as $t\to-\infty$.
Similar argument shows that all jumps in the lower half plane admit uniform limit
$\|\V_{j,k}(\conj\lambda;t,z)^\dagger-\I\| = \O(\ee^{-\aleph |t|})\to0$
as $t\to-\infty$ as well.
As RHP~\ref{rhp:N-soliton-jump-form} becoming a small-norm problem,
one concludes that $\M(\lambda;t,z) = \I + \O(\ee^{-\aleph |t|}) \to0$ as $t\to-\infty$,
}
yielding the desired boundary conditions in the Cauchy problem~\eqref{e:mbe} by the reconstruction formula from Lemma~\ref{thm:reconstruction}.

\subsection{Proof of Theorem~\ref{thm:N-soliton-formula}}
\label{s:derivation-Nsoliton}

In this section, we suppress the $(t,z)$ parametric dependence in all relative quantities for simplicity. In particular, we write $\M(\lambda;t,z) = \M(\lambda)$.
We start from RHP~\ref{rhp:N-soliton-residue-form}.
The solution can be written as
\begin{equation}
\M(\lambda)
 = \I + \sum_{\lambda_{j,k}\in\Lambda}\frac{\Res_{\lambda = \lambda_{j,k}}\M(\lambda)}{\lambda-\lambda_{j,k}} + \sum_{\lambda_{j,k}\in\Lambda}\frac{\Res_{\lambda = \conj{\lambda_{j,k}}}\M(\lambda)}{\lambda-\conj{\lambda_{j,k}}}\,.
\end{equation}
Again, using the residue conditions in RHP~\ref{rhp:N-soliton-residue-form}, the above equation can be rewritten as
\begin{equation}
\M(\lambda)
 = \I +\sum_{\lambda_{j,k}\in\Lambda}\frac{1}{\lambda-\lambda_{j,k}}\lim_{\lambda\to\lambda_{j,k}}\M(\lambda)\bpm 0 & 0 \\ \omega_{j,k}\ee^{-2\ii\theta(\lambda)} & 0 \epm
+ \sum_{\lambda_{j,k}\in\Lambda}\frac{1}{\lambda - \conj{\lambda_{j,k}}}\lim_{\lambda\to\conj{\lambda_{j,k}}}\M(\lambda)\bpm 0 & - \conj{\omega_{j,k}}\ee^{2\ii\theta(\conj{\lambda})t} \\ 0 & 0 \epm\,.
\end{equation}
Looking at the first and second columns and denoting them as $\M_1(\lambda)$ and $\M_2(\lambda)$, respectively, one has
\begin{equation}
\label{e:M-sol}
\begin{aligned}
\M_1(\lambda)
 = \bpm 1 \\ 0 \epm + \sum_{\lambda_{j,k}\in\Lambda} \frac{\omega_{j,k}\ee^{-2\ii\theta_{j,k}}}{\lambda-\lambda_{j,k}}\M_2(\lambda_{j,k})\,,\qquad
\M_2(\lambda)
 = \bpm 0 \\ 1 \epm - \sum_{\lambda_{j,k}\in\Lambda} \frac{\conj{\omega_{j,k}}\ee^{2\ii\conj{\theta_{j,k}}}}{\lambda - \conj{\lambda_{j,k}}}\M_1(\conj{\lambda_{j,k}})\,,
\end{aligned}
\end{equation}
where we recall the definition $\theta_{j,k}\coloneqq \theta(\lambda_{j,k})$ from Theorem~\ref{thm:N-soliton-formula}.
Recall the \myblue{discussion} in Remark~\ref{thm:M-symmetry} on the symmetries of entries of $\M(\lambda)$. Hence, it is sufficient to look at the first row of Equation~\eqref{e:M-sol}. Moreover, by substituting $\lambda = \conj{\lambda_{j,k}}$ in $M_{1,1}(\lambda)$ and $\lambda = \lambda_{j,k}$ in $M_{1,2}(\lambda)$, respectively, for $k = 1,2,\dots,N_j$ and $j = 1,2,\dots,J$, one obtains the following $N\times N$ linear system,
\begin{equation}
\label{e:M1112-system}
\begin{aligned}
M_{1,1}(\conj{\lambda_{j,k}})
 & = 1 + \sum_{\lambda_{j',k'}\in\Lambda}c_{j',k'}(\conj{\lambda_{j,k}})M_{1,2}(\lambda_{j',k'})\,,\\
M_{1,2}(\lambda_{j,k})
 & = - \sum_{\lambda_{j',k'}\in\Lambda}c_{j',k'}^*(\lambda_{j,k})M_{1,1}(\conj{\lambda_{j',k'}})\,,
\end{aligned}
\end{equation}
where
\myblue{
the quantity $c_{j,k}(\lambda)$ is defined in Equation~\eqref{e:C-Gamma-definition},
}
and $k = 1,2,\dots,N_j$ and $j = 1,2,\dots,J$.
In order to rewrite the above linear system in matrix form, one defines vectors of unknowns
\begin{equation}
\begin{aligned}
\overrightarrow{M_{1,1}(\conj{\lambda_{j,k}})}
 & \coloneqq \bpm M_{1,1}(\conj{\lambda_{1,1}}) & \cdots & M_{1,1}(\conj{\lambda_{1,N_1}}) & \cdots & M_{1,1}(\conj{\lambda_{J,1}}) & \cdots & M_{1,1}(\conj{\lambda_{J,N_J}}) \epm^\top\,,\\
\overrightarrow{M_{1,2}(\lambda_{j,k})}
 & \coloneqq \bpm M_{1,2}(\lambda_{1,1}) & \cdots & M_{1,2}(\lambda_{1,N_1}) & \cdots & M_{1,2}(\lambda_{J,1}) & \cdots & M_{1,2}(\lambda_{J,N_J}) \epm^\top\,.
\end{aligned}
\end{equation}
As a result, the linear system~\eqref{e:M1112-system} becomes
\begin{equation}
\overrightarrow{M_{1,1}(\conj{\lambda_{j,k}})} = \vbone + \bGamma \overrightarrow{M_{1,2}(\lambda_{j,k})}\,,\qquad
\overrightarrow{M_{1,2}(\lambda_{j,k})} = -\conj{\bGamma} \overrightarrow{M_{1,1}(\conj{\lambda_{j,k}})}\,,
\end{equation}
where the vector $\vbone$ and the matrix $\bGamma$ are defined in Equation~\eqref{e:C-Gamma-definition}.
Cramer's rule immediately yields the solutions to the above system
\begin{equation}
M_{1,1}(\conj{\lambda_{j,k}})
 = \frac{\det(\I + \bGamma\conj{\bGamma})_{j,k}}{\det(\I + \bGamma\conj{\bGamma})}\,,\qquad
M_{1,2}(\lambda_{j,k})
 = -\frac{\det(\conj{\bGamma^{-1}} + \bGamma)_{j,k}}{\det(\conj{\bGamma^{-1}} + \bGamma)}\,,
\end{equation}
where $\det(\A)_{j,k}$ denotes the determinant of the matrix that is obtained by replacing the $(N_1+\dots+N_{j-1}+k)$-th column of the matrix $\A$ by $\vbone$.

Combining the above quantities with the first row of Equation~\eqref{e:M-sol}, one has
\begin{equation}
\label{e:M-sol1}
\begin{aligned}
M_{1,1}(\lambda;t,z)
 & = 1 + \frac{\det\bpm 0 & \C \\ \vbone & \conj{\bGamma^{-1}} + \bGamma \epm}{\det(\conj{\bGamma^{-1}} + \bGamma)}\\
 & = 1 + \frac{\det(\conj{\bGamma^{-1}} + \bGamma - \myblue{\vbone\C}) - \det(\conj{\bGamma^{-1}} + \bGamma)}{\det(\conj{\bGamma^{-1}} + \bGamma)}\\
 & = \frac{\det(\conj{\bGamma^{-1}} + \bGamma - \myblue{\vbone\C})}{\det(\conj{\bGamma^{-1}} + \bGamma)}\\
 & = \frac{\det(\I - \myblue{\vbone\C\conj{\bGamma}} + \bGamma\conj{\bGamma})}{\det(\I + \bGamma\conj{\bGamma})} \\
M_{1,2}(\lambda;t,z)
 & = \frac{\det\bpm 0 & \conj{\C} \\ \vbone & \I + \bGamma\conj{\bGamma}\epm}{\det(\I + \bGamma\conj{\bGamma})}
 = \frac{\det(\I - \myblue{\vbone\conj{\C}} + \bGamma\conj{\bGamma})}{\det(\I + \bGamma\conj{\bGamma})} - 1\,,
\end{aligned}
\end{equation}
where $\C$ is defined in Theorem~\ref{thm:N-soliton-formula}.
The reconstruction formula~\eqref{e:reconstructon-matrix} yields the explicit {formul\ae} for the MBEs solutions presented in Theorem~\ref{thm:N-soliton-formula}.

\myblue{
We next prove $\det(\I + \bGamma\conj{\bGamma})\ne0$ for all $(t,z)\in\Real^2$,
following a similar calculation for the NLS equation but for high-order solitons~\cite{lsg2024}.
First, we factor the matrix $\bGamma$,
\begin{equation}
\everymath{\displaystyle}
\begin{aligned}
\bGamma
 & = \bGamma^{(1)}\bGamma^{(2)}(t,z)\,,\\
\bGamma^{(1)}
 & \coloneqq \bpm
    \frac{1}{\conj{\lambda_{1,1}} - \lambda_{1,1}} & \cdots & \frac{1}{\conj{\lambda_{1,1}} - \lambda_{1,N_1}} & \dots & \frac{1}{\conj{\lambda_{1,1}} - \lambda_{J,1}} & \dots & \frac{1}{\conj{\lambda_{1,1}} - \lambda_{J,N_J}} \\
    \vdots & & \vdots & & \vdots & & \vdots \\
    \frac{1}{\conj{\lambda_{J,N_J}} - \lambda_{1,1}} & \cdots & \frac{1}{\conj{\lambda_{J,N_J}} - \lambda_{1,N_1}} & \dots & \frac{1}{\conj{\lambda_{J,N_J}} - \lambda_{1,N_J}} & \dots & \frac{1}{\conj{\lambda_{J,N_J}} - \lambda_{J,N_J}} \epm\,,\\
\bGamma^{(2)}(t,z)
 & \coloneqq \diag(
    \omega_{1,1}\ee^{-2\ii\theta_{1,1}},
    \cdots,
    \omega_{1,N_1}\ee^{-2\ii\theta_{1,N_1}},
    \cdots,
    \omega_{J,1}\ee^{-2\ii\theta_{J,1}},
    \cdots,
    \omega_{J,N_J}\ee^{-2\ii\theta_{J,N_J}}
    )\,.
\end{aligned}
\end{equation}
where we recall $\theta_{j,k} = \theta(\lambda_{j,k};t,z)$.
\begin{lemma}
\label{thm:bGamma1-positive-definite}
The constant matrix $-\ii\bGamma^{(1)}$ is Hermitian positive definite.
\end{lemma}
\begin{proof}
It can be directly verified that $(-\ii\bGamma^{(1)})^\dagger = -\ii\bGamma^{(1)}$,
so it is Hermitian.
We next prove that it is positive semidefinite.
Let us consider the functions
\begin{equation}
f_{j,k}(x) \coloneqq \ee^{\ii\lambda_{j,k}x}\,,\qquad
x\in[0,+\infty)\,,\qquad
1\le k\le N_j\,,\qquad
1\le j\le J\,.
\end{equation}
Recall that $\Im(\lambda_{j,k}) > 0$,
so
$\int_0^{+\infty}|f_{j,k}(x)|^2\dd x
 = \int_0^{+\infty} \ee^{-2 \Im(\lambda_{j,k})x}\dd x
 = 1/(2\Im(\lambda_{j,k})) < \infty$,
implying that all functions $f_{j,k}(x)$ are square integrable on $[0,+\infty)$.
Let us consider the inner product $<\cdot,\cdot >$ on $L^2([0,+\infty))$ as follows,
\begin{equation}
\label{e:inner-product}
<f_{j,k}(x),f_{j',k'}(x)>
 \coloneqq \int_0^{+\infty} \conj{f_{j,k}(x)}f_{j',k'}(x)\dd x
 = -\frac{\ii}{\conj{\lambda_{j,k}} - \lambda_{j',k'}}\,.
\end{equation}
Thus, the matrix $-\ii\bGamma^{(1)}$ can be realized as
\begin{equation}
-\ii\bGamma^{(1)}
 = \bpm
    <f_{1,1}, f_{1,1}> & \cdots & <f_{1,1}, f_{1,N_1}> & \cdots & <f_{1,1}, f_{J,1}> & \cdots & <f_{1,1}, f_{J,N_J}> \\
    \vdots & & \vdots & & \vdots & & \vdots  \\
    <f_{J,N_J}, f_{1,1}> & \cdots & <f_{J,N_J}, f_{1,N_1}> & \cdots & <f_{J,N_J}, f_{J,1}> & \cdots & <f_{J,N_J}, f_{J,N_J}>
    \epm\,.
\end{equation}
So, $-\ii\bGamma^{(1)}$ is a Gram matrix, which is well-known to be positive semidefinite.
Moreover, because $\{\lambda_{j,k}\}$ are distinct,
all functions $\{f_{j,k}(x)\}$ are linear independent from Equation~\eqref{e:inner-product}.
Thus, $-\ii\bGamma^{(1)}$ is positive definite.
\end{proof}
Lemma~\ref{thm:bGamma1-positive-definite} immediately yields:
(i) there exists a Hermitian positive definite square root matrix denoted by $(-\ii\bGamma^{(1)})^{\frac12}$,
which is of course invertible;
(ii) $\conj{(-\ii\bGamma^{(1)})}$ is also Hermitian positive definite.
Because the diagonal matrix $\bGamma^{(2)}(t,z)$ is invertible for all $(t,z)\in\Real^2$,
one concludes that
$\bGamma^{(2)}(t,z)\conj{(-\ii\bGamma^{(1)})}\conj{\bGamma^{(2)}(t,z)}
 = \bGamma^{(2)}(t,z)\conj{(-\ii\bGamma^{(1)})}\bGamma^{(2)}(t,z)^\dagger$ is also positive definite for all $(t,z)\in\Real^2$.
Note that
\begin{equation}
\label{e:det-I+GG}
\begin{aligned}
\det(\I + \bGamma\conj{\bGamma})
 & = \det\big(\I + (-\ii\bGamma^{(1)})\bGamma^{(2)}(t,z)\conj{(-\ii\bGamma^{(1)})\bGamma^{(2)}(t,z)}\big) \\
 & = \det\big(\I + (-\ii\bGamma^{(1)})^{\frac12}\bGamma^{(2)}(t,z)\conj{(-\ii\bGamma^{(1)})}\bGamma^{(2)}(t,z)^\dagger(-\ii\bGamma^{(1)})^{\frac12}\big) \,.
\end{aligned}
\end{equation}
We recall well-known results of two positive definite matrices $\A$ and $\B$:
(i) $\A\B\A$ is also positive definite;
(ii) $\det(\A+\B)\ge\det(\A) + \det(\B)$.
Then, the second matrix in the determinant in Equation~\eqref{e:det-I+GG} is positive definite,
implying
$\det(\I+\bGamma\conj{\bGamma}) > \det\I = 1$ for all $(t,z)\in\Real^2$.
}

\section{Degenerate $N$-soliton group}
\label{s:N-DSG}

This section proves results in Theorem~\ref{thm:N-DSG}, so one focuses on the degenerate $N$-DSGs, with $J = 1$ and $N = N_1$. We recall $|\lambda_{1,k}| = r_1 > 0$ for $k = 1,\dots, N_1$. Recall that all eigenvalues and $\alpha_{1,k}$ are distinct.

\subsection{Localization of $N$-DSG}
\label{s:N-DSG-localization}

\begin{figure}[tp]
\centering
\includegraphics[width=0.32\textwidth]{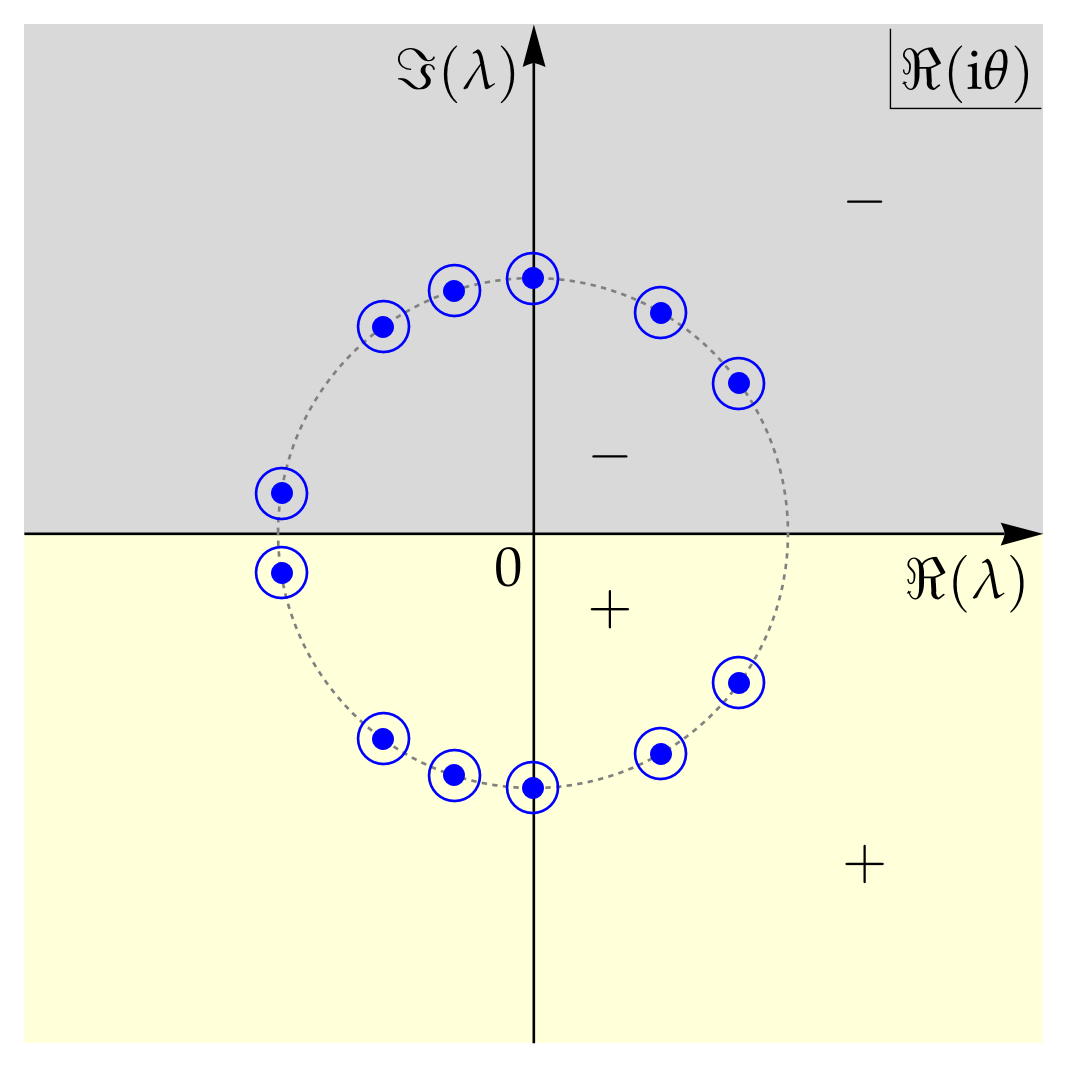}\qquad
\includegraphics[width=0.32\textwidth]{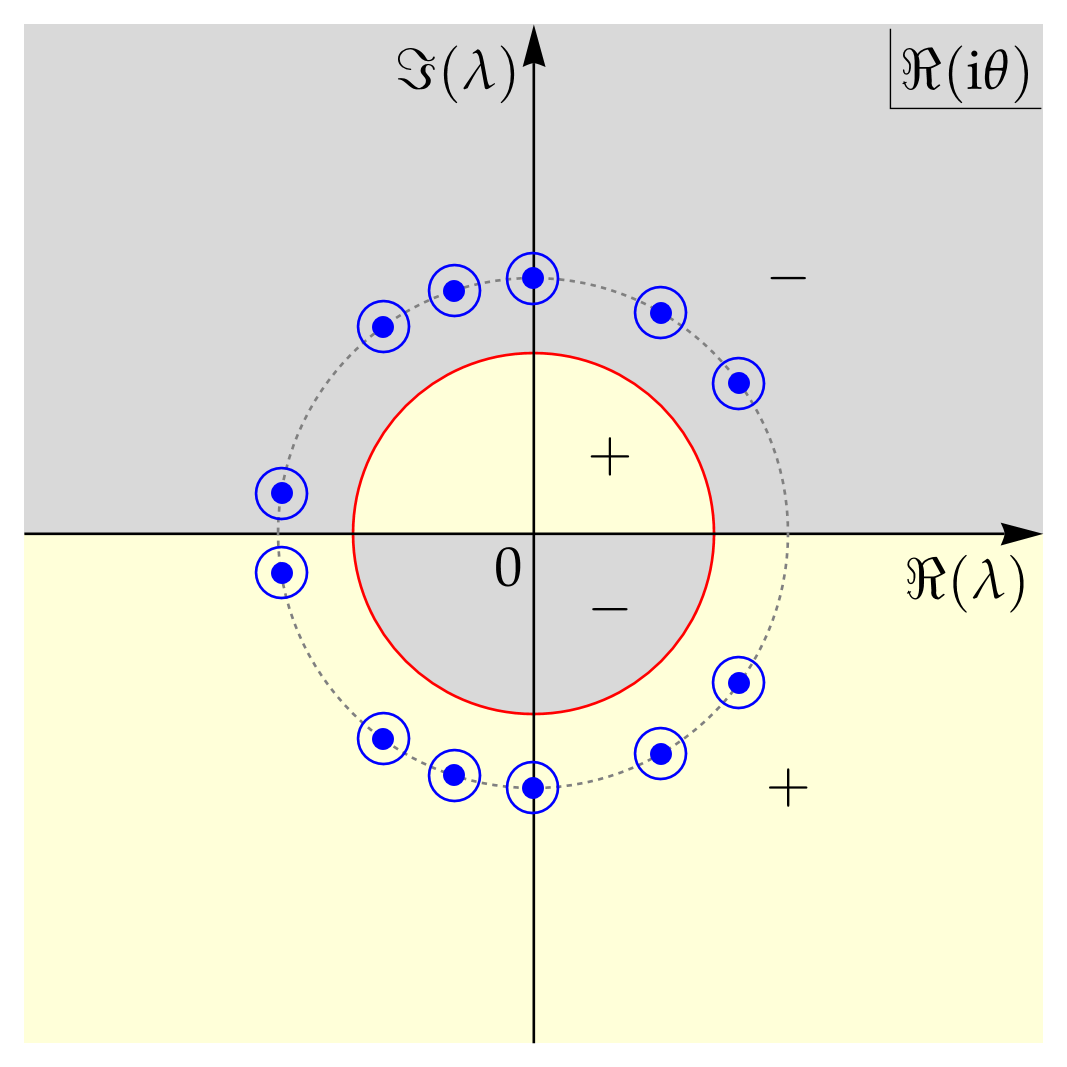}\\[2ex]
\includegraphics[width=0.32\textwidth]{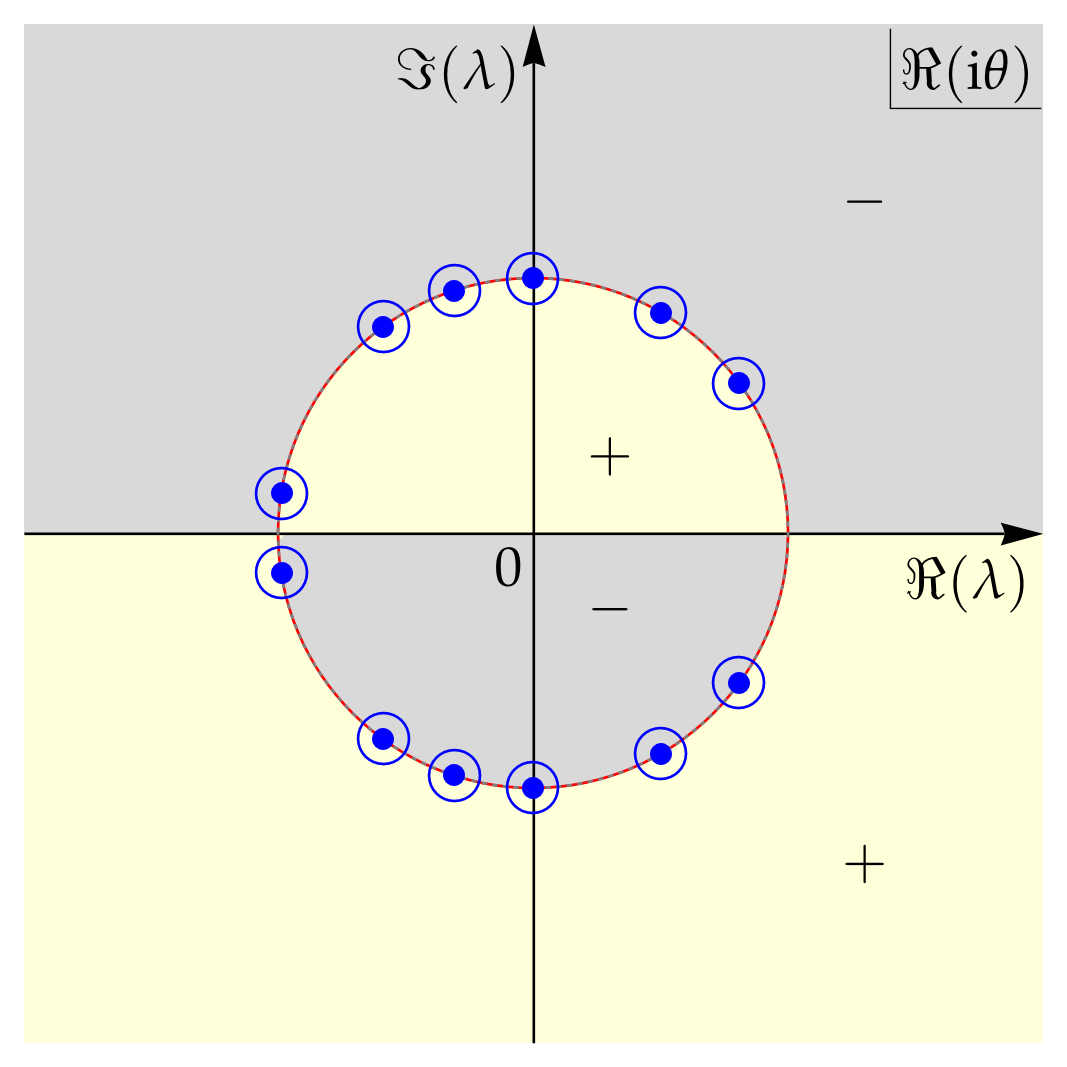}\qquad
\includegraphics[width=0.32\textwidth]{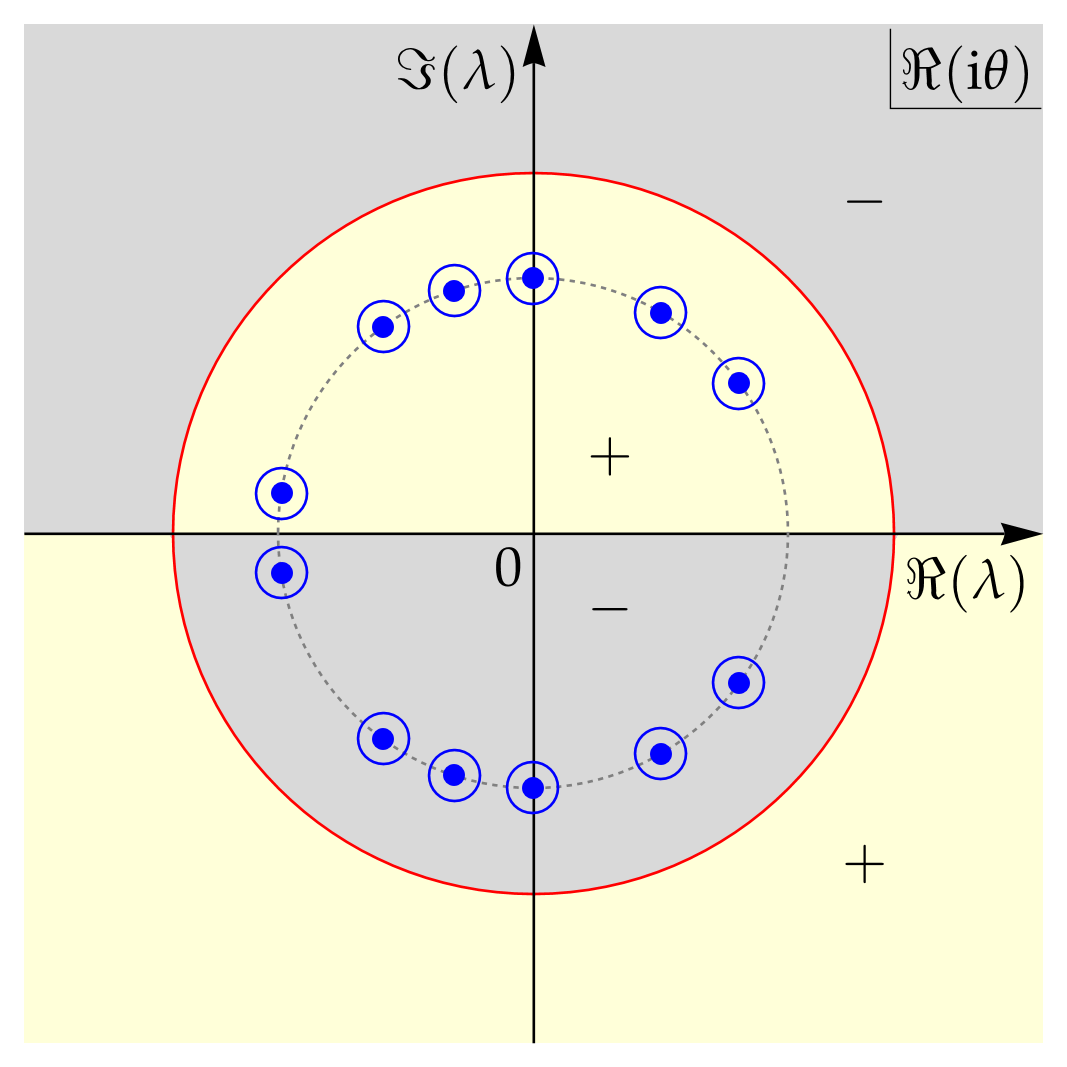}
\caption{
The sign structure of $\Re(\ii\theta(\lambda;t,z))$ with $z = \xi t$ in an initially stable medium $D_- = -1$. The blue dots represent the $N$ pairs of eigenvalues $\Lambda$ of the $N$-DSG.
The blue circles surrounding the eigenvalues are the jumps from RHP~\ref{rhp:N-soliton-jump-form}.
Gray regions denote $\Re(\ii\theta(\lambda;t,z)) < 0$, and the yellow regions denote $\Re(\ii\theta(\lambda;t,z)) > 0$.
Top left: $\xi \le 0 < V$.
Top right: $0 < \xi < V$.
Bottom left: $\xi = V$.
Bottom right: $\xi > V$.
}
\label{f:N-DSG-stable-configuration}
\end{figure}

\myblue{
In this section,
we calculate the asymptotics of the $N$-DSG $q(t,z)$ and $\brho(t,z)$ as
$t\to\pm\infty$ in a few steps.
First, we assume $z = \xi t$ with $\xi\in\Real$,
and we recall the quantity $V$ defined in Theorem~\ref{thm:N-DSG}.
We show that as long as $\xi\ne V$,
it is $q(t,z) = \O(\ee^{-\aleph|t|})$ and $\brho(t,z) = D_-\sigma_3 + \O(\ee^{-\aleph|t|})$
as $t\to\pm\infty$,
and only if $\xi = V$, $q(t,z)$ does not decay in the limits.
Hence, the $N$-DSG is traveling in the direct $z = Vt$,
and this is why $V$ is called soliton velocity.
We then compute its center using a more precise ansatz $z = Vt + z_\d $,
and determine the displacement constant $z_\d$.
}

As always, it is important to analyze the sign structure of $\Re(\ii\theta(\lambda;t,z))$ when calculating the asymptotics using Deift-Zhou's nonlinear steepest descent method. Simple calculations show that
$\Re(\ii\theta(\lambda;t,z))
 = -D_-yt[\xi + 2D_-(x^2+y^2)]/[2(x^2+y^2)]$,
with $x = \Re(\lambda)$, $y = \Im(\lambda) > 0$, and $z = \xi t$.
Looking at each eigenvalue $\lambda_{1,k}$, the equation yields
\begin{equation}
\label{e:reitheta-1k}
\Re(\ii\theta(\lambda_{1,k};t,z))
 = -\frac{D_- y_{1,k}}{2r_1^2}(\xi - V)t\,,\qquad
k = 1,2,\dots, N_1\,,
\end{equation}
where one uses $V = -2D_-r_1^2$, $r_1 = |\lambda_{1,k}|$, and $\lambda_{1,k} = x_{1,k} + \ii y_{1,k}$. Obviously, $\Re(\ii\theta(\lambda_{1,k};t,z))$ has a single root $\xi = V$, meaning that $\Re(\ii\theta) = 0$ if and only if $\xi = V$. Hence, the sign of $\Re(\ii\theta(\lambda_{1,k};t,z))$ is determined by: (i) the sign of $\xi - V$; (ii) the sign of $D_-$; and (iii) the sign of $t$.

For (i), one should recall that in RHP~\ref{rhp:N-soliton-jump-form} all the jumps are on the small circles centered at eigenvalues with radii $\epsilon\ll1$, so $\Re(\ii\theta(\lambda;t,z))$ on the jump surrounding $\lambda_{1,k}$ has the same sign as $\Re(\ii\theta(\lambda_{1,k};t,z))$, provided $\xi\ne V$. Therefore, the exponential functions $\ee^{\pm2\ii\theta(\lambda;t,z)}$ inside the jump matrices have the same growing/decaying properties as $\ee^{\pm2\ii\theta(\lambda_{1,k};t,z)}$, respectively, with $\xi\ne V$.

For (ii), because the initial state of the medium $D_-$ affects $\Re(\ii\theta(\lambda_{1,k};t,z))$, it is necessary to  discuss the two cases (stable/unstable) separately. We start with the more important and realistic case---initially stable medium ($D_- = -1$). After the stable case is properly and rigorously addressed, the unstable case can be easily calculated by slight modification of the stable one.

For (iii), one needs to discuss the asymptotics as $t\to\pm\infty$, corresponding to the situation $t \gl 0$. However, it turns out the discussion of $t\to-\infty$ is almost identical to the case of $t\to+\infty$, so \textit{we only consider the case of $t\to+\infty$ in this section}, and the case of $t\to-\infty$ is omitted for brevity.

To summarize, in an initially stable medium with $t > 0$, Equation~\eqref{e:reitheta-1k} dictates that $\Re(\ii\theta(\lambda_{1,k};t,z)) > 0$ if and only if $\xi > V$, and $\Re(\ii\theta(\lambda_{1,k};t,z)) < 0$ if and only if $\xi < V$. This sign structure of $\Re(\ii\theta)$ is shown in Figure~\ref{f:N-DSG-stable-configuration}, for the cases $\xi < V$, $\xi = V$ and $\xi > V$. The blue dots represent $N$ pairs of eigenvalues with the same radius $r$. In an initially unstable medium with $t > 0$, Equation~\eqref{e:reitheta-1k} dictates that $\Re(\ii\theta(\lambda_{1,k};t,z)) > 0$ if and only if $\xi < V$ and $\Re(\ii\theta(\lambda_{1,k};t,z)) < 0$ if and only if $\xi > V$.

\subsubsection{The case of $\xi < V$ with stable \myblue{media}}
\label{s:N-DSG-localization-stable-xi<V}

First of all, there are two subcases which are $\xi \le 0 < V$ and $0 < \xi <V$, shown in the top row in Figure~\ref{f:N-DSG-stable-configuration}, corresponding to the physical situations outside of and inside the light cone, respectively. However, it turns out that the mathematical treatments are identical, so we analyze them together in this subsection.

According to previous discussions and Figure~\ref{f:N-DSG-stable-configuration}(left), one knows that $\Re(\ii\theta(\lambda_{1,k};t,z)) \to -\infty$ as $t\to+\infty$, which proves that all the jumps in RHP~\ref{rhp:N-soliton-jump-form} are growing as $t\to+\infty$. Hence, one needs to change the jumps according to the Deift-Zhou's nonlinear steepest descent method. In particular, changing the exponential functions so that they decay as $t\to+\infty$, one needs to define new functions and deform the jumps from RHP~\ref{rhp:N-soliton-jump-form}. To achieve this, one first defines the following useful functions
\begin{equation}
\label{e:deltap-DSG-def}
\delta_\dsg(\lambda)
 \coloneqq \frac{p_\dsg^*(\lambda)}{p_\dsg(\lambda)}\,,\qquad
p_\dsg(\lambda)\coloneqq \prod_{k=1}^{N_1} (\lambda - \lambda_{1,k})\,,
\end{equation}
where we recall $p_\dsg^*(\lambda)$ is the Schwarz reflection of $p_\dsg(\lambda)$.
Obviously, $\delta_\dsg(\lambda)\to1$ as $\lambda\to\infty$.

Then, one defines a new matrix as follows,
\begin{equation}
\label{e:Ps-def}
\M_\dsg(\lambda;t,z)
 \coloneqq
  \begin{cases}
   \M(\lambda;t,z)\delta_\dsg(\lambda)^{-\sigma_3}\,, & \quad \lambda\in\Complex\setminus\bigcup\limits_{k = 1}^{N_1}\big(D_{\lambda_{1,k}}^\epsilon\bigcup D_{\conj{\lambda_{1,k}}}^\epsilon\big)\,,\\
   \M(\lambda;t,z)(\lambda - \conj{\lambda_{1,k}})^{\sigma_3}\S_{k}(\lambda;t,z)^{-1}\,, &\quad \lambda\in D_{\lambda_{1,k}}^\epsilon\,,\quad k = 1,2,\dots,N_1\,,\\
   \M(\lambda;t,z)(\lambda - \lambda_{1,k})^{-\sigma_3}\S_{k}(\conj\lambda;t,z)^\dagger\,, &\quad \lambda\in D_{\conj{\lambda_{1,k}}}^\epsilon\,,\quad k = 1,2,\dots,N_1\,,
  \end{cases}
\end{equation}
where for $k = 1,2,\dots,N_1$,
\begin{equation}
\everymath{\displaystyle}
\S_k(\lambda;t,z)
 \coloneqq
  \bpm
   \frac{\lambda - \conj{\lambda_{1,k}}}{p_\dsg^*(\lambda)}\frac{p_\dsg^*(\lambda)^2 - p_\dsg^*(\lambda_{1,k})^2}{p_\dsg(\lambda)} &
   \omega_{1,k}^{-1}\frac{\lambda - \lambda_{1,k}}{\lambda - \conj{\lambda_{1,k}}}\frac{p_\dsg^*(\lambda_{1,k})^2}{p_\dsg(\lambda)p_\dsg^*(\lambda)}\ee^{2\ii\theta(\lambda;t,z)} \\
   -\omega_{1,k}\frac{p_\dsg(\lambda)}{p_\dsg^*(\lambda)}\frac{\lambda - \conj{\lambda_{1,k}}}{\lambda - \lambda_{1,k}}\ee^{-2\ii\theta(\lambda;t,z)} &
   \frac{1}{\lambda - \conj{\lambda_{1,k}}}\frac{p_\dsg(\lambda)}{p_\dsg^*(\lambda)}
  \epm\,.
\end{equation}
It is worth pointing out that the matrices $\S_k(\lambda;t,z)$ are analytic in
$D_{\lambda_{1,k}}^\epsilon$,
and the matrices $\S_k(\conj\lambda;t,z)^\dagger$ are analytic in
$D_{\conj{\lambda_{1,k}}}^\epsilon$,
for $k = 1,2,\dots,N_1$.
Therefore,
the newly defined matrix $\M_\dsg(\lambda;t,z)$ is still analytic in
$D_{\lambda_{1,k}}^\epsilon\bigcup D_{\conj{\lambda_{1,k}}}^\epsilon$ for all integer $k$.
One can verify that $\M_\dsg(\lambda;t,z)$ solves the following RHP.
\begin{rhp}
\label{rhp:MDSG}
\myblue{
Seek a $2\times2$
matrix function $\lambda\mapsto\M_\dsg(\lambda;t,z)$ analytic on $\Complex\setminus \partial D^\epsilon$ with continuous boundary values,
where $D^\epsilon\coloneqq \bigcup_{k=1}^{N_1}\big(D_{\lambda_{1,k}}^\epsilon\bigcup D_{\conj{\lambda_{1,k}}}^\epsilon\big)$.
It has the asymptotics
$\M_\dsg(\lambda;t,z)\to\I$ as $\lambda\to\infty$ and jumps for $1\le k\le N_1$,
}
\begin{equation}
\begin{aligned}
\M_\dsg^{+}(\lambda;t,z)
 & = \M_\dsg^{-}(\lambda;t,z)\V_{\dsg,k}(\lambda;t,z)\,,\quad && \lambda\in\partial D_{\lambda_{1,k}}^\epsilon\,,\\
\M_\dsg^{+}(\lambda;t,z)
 & = \M_\dsg^{-}(\lambda;t,z)\V_{\dsg,k}(\conj\lambda;t,z)^{-\dagger}\,,\quad && \lambda\in\partial D_{\conj{\lambda_{1,k}}}^\epsilon\,,
\end{aligned}
\end{equation}
where the jump matrices are given by
\begin{equation}
\everymath{\displaystyle}
\V_{\dsg,k}(\lambda;t,z)
 =
  \bpm
   1 &
   -\frac{\lambda - \lambda_{1,k}}{\omega_{1,k}}\frac{p_\dsg^*(\lambda_{1,k})^2}{p_\dsg(\lambda)^2}\ee^{2\ii\theta(\lambda;t,z)} \\
   0 &
   1
  \epm\,,\qquad k = 1,2,\dots,N_1\,.
\end{equation}
\end{rhp}
Now,
one sees that for each $\lambda = \lambda_{1,k}$,
the surrounding jump $\V_{\dsg,k}(\lambda;t,z)$ contains exponential
$\ee^{2\ii\theta(\lambda;t,z)}$ instead of $\ee^{-2\ii\theta(\lambda;t,z)}$
from the original jump matrix $\V_{1,k}(\lambda;t,z)$,
so the growing jumps become decaying ones.
Such switches happen in the lower half plane simultaneously.
\myblue{
Equation~\eqref{e:reitheta-1k} yields that there exists a constant $\aleph > 0$
such that $|\ee^{2\ii\theta(\lambda;t,z)}| \le \ee^{-\aleph t}$ in all upper-half-plane jumps,
and similarly $|\ee^{-2\ii\theta(\lambda;t,z)}| \le \ee^{-\aleph t}$ in all lower-half-plane jumps,
with $t$ sufficiently large.
The boundedness of all jump contours implies that
$\|\V_{\dsg,k}(\lambda;t,z) - \I \| = \O(\ee^{-\aleph t})$ with $\lambda\in\partial D_{\lambda_{1,k}}^\epsilon$
and
$\|\V_{\dsg,k}(\conj\lambda;t,z)^{-\dagger} - \I\| = \O(\ee^{-\aleph t})$ with $\lambda\in\partial D_{\conj{\lambda_{1,k}}}^\epsilon$ uniformly,
as $t\to+\infty$.
Therefore,
RHP~\ref{rhp:MDSG} is a classic small-norm problem with solutions
}
\begin{equation}
\myblue{
\M_\dsg(\lambda;t,z) = \I + \O(\ee^{-\aleph t})\,,\qquad t\to+\infty\,.
}
\end{equation}
Then, reversing the definition~\eqref{e:Ps-def}
yields the asymptotics for $\M(\lambda;t,z)$ as $t\to+\infty$.
Note that the solution $q(t,z)$ is reconstructed from $\M(\lambda;t,z)$
with $|\lambda|\gg1$
and $\brho(t,z)$ is reconstructed from $\M(0;t,z)$,
from Lemma~\ref{thm:reconstruction}.
Since all eigenvalues are finite and off the real line
(cf. Definition~\ref{def:LambdaOmega}),
one concludes
\begin{equation}
\myblue{
\M(\lambda;t,z)
 = \M_\dsg(\lambda;t,z)\delta_\dsg^{\sigma_3}
 = (\I + \O(\ee^{-\aleph t}))\delta_\dsg(\lambda)^{\sigma_3}\,,\qquad
 t\to+\infty\,,
}
\end{equation}
with $|\lambda|\gg1$ or $|\lambda| \ll 1$.
The reconstruction formula from Lemma~\ref{thm:reconstruction} yields
\begin{equation}
\label{e:N-DSG-localization-stable-xi<V-qbrho}
\myblue{
q(t,z) = \O(\ee^{-\aleph t})\,,\qquad
\brho(t,z) = D_-\sigma_3(\I + \O(\ee^{-\aleph t}))\,,\qquad
t\to+\infty\,.
}
\end{equation}

\subsubsection{The case of $\xi = V$ with stable \myblue{media}}
\label{s:N-DSG-localization-stable-xi=V}

Recall that $\Re(\ii\theta(\lambda_{1,k};t,z) = 0$
for
$k = 1,2,\dots,N_1$
from the discussion at the beginning of Section~\ref{s:N-DSG-localization}.
In terms of the jumps surrounding the eigenvalues,
one part of each circular jump is growing while the rest is decaying as
$t\to+\infty$
[cf, Figure~\ref{f:N-DSG-stable-configuration}(bottom left)].
As a result,
the jumps $\V_{1,k}(\lambda;t,z)^{-1}$ on $\partial D_{\lambda_{1,k}}^\epsilon$
and
$\V_{1,k}(\conj\lambda;t,z)^\dagger$ on $\partial D_{\conj{\lambda_{1,k}}}^\epsilon$
in RHP~\ref{rhp:N-soliton-jump-form} do not decay nor grow uniformly as $t\to+\infty$.
So,
the asymptotics cannot be easily calculated from RHP~\ref{rhp:N-soliton-jump-form}.
Instead, one can directly solve the equivalent RHP~\ref{rhp:N-soliton-residue-form},
and obtains the explicit $N$-DSG given via Theorem~\ref{thm:N-soliton-formula}.
Then,
one observes that all the exponential functions
$\ee^{\pm2\ii\theta_{1,k}}$ for $k = 1,2,\dots, N_1$ with $J = 1$
from the solution formula in Theorem~\ref{thm:N-soliton-formula}
are purely oscillatory,
because $\Re(\ii\theta_{1,k}) = 0$.
\myblue{
Moreover, Theorem~\ref{thm:reconstruction} guarantees that the $N$-DSG is non-singular for all $(t,z)\in\Real^2$.
}
Hence,
one concludes that $q(t,z)$ and $\brho(t,z)$ do not decay nor grow in the limit
$t\to+\infty$.

\subsubsection{The case of $\xi > V$ with stable \myblue{media}}
\label{s:N-DSG-localization-stable-xi>V}

Similarly to previous subsections, one recalls the discussion at the beginning of Section~\ref{s:N-DSG-localization} [cf. Figure~\ref{f:N-DSG-stable-configuration}(Bottom right)], to see that all the jumps in RHP~\ref{rhp:N-soliton-jump-form} decay to the identity matrix uniformly as $t\to+\infty$.
\myblue{
Similarly to RHP~\ref{rhp:MDSG},
one can directly write $\M(\lambda;t,z) = \I + \O(\ee^{-\aleph t})$,
as $t\to+\infty$,
where $\aleph > 0$ is a constant.
}
Lemma~\ref{thm:reconstruction} yields
\begin{equation}
\label{e:N-DSG-localization-stable-xi>V-qbrho}
\myblue{
q(t,z) = \O(\ee^{-\aleph t})\,,\qquad
\brho(t,z) = D_-\sigma_3 + \O(\ee^{-\aleph t})\,,\qquad
t\to+\infty\,.
}
\end{equation}

Combining all cases,
one concludes that the $N$-DSG is traveling along a line in the form of
$z = Vt + z_\d$ with an undetermined constant $z_\d$.
When observed in a different direction $z = \xi t$ where $\xi\ne V$,
the $N$-DSG rapidly vanishes to the ZBG.
In fact,
it has proved that the solution decays exponentially.
Therefore,
one obtains the localization of the $N$-DSG,
and knows that the $N$-DSG travels with the speed $V$.
The localization [results (1) and (2)] of $N$-DSG in Theorem~\ref{thm:N-DSG}
in a stable medium is proved.

\subsubsection{Localization of $N$-DSG with unstable \myblue{media}}
\label{s:N-DSG-localization-unstable}

It is time to discuss $N$-DSG in an unstable medium, i.e., $D_- = 1$. Recall Equation~\eqref{e:reitheta-1k} and the discussion afterwards. It is also worth recalling the velocity $V_\u \coloneqq -2r_1^2 < 0$ in the unstable medium. Instead of calculating the asymptotics of RHP~\ref{rhp:N-soliton-jump-form} step-by-step, we make the following connections between the unstable cases and the previously discussed stable cases.
\begin{enumerate}
\item
In the case of $\xi < V_\u$,
Equation~\eqref{e:reitheta-1k} implies that
$\Re(\ii\theta(\lambda_{1,k};t,z)) > 0$,
yielding that all the jumps in RHP~\ref{rhp:N-soliton-jump-form}
\myblue{are decaying to the identity matrix exponentially as $t\to+\infty$.}
Similarly to Section~\ref{s:N-DSG-localization-stable-xi>V},
one immediately obtains the same result~\eqref{e:N-DSG-localization-stable-xi>V-qbrho}.
\item
In the case of $\xi = V_\u$,
Equation~\eqref{e:reitheta-1k} implies that
$\Re(\ii\theta(\lambda_{1,k};t,z)) = 0$.
Recall the discussion in Section~\ref{s:N-DSG-localization-stable-xi=V}.
The MBEs solutions can be solved exactly via RHP~\ref{rhp:N-soliton-residue-form}
and can be shown to not decay nor grow in the limit $t\to+\infty$.
\item
In the case of $V_\u < \xi$,
Equation~\eqref{e:reitheta-1k} implies that
$\Re(\ii\theta(\lambda_{1,k};t,z)) < 0$,
so all jumps in RHP~\ref{rhp:N-soliton-jump-form} are growing as $t\to+\infty$.
One can follow the calculations in Section~\ref{s:N-DSG-localization-stable-xi<V},
apply similar treatments to RHP~\ref{rhp:N-soliton-jump-form},
and obtain Equation~\eqref{e:N-DSG-localization-stable-xi<V-qbrho}.
\end{enumerate}
Finally, one proves the results~(1) and~(2) in Theorem~\ref{thm:N-DSG} in the unstable case.

\subsection{Center of $N$-DSG}

The previous subsections prove that the $N$-DSG is traveling in the direction $z = Vt$. Hence, the center of $N$-DSG also moves with the same velocity $V$ as $t$ changes. One defines the center $z_\c(t) \coloneqq Vt + z_\d$, where the displacement $z_\d\in\Real$ is to be determined.

Our next task is to \myblue{look for} $z_\d$.
We perform two pre-treatments before the actual calculation.
Firstly,
\myblue{let} us substitute $z = z_\c(t)$ into the quantity $\theta(\lambda_{1,k};t,z)$ yielding
\begin{equation}
\label{e:N-DSG-center-reitheta}
\myblue{
\Re(\ii\theta(\lambda_{1,k};t,z_\c(t)))
 = -\frac{D_-\Im(\lambda_{1,k})}{2r_1^2}z_\d\,,\qquad
k = 1,2,\dots, N_1\,.
}
\end{equation}
Note that all $\Im(\lambda_{1,k}) > 0$.
\myblue{
Secondly, we factor the matrix $\bGamma$ identically to Equation~\eqref{e:det-I+GG},
but with $J = 1$,
\begin{equation}
\everymath{\displaystyle}
\begin{aligned}
\bGamma
 & = \bGamma^{(1)}\bGamma^{(2)}\,,\\
\bGamma^{(1)}
 & \coloneqq \bpm \frac{1}{\conj{\lambda_{1,1}} - \lambda_{1,1}} & \frac{1}{\conj{\lambda_{1,1}} - \lambda_{1,2}} & \cdots & \frac{1}{\conj{\lambda_{1,1}} - \lambda_{1,N_1}} \\
\vdots & \vdots & & \vdots\\
\frac{1}{\conj{\lambda_{1,N_1}} - \lambda_{1,1}} & \frac{1}{\conj{\lambda_{1,N_1}} - \lambda_{1,2}} & \cdots & \frac{1}{\conj{\lambda_{1,N_1}} - \lambda_{1,N_1}} \epm\,,\\
\bGamma^{(2)}
 & \coloneqq \diag(\omega_{1,1}\ee^{-2\ii\theta_{1,1}}, \omega_{1,2}\ee^{-2\ii\theta_{1,2}},\cdots,\omega_{1,N_1}\ee^{-2\ii\theta_{1,N_1}})\,.
\end{aligned}
\end{equation}
Recall that $\theta_{1,k} = \theta(\lambda_{1,k};t,z_\c(t))$ from Theorem~\ref{thm:reconstruction},
with substitution $z = z_\c(t)$.
}
It is worth pointing out that $\bGamma^{(1)}$ is independent of $\theta$,
so it always remains constant no matter $z_\d\to\pm\infty$.
On the other hand,
$\bGamma^{(2)}$ grows to infinity or decays to the identity matrix,
depending on the state of the medium and whether $z_\d\to-\infty$ or $z_\d\to+\infty$.

As demonstrated in~\cite{lb2018},
an effective way to find the center of a DSG is
to match the asymptotic behavior of solutions as $z_\d\to\pm\infty$.
This follows from a simple observation:
\myblue{the behavior of the $N$-DSG is considerably complicated} when the observation point $z$ is close to its center $z_\c$; but whenever the observation point $z$ moves away from the center $z_\c$,
the solution decays rapidly (the proved localization property in Section~\ref{s:N-DSG-localization}), yielding simple behavior of $q(t,z)$. Because $z_\c$ is the center, the solution to its right and to its left should behave similarly further away. Thus, calculating the asymptotics of the solution $q(t,z)$ as $z_\d\to\pm\infty$ and matching them should give rise to a proper value of $z_\d$, describing the displacement. We next discuss the asymptotics.

\subsubsection{Stable medium ($D_- = -1$)}

In this case, $z_\mathrm{d}\to-\infty$ yields $\ee^{-2\ii\theta_{1,k}}\to+\infty$ from Equation~\eqref{e:N-DSG-center-reitheta} for all $k = 1,2,\dots, N_1$.
Consequently, $\bGamma^{(2)}\to\infty$ implies that $\bGamma\to\infty$ as well. Examining the solution $q(t,z)$ from Equation~\eqref{e:Nsoliton-formula} yields that
as $z_\mathrm{d}\to-\infty$
\begin{equation}
\label{e:N-DSG-center-expand}
q(t,z)
 = 2\ii - 2\ii \frac{\det(\bGamma\conj{\bGamma})\det(\conj{\bGamma}^{-1}\bGamma^{-1}(\I - \myblue{\vbone\conj{\C_\infty}}) + \I)}{\det(\bGamma\conj{\bGamma})\det(\conj{\bGamma}^{-1}\bGamma^{-1} + \I)}
 = 2\ii - 2\ii\frac{1 + o(1)}{1 + o(1)} = o(1)\,.
\end{equation}
As expected, the solution decay to zero as one moves away from its center. In fact, it can be seen from the above equation that $q(t,z)$ decays exponentially. The exact decay rate is necessary for later comparison. However, the current form is inconvenient for calculating the leading-order term. So, we rewrite the asymptotics differently. The above calculation shows that the numerator and the denominator inside the fraction in $q(t,z)$ behave similarly as $z_\d\to-\infty$. (Note that the fraction of Equation~\eqref{e:N-DSG-center-expand} tends to one, not zero.)
One can extract the leading-order term alone from the denominator, as follows,
\begin{equation}
\label{e:NDSG-zd-leading1}
\det(\I + \bGamma\conj{\bGamma})
 = \det(\bGamma\conj{\bGamma})\det(\conj{\bGamma}^{-1}\bGamma^{-1} + \I)
 = \det(\bGamma\conj{\bGamma})(1 + o(1))\,,\qquad
z_\d\to-\infty\,.
\end{equation}

The other limit $z_\d\to+\infty$ implies that $\ee^{-2\ii\theta_{1,k}}\to0$ from Equation~\eqref{e:N-DSG-center-reitheta} for $k = 1,2,\dots,N_1$. Correspondingly, $\bGamma^{(2)}\to0$ and $\bGamma\to0$.
Clearly, one sees that
\begin{equation}
q(t,z)
 = 2\ii - 2\ii\frac{\det(\I - \myblue{\vbone\conj{\C_\infty}} + \bGamma\conj{\bGamma})}{\det(\I + \bGamma\conj{\bGamma})}
 = 2\ii - 2\ii\frac{1 + o(1)}{1 + o(1)} = o(1)\,.
\end{equation}
Again, $q(t,z)$ decays to zero as one moves away from the soliton group center as $z_\d\to+\infty$. As before, one looks at the denominator of $q(t,z)$, but the direct calculation of $\det(\I + \bGamma\conj{\bGamma})$ yields the leading-order term as $1$. To extract useful information, we rewrite the fraction so that both the numerator and the denominator are growing instead of decaying as $z_\d\to+\infty$,
\begin{equation}
\frac{\det(\I - \myblue{\vbone\conj{\C_\infty}} + \bGamma\conj{\bGamma})}{\det(\I + \bGamma\conj{\bGamma})}
 = \frac{\det(\conj{\bGamma}^{-1}\bGamma^{-1}(\I - \myblue{\vbone\conj{\C_\infty}}) + \I)}{\det(\conj{\bGamma}^{-1}\bGamma^{-1} + \I)}\,.
\end{equation}
The denominator yields the leading-order term as $z_\d\to+\infty$
\begin{equation}
\label{e:NDSG-zd-leading2}
\det(\conj{\bGamma}^{-1}\bGamma^{-1} + \I)
 = \det(\conj{\bGamma}^{-1}\bGamma^{-1})\det(\I + \bGamma\conj{\bGamma})
 = \det(\conj{\bGamma}^{-1}\bGamma^{-1})(1 + o(1))\,.
\end{equation}

Matching the two asymptotics from Equations~\eqref{e:NDSG-zd-leading1} and~\eqref{e:NDSG-zd-leading2} implies $\det(\bGamma\conj{\bGamma}) = \det(\conj{\bGamma}^{-1}\bGamma^{-1})$. The center $z_\d$ is therefore determined by the equation $\det(\bGamma\conj{\bGamma})^2 = 1$, which can be calculated as
\begin{equation}
\label{e:DSG-center-equation}
1 = |\det(\bGamma)|^4
 = |\det(\bGamma^{(1)})|^4|\det(\bGamma^{(2)})|^4
 = |\det(\bGamma^{(1)})|^4| \bigg|\prod_{k=1}^{N_1}\omega_{1,k}\bigg|^4\bigg|\ee^{-2\ii\sum_{k = 1}^{N_1}\theta_{1,k}}\bigg|^4\,.
\end{equation}
Note that $\bGamma^{(1)}$ is a Cauchy matrix, whose determinant can be readily computed,
\begin{equation}
\det\big(\bGamma^{(1)}\big)
 = \frac{\prod_{k = 2}^{N_1}\prod_{l = 1}^{k-1}|\lambda_{1,k} - \lambda_{1,l}|^2}{\prod_{k=1}^{N_1}\prod_{l=1}^{N_1}(\conj{\lambda_{1,k}} - \lambda_{1,l})}\,.
\end{equation}
Finally, one obtains the center of the soliton group $z_\c(t) = Vt + z_\d$ with $z_\d$ given in Theorem~\ref{thm:N-DSG}. This proves the result (3) in the theorem in a stable medium.

\subsubsection{Unstable medium ($D_- = 1$)}

Here, one employs similar arguments as in Section~\ref{s:N-DSG-localization}, by comparing and relating the stable and unstable case. Opposite to the previous stable case, Equation~\eqref{e:N-DSG-center-reitheta} implies the following equivalence
\begin{enumerate}
\item
The case of an unstable medium with $z_\d\to-\infty$ implies that
$\Re(\ii\theta_{1,k})\to-\infty$,
equivalent to the calculation of the case of a stable medium with
$z_\d\to+\infty$.
\item
The case of an unstable medium with $z_\d\to+\infty$ implies that
$\Re(\ii\theta_{1,k})\to+\infty$,
which is equivalent to the calculation of the case of a stable medium with
$z_\d\to-\infty$.
\end{enumerate}
Hence, the two asymptotics switch between the stable and unstable cases.
Nonetheless,
matching the two leading-order terms in the asymptotics in either case yields identical result,
i.e.,
Equation~\eqref{e:DSG-center-equation}.
Thus,
the formula for the center of $N$-DSG in an unstable case in Theorem~\ref{thm:N-DSG} is proved.

\subsection{Boundedness of $\M(\lambda;t,z)$}
\label{s:M-boundedness}

Finally,
we would like to prove that with $J = 1$,
the solution $\M(\lambda;t,z)$ to RHPs~\ref{rhp:N-soliton-residue-form} and~\ref{rhp:N-soliton-jump-form} are always bounded as $t\to\pm\infty$ in the direction $z = V t$. This result is used in Section~\ref{s:N-soliton-asymptotics} when calculating the soliton asymptotics.
Of course, there are four subcases resulting from all combinations of $t\to\pm\infty$ and stable/unstable medium. We first discuss the case of $t\to+\infty$ in a stable medium. The other three can be easily analyzed afterwards.
\begin{enumerate}[align = left]
\item[\bf Case I] Boundedness as $t\to+\infty$ in a stable medium.
One needs to look at RHP~\ref{rhp:N-soliton-residue-form}.
The first row of its solution $\M(\lambda;t,z)$
has been solved in Equation~\eqref{e:M-sol1}.
Using Equation~\eqref{e:reitheta-1k} and $z = V t$,
one can see that
all the $(t,z)$ dependence in Equation~\eqref{e:M-sol1}
are in the exponential functions $\ee^{\pm2\ii\theta_{1,k}}$
(cf. Theorem~\ref{thm:N-DSG}),
which are oscillatory.
\myblue{
Moreover, it is proved that the denominator can never be zero.
}
Similar phenomenon happens to the second row of $\M(\lambda;t,z)$
whose explicit expression is omitted for brevity.
Hence, the matrix $\M(\lambda;t,z)$ is bounded as $t\to+\infty$.
\item[\bf Case II]
Boundedness as $t\to-\infty$ in a stable medium.
The exponentials $\ee^{\pm2\ii\theta_{1,k}}$ are still oscillatory.
Hence, repeating similar arguments in \textbf{Case~I} yields that $\M(\lambda;t,z)$ is bounded as $t\to-\infty$, as well.
\item[\bf Case III]
Boundedness as $t\to+\infty$ in an unstable medium. Using the analysis in Section~\ref{s:N-DSG-localization-unstable} and repeating similar arguments of \textbf{Case~I} yields desired results.
\item[\bf Case IV] Boundedness as $t\to-\infty$ in an unstable medium. One simply repeat similar argument of \textbf{Case~II} and obtains the boundedness of $\M(\lambda;t,z)$.
\end{enumerate}
Hence, we have proved the boundedness of $\M(\lambda;t,z)$ of RHPs~\ref{rhp:N-soliton-residue-form} and~\ref{rhp:N-soliton-jump-form} as $t\to\pm\infty$ in both stable and unstable \myblue{media}, with $J = 1$ and $z = V t$.

\section{Soliton asymptotics}
\label{s:N-soliton-asymptotics}

We prove Theorem~\ref{thm:soliton-asymptotics} in this section. Several cases are necessary to be considered.
First,
\myblue{we need to discuss whether the medium is initially in its stable or unstable state}. Also, it is necessary to calculate the asymptotics as $t\to\pm\infty$. Finally, the direction $\xi$ from $z = \xi t$ affects the application of the Deift-Zhou's nonlinear steepest descent method. In order to simplify our discussion regarding all the cases, we mainly focus on the case of the stable medium, then briefly discuss what happens in an unstable medium.

\subsection{Basic analysis of soliton asymptotics in a stable medium}
\label{s:N-soliton-asymptotics-stable}

\begin{figure}[tp]
\centering
\includegraphics[width=0.32\textwidth]{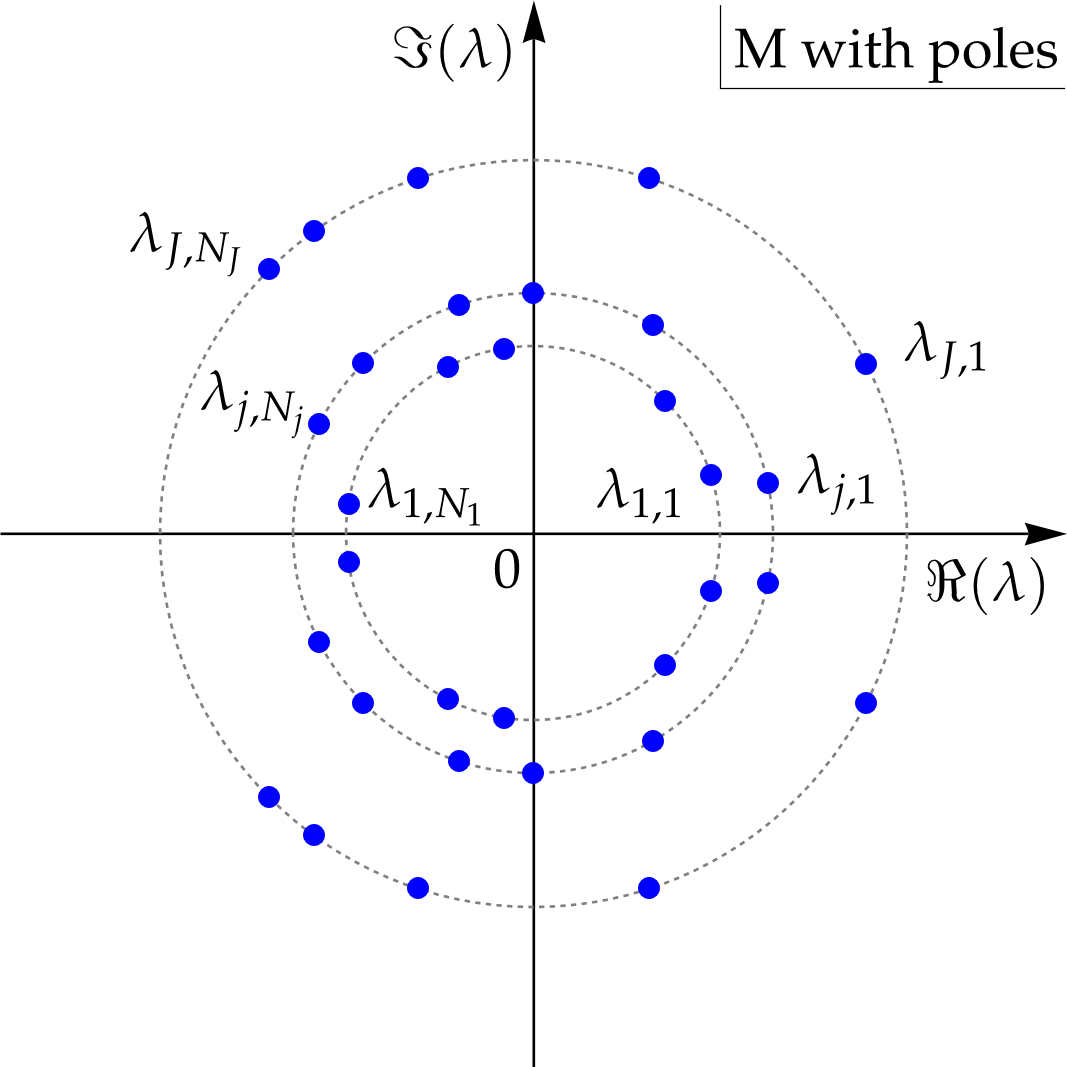}\qquad\qquad
\includegraphics[width=0.32\textwidth]{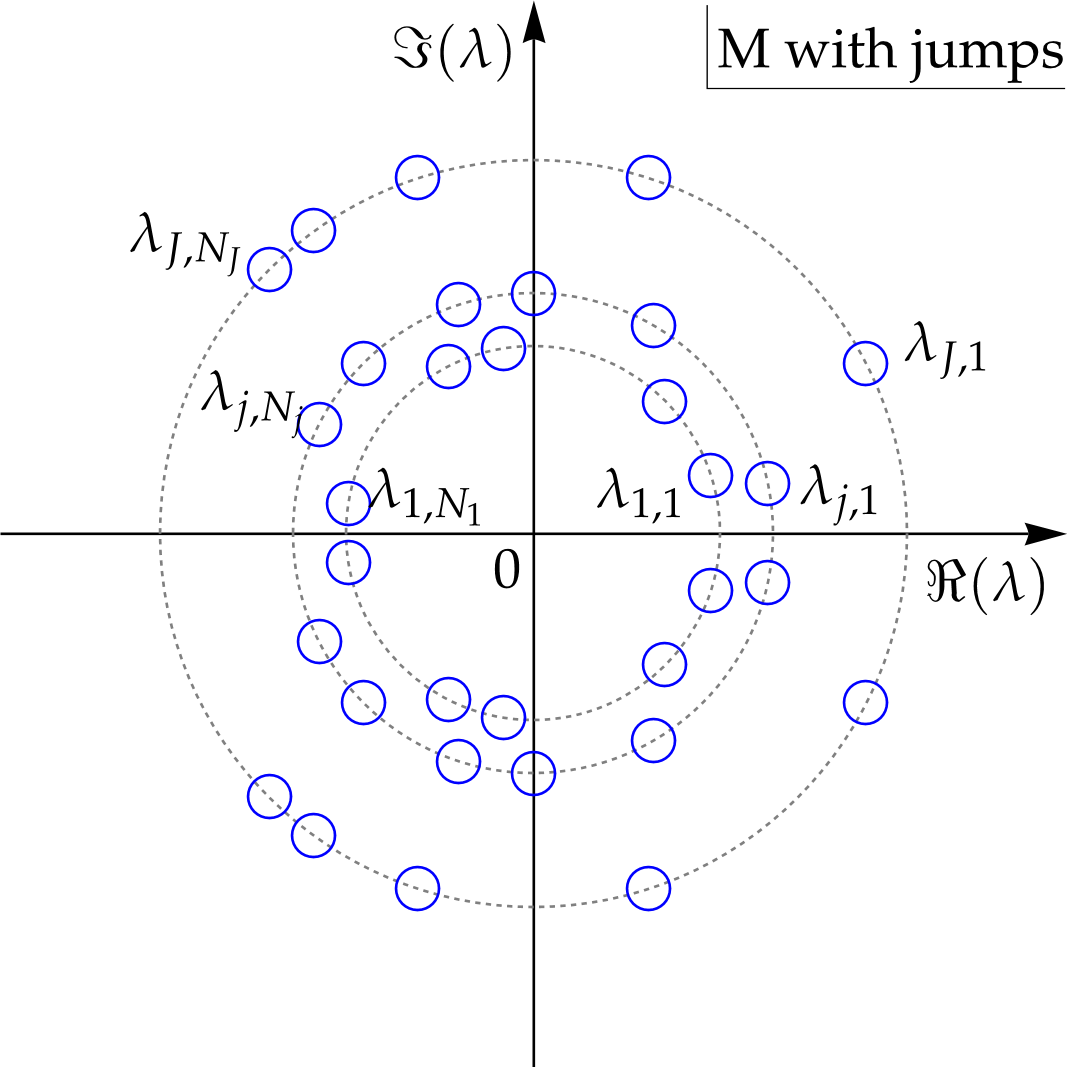}
\caption{Left: An illustration of distributions of $\Lambda$ with $1\le j \le J$ and $J\ge2$ in the complex plane.
Right: The equivalent RHP with jumps transformed from poles.
}
\label{f:N-soliton-configuration1}
\end{figure}
We now start proving Theorem~\ref{thm:soliton-asymptotics} in the case of a stable medium
($D_- = -1$).
Recall the assumption that $J \ge 2$,
so the configuration of the spectra is shown in Figure~\ref{f:N-soliton-configuration1}.
\myblue{
As such, the quantities $V_j = -2D_-r_j^2$ defined in Equation~\eqref{e:Vj-def} are ordered as
\begin{equation}
0 < V_1 < V_2 < \dots < V_J\,.
\end{equation}
As the direction $\xi$ from $z = \xi t$ varies from $-\infty$ to $+\infty$,
there are several cases summarized in Table~\ref{tab:soliton-asymptotics-stable-cases}.
}
\begin{table}[t]
\caption{Four cases in the long-time asymptotics in a stable medium.}
\begin{tabular}{|c|c|c|}
\hline
Stable Cases & $\vphantom{\Big|}$\quad Relations between $\xi$ and $V_j$ \quad & \hspace{6em} Description \hspace{6em} \\
\hline\hline
Case 1 & $\vphantom{\Big|}\xi < V_1 < \dots < V_J$ & \myblue{$\xi$ is the smallest} \\
\hline
Case 2 & $\vphantom{\Big|}V_1 < \dots < V_J < \xi$ & \myblue{$\xi$ is the largest} \\
\hline
Case 3 & $\vphantom{\Big|}V_j < \xi < V_{j+1}$ with $j = 1,\dots,J-1$ & \myblue{$\xi$ is in-between two quantities $V_j$ and $V_{j+1}$} \\
\hline
Case 4 & $\vphantom{\Big|}V_{j-1} < \xi = V_j < V_{j+1}$ with $j = 1,\dots,J$ & \myblue{$\xi$ coincides with some $V_j$} \\
\hline
\end{tabular}
\label{tab:soliton-asymptotics-stable-cases}
\end{table}

\myblue{
The outline of the proof is described here.
The first three cases yield trivial results in the asymptotics $t\to\pm\infty$,
where the leading-order terms are the zero background solutions $q(t,z) = 0$, $D(t,z) = D_-$ and $P(t,z) = 0$.
The nontrivial result arises from case 4,
which is shown as a DSG corresponding to eigenvalue set $\Lambda_j$ as the leading-order term in the asymptotics.
Also, we recall the useful quantities $\delta_j(\lambda)$ and $p_j(\lambda)$ defined in Equation~\eqref{e:deltan-def}.
Note that $\delta_j(\lambda)\to1$ as $\lambda\to\infty$ for each $j = 1,2,\dots, J$.
Note also that $\delta_1(\lambda)$ with $J = 1$ is equivalent to $\delta_\dsg(\lambda)$ defined in Equation~\eqref{e:deltap-DSG-def}.
}

\subsection{Soliton asymptotics as $t\to+\infty$ in a stable medium}
\label{s:N-soliton-asymptotics-stable-tpos}

We first present the asymptotic calculations of all cases in Table~\ref{tab:soliton-asymptotics-stable-cases} as $t\to+\infty$.

\subsubsection{The case of $\xi < V_1$}
\label{s:N-soliton-asymptotics-stable-tpos-case1}

\begin{figure}[tp]
\centering
\includegraphics[width=0.32\textwidth]{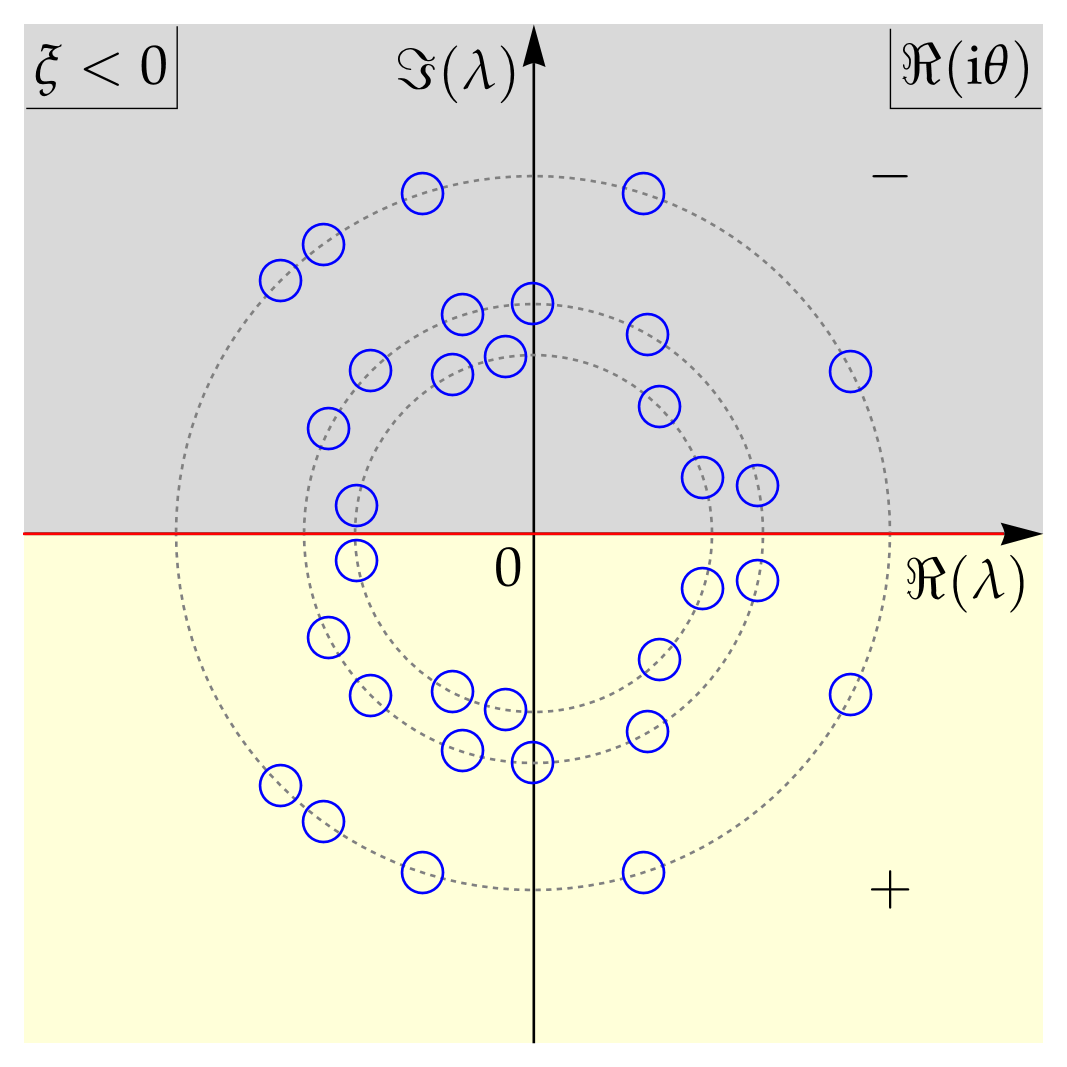}\qquad\qquad
\includegraphics[width=0.32\textwidth]{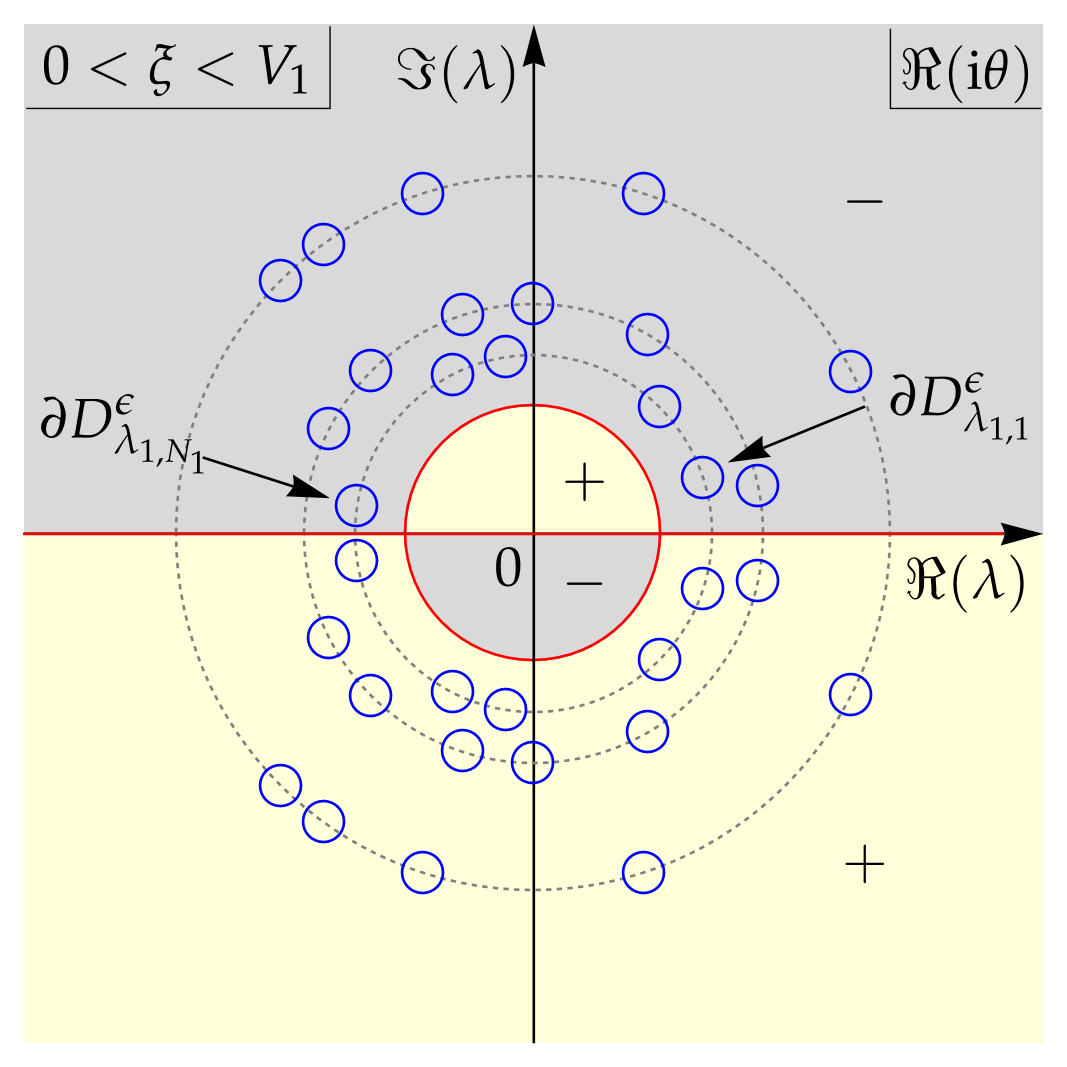}
\caption{Representative plots of Case 1 from Table~\ref{tab:soliton-asymptotics-stable-cases} with $t\to+\infty$, where the direction of the observation line $\xi$ is smaller than all velocities.
Left: $\xi \le 0 < V_1$.
Right: $0 < \xi < V_1$
The blue circles represent all the jump contours in RHP~\ref{rhp:N-soliton-jump-form}.
Yellow or gray region corresponds to $\Re(\ii\theta) > 0$ or $\Re(\ii\theta) < 0$, respectively.
The red contours are the separation between the regions.}
\label{f:soliton-asymptotics-stable-tpos-case1}
\end{figure}
There are two subcases, which are $\xi \le 0$ and $0 < \xi < V_1$. The first one discusses the situation when looking at the solution outside and on the boundary of the light cone, whereas the latter one discusses what happens inside the light cone. The sign structures of $\Re(\ii\theta(\lambda;t,z))$ differ slightly, as illustrated in Figure~\ref{f:soliton-asymptotics-stable-tpos-case1}, which are calculated by substituting $z = \xi t$ into $\theta(\lambda;t,z)$. Nonetheless, in both subcases all jumps (with small enough radii $\epsilon$) in RHP~\ref{rhp:N-soliton-jump-form} contain growing exponentially, specifically, $\exp(-2\ii\theta(\lambda;t,z))$ in the upper half plane, and $\exp(2\ii\theta(\lambda;t,z))$ in the lower half plane. Therefore, it is unnecessary to distinguish the two subcases mathematically, and we treat them simultaneously below.

One needs to modify all the growing jumps in order to perform the nonlinear steepest descent. It turns out that the process is \textit{quite similar} to what have been done in Section~\ref{s:N-DSG-localization} in the case of $\xi < V$. We only outline the process here in order to avoid repetitions. Step (i), one utilizes the rational function $\delta_1(\lambda)$ from Equation~\eqref{e:deltan-def}, similarly to the usage of $\delta_\dsg(\lambda)$ in Section~\ref{s:N-DSG-localization-stable-xi<V}. Step (ii), one defines a new matrix function mimicking Equation~\eqref{e:Ps-def},
\begin{equation}
\label{e:P1-def}
\M^{(1)}(\lambda;t,z)
 \coloneqq
  \begin{cases}
   \M(\lambda;t,z)\delta_1(\lambda)^{-\sigma_3}\,, & \qquad \lambda\in\Complex\setminus\bigcup\limits_{j = 1}^J\bigcup\limits_{k = 1}^{N_j}\big(D_{\lambda_{j,k}}^\epsilon\bigcup D_{\conj{\lambda_{j,k}}}^\epsilon\big)\,,\\
   \M(\lambda;t,z)(\lambda - \conj{\lambda_{j,k}})^{\sigma_3}\Y_{j,k}^{(1)}(\lambda;t,z)^{-1}\,, &\qquad \lambda\in D_{\lambda_{j,k}}^\epsilon\,,\quad \lambda_{j,k}\in\Lambda\,,\\
   \M(\lambda;t,z)(\lambda - \lambda_{j,k})^{-\sigma_3}\Y_{j,k}^{(1)}(\conj\lambda;t,z)^\dagger\,, &\qquad \lambda\in D_{\conj{\lambda_{j,k}}}^\epsilon\,,\quad \lambda_{j,k}\in\Lambda\,,
  \end{cases}
\end{equation}
where
\begin{equation}
\label{e:Yjk1-def}
\everymath{\displaystyle}
\Y_{j,k}^{(1)}(\lambda;t,z)
 \coloneqq
  \bpm
   \frac{\lambda - \conj{\lambda_{j,k}}}{p_1^*(\lambda)}\frac{p_1^*(\lambda)^2 - p_1^*(\lambda_{j,k})^2}{p_1(\lambda)} &
   \omega_{j,k}^{-1}\frac{\lambda - \lambda_{j,k}}{\lambda - \conj{\lambda_{j,k}}}\frac{p_1^*(\lambda_{j,k})^2}{p_1(\lambda)p_1^*(\lambda)}\ee^{2\ii\theta(\lambda;t,z)} \\
   -\omega_{j,k}\frac{p_1(\lambda)}{p_1^*(\lambda)}\frac{\lambda - \conj{\lambda_{j,k}}}{\lambda - \lambda_{j,k}}\ee^{-2\ii\theta(\lambda;t,z)} &
   \frac{1}{\lambda - \conj{\lambda_{j,k}}}\frac{p_1(\lambda)}{p_1^*(\lambda)}
  \epm\,.
\end{equation}
\myblue{
Recall that $p_1(\lambda)$ is defined in Equation~\eqref{e:deltan-def}.
}
Consequently, one can verify that $\M^{(1)}(\lambda;t,z)$ satisfies the following RHP.
\begin{rhp}
\label{rhp:M1}
\myblue{
Seek a $2\times2$
matrix function $\M^{(1)}(\lambda;t,z)$ analytic on $\Complex\setminus\partial D^\epsilon$ with continuous boundary values,
where
$D^\epsilon\coloneqq \bigcup_{\lambda_{j,k}\in\Lambda}\big(D_{\lambda_{j,k}}^\epsilon\bigcup D_{\conj{\lambda_{j,k}}}^\epsilon\big)$.
It has the asymptotics $\M^{(1)}(\lambda;t,z)\to\I$ as $\lambda\to\infty$ and jumps
}
\begin{equation}
\begin{aligned}
\M^{(1)+}(\lambda;t,z)
 & = \M^{(1)-}(\lambda;t,z)\V_{j,k}^{(1)}(\lambda;t,z)\,,\qquad && \lambda\in\partial D_{\lambda_{j,k}}^\epsilon\,,\\
\M^{(1)+}(\lambda;t,z)
 & = \M^{(1)-}(\lambda;t,z)\V_{j,k}^{(1)}(\conj\lambda;t,z)^{-\dagger}\,,\qquad && \lambda\in\partial D_{\conj{\lambda_{j,k}}}^\epsilon\,,
\end{aligned}
\end{equation}
where $\lambda_{j,k}\in\Lambda$ and the jump matrices are given by
\begin{equation}
\everymath{\displaystyle}
\V_{j,k}^{(1)}(\lambda;t,z)
 \coloneqq \bpm
    1 &
    -\frac{\lambda - \lambda_{j,k}}{\omega_{j,k}}\frac{p_1^*(\lambda_{j,k})^2}{p_1(\lambda)^2}\ee^{2\ii\theta(\lambda;t,z)} \\
    0 &
    1
    \epm\,.
\end{equation}
\end{rhp}
Step (iii),
it can be easily shown that
the solution to the new RHP~\ref{rhp:M1} tends to the identity matrix
exponentially fast as $t\to+\infty$,
just like what happens to RHP~\ref{rhp:MDSG}.
One obtains the final asymptotics
\myblue{
\begin{equation}
q(t,z) = \O(\ee^{-\aleph t})\,,\qquad
\brho(t,z) = D_-\sigma_3(\I + \O(\ee^{-\aleph t}))\,,\qquad t\to+\infty\,,
\end{equation}
with some positive constant $\aleph$.
}

\subsubsection{The case of $\xi > V_J$}
\label{s:N-soliton-asymptotics-stable-tpos-case2}

\begin{figure}[tp]
\centering
\includegraphics[width=0.31\textwidth]{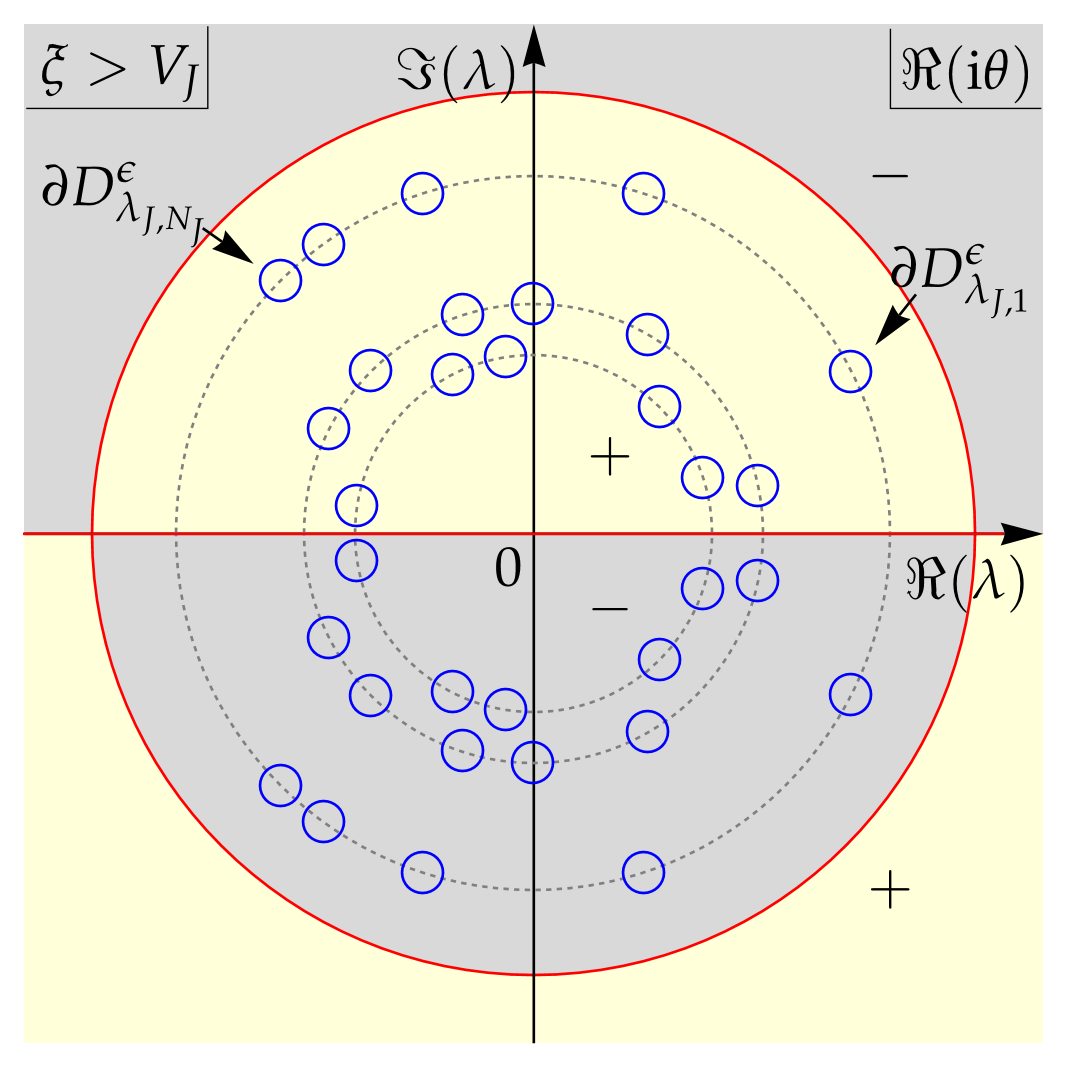}\quad
\includegraphics[width=0.31\textwidth]{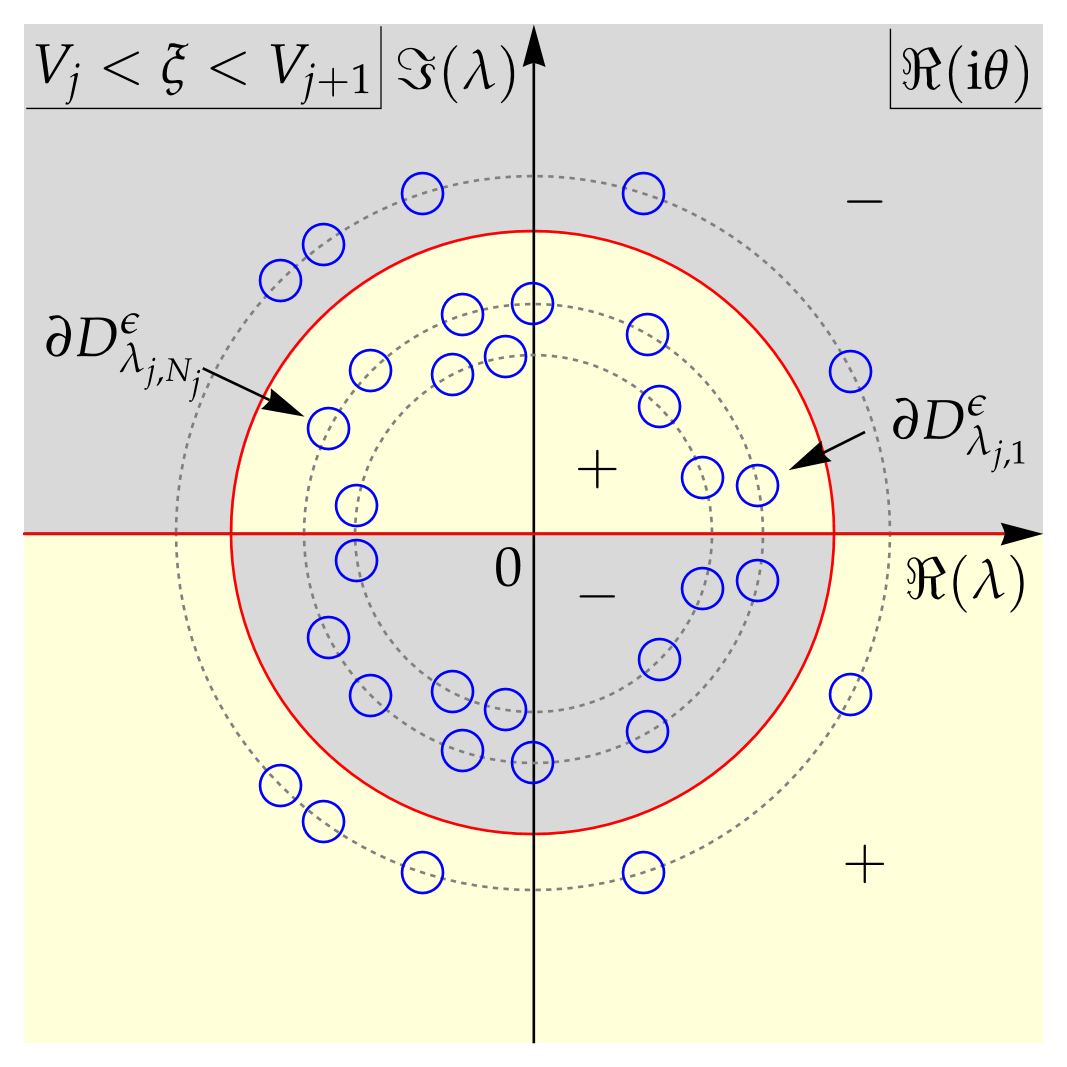}\quad
\includegraphics[width=0.31\textwidth]{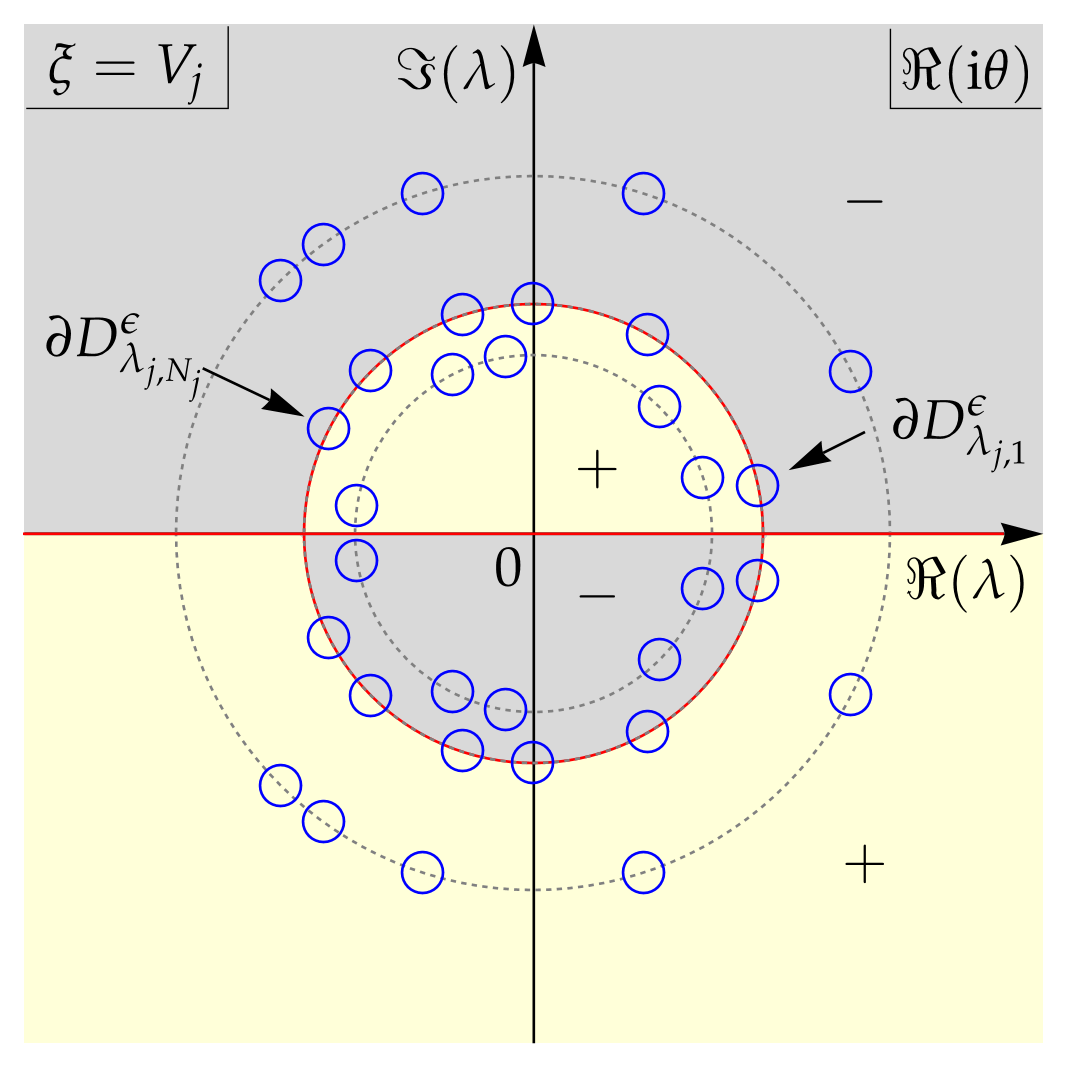}
\caption{Representative plots of jump configuration of RHP~\ref{rhp:N-soliton-jump-form} similarly to Figure~\ref{f:soliton-asymptotics-stable-tpos-case1} but for Cases 2--4 (from left to right) in Table~\ref{tab:soliton-asymptotics-stable-cases} with $t > 0$.}
\label{f:soliton-asymptotics-stable-tpos-case234}
\end{figure}
As shown in Figure~\ref{f:soliton-asymptotics-stable-tpos-case234}(left),
and as calculated by plugging $z = \xi t$ into $\Re(\ii\theta(\lambda;t,z))$,
with small enough $\epsilon$ in RHP~\ref{rhp:N-soliton-jump-form},
all the jump matrices $\V_{j,k}(\lambda;t,z)^{-1}$ in the upper half plane and
$\V_{j,k}(\conj\lambda;t,z)^\dagger$ in the lower half plane contain decaying exponentials,
\myblue{similarly to RHP~\ref{rhp:MDSG}}.
Taking the limit $t\to+\infty$, one directly arrives at the following result
\begin{equation}
\myblue{
q(t,z) = \O(\ee^{-\aleph t})\,,\qquad
\brho(t,z) = D_-\sigma_3(\I+\O(\ee^{-\aleph t}))\,,\qquad t\to+\infty\,.
}
\end{equation}
Recall that this also happens in Section~\ref{s:N-DSG-localization} in the case of $\xi > V$.

\subsubsection{The case of $V_j < \xi < V_{j+1}$}
\label{s:N-soliton-asymptotics-stable-tpos-case3}

Now, let us discuss a novel case (comparing to Section~\ref{s:N-DSG-localization}) where in RHP~\ref{rhp:N-soliton-jump-form} some jumps contain growing exponentials and others contain decaying exponentials. The configuration can be seen in Figure~\ref{f:soliton-asymptotics-stable-tpos-case234}(center).

In particular,
let us fix the value of $j$ with $1\le j\le J - 1$ in this section.
Then,
as can be calculated easily,
the jump matrices
$\V_{l,k}(\lambda;t,z)^{-1}$ and
$\V_{l,k}(\conj\lambda;t,z)^\dagger$ with
$1\le l\le j$ and $1 \le k \le N_l$ contain decaying exponentials,
whereas matrices
$\V_{l,k}(\lambda;t,z)^{-1}$ and
$\V_{l,k}(\conj\lambda;t,z)^\dagger$ with
$j < l \le J$ and $1 \le k \le N_l$ contain growing exponentials,
as $t\to+\infty$ in RHP~\ref{rhp:N-soliton-jump-form}.
Thus,
one can ignore the decaying jumps for now,
and focus on the growing jump.
Recall $\delta_j(\lambda)$ defined in Equation~\eqref{e:deltan-def}.
Then, one defines a new matrix function
\begin{equation}
\label{e:Mj2-def}
\M_j^{(2)}(\lambda;t,z)
 \coloneqq \begin{cases}
\M(\lambda;t,z) \delta_{j+1}(\lambda)^{-\sigma_3}\,, &\qquad \lambda\in\Complex\setminus\bigcup\limits_{l = j+1}^J\bigcup\limits_{k = 1}^{N_l}\big(D_{\lambda_{l,k}}^\epsilon\bigcup D_{\conj{\lambda_{l,k}}}^\epsilon\big)\,,\\
\M(\lambda;t,z)(\lambda - \conj{\lambda_{l,k}})^{\sigma_3}\,, &\qquad \lambda\in D_{\lambda_{l,k}}^\epsilon\,,\quad l = j+1,\dots,J\,,\quad k = 1,\dots, N_l\,,\\
\M(\lambda;t,z)(\lambda - \lambda_{l,k})^{-\sigma_3}\,, &\qquad \lambda\in D_{\conj{\lambda_{l,k}}}^\epsilon\,,\quad l = j+1,\dots,J\,,\quad k = 1,\dots, N_l\,.
 \end{cases}
\end{equation}
Hence, one can check that $\M_j^{(2)}(\lambda;t,z)$ solves the following RHP.
\begin{rhp}
\label{e:Mj2}
\myblue{
Seek a $2\times2$
matrix function $\M_j^{(2)}(\lambda;t,z)$ analytic on $\Complex\setminus\partial D^\epsilon$ with continuous boundary values,
where
$D^\epsilon\coloneqq \bigcup_{\lambda_{j,k}\in\Lambda}\big(D_{\lambda_{j,k}}^\epsilon\bigcup D_{\conj{\lambda_{j,k}}}^\epsilon\big)$.
It has the asymptotics $\M_j^{(2)}(\lambda;t,z)\to\I$ as $\lambda\to\infty$ and jumps,
for $1\le l\le j$ and $1\le k\le N_l$,
}
\begin{equation}
\begin{aligned}
\M_j^{(2)+}(\lambda;t,z)
 & = \M_j^{(2)-}(\lambda;t,z) \delta_{j+1}(\lambda)^{\sigma_3}\V_{l,k}(\lambda;t,z)^{-1}\delta_{j+1}(\lambda)^{-\sigma_3}\,,\quad && \lambda\in\partial D_{\lambda_{l,k}}^\epsilon\,,\\
\M_j^{(2)+}(\lambda;t,z)
 & = \M_j^{(2)-}(\lambda;t,z) \delta_{j+1}(\lambda)^{\sigma_3}\V_{l,k}(\conj\lambda;t,z)^\dagger\delta_{j+1}(\lambda)^{-\sigma_3}\,,\quad && \lambda\in\partial D_{\conj{\lambda_{l,k}}}^\epsilon\,,
\end{aligned}
\end{equation}
where $\V_{l,k}(\lambda;t,z)$ is given in RHP~\ref{rhp:N-soliton-jump-form}, and for $j+1\le l\le J$ and $1\le k\le N_l$,
\begin{equation}
\begin{aligned}
\M_j^{(2)+}(\lambda;t,z)
 & = \M_j^{(2)-}(\lambda;t,z)\W_{j,l,k}(\lambda;t,z)\,,\qquad &&
 \lambda\in \partial D_{\lambda_{l,k}}^\epsilon\,,\\
\M_j^{(2)+}(\lambda;t,z)
 & = \M_j^{(2)-}(\lambda;t,z)\W_{j,l,k}(\conj\lambda;t,z)^{-\dagger}\,,\qquad &&
 \lambda\in \partial D_{\conj{\lambda_{l,k}}}^\epsilon\,,
\end{aligned}
\end{equation}
with
\begin{equation}
\everymath{\displaystyle}
\W_{j,l,k}(\lambda;t,z)
 \coloneqq \delta_{j+1}(\lambda)^{\sigma_3}\V_{l,k}(\lambda;t,z)^{-1}(\lambda - \conj{\lambda_{l,k}})^{\sigma_3}
 = \bpm
    \delta_{j+1}(\lambda) (\lambda - \conj{\lambda_{l,k}}) &
    0 \\
    -\frac{\lambda - \conj{\lambda_{l,k}}}{\lambda - \lambda_{l,k}}\frac{\omega_{l,k}\ee^{-2\ii\theta(\lambda;t,z)}}{\delta_{j+1}(\lambda)} &
    \frac{1}{\delta_{j+1}(\lambda)(\lambda - \conj{\lambda_{l,k}})}
    \epm\,.
\end{equation}
\end{rhp}
Let us define another matrix $\M_j^{(3)}(\lambda;t,z)$ as follows,
\begin{equation}
\label{e:Mj3-def}
\M_j^{(3)}(\lambda;t,z)
 \coloneqq \begin{cases}
  \M_j^{(2)}(\lambda;t,z)\,, & \quad \lambda\in\Complex\setminus\bigcup\limits_{l = j+1}^J\bigcup\limits_{k = 1}^{N_l}\big(D_{\lambda_{l,k}}^\epsilon\bigcup D_{\conj{\lambda_{l,k}}}^\epsilon\big)\,,\\
  \M_j^{(2)}(\lambda;t,z)\Z_{j,l,k}(\lambda;t,z)^{-1}\,, &\quad \lambda\in D_{\lambda_{l,k}}^\epsilon\,,\quad j+1\le l\le J\,,\quad 1\le k\le N_l\,,\\
  \M_j^{(2)}(\lambda;t,z)\Z_{j,l,k}(\conj\lambda;t,z)^\dagger\,, &\quad \lambda\in D_{\conj{\lambda_{l,k}}}^\epsilon\,,\quad j+1\le l\le J\,,\quad 1\le k\le N_l\,,
\end{cases}
\end{equation}
where for $j+1\le l\le J$ and $1\le k\le N_l$, and one defines
\begin{equation}
\everymath{\displaystyle}
\Z_{j,l,k}(\lambda;t,z)
 \coloneqq
 \bpm
  \frac{\lambda - \conj{\lambda_{l,k}}}{p_{j+1}^*(\lambda)}\frac{p_{j+1}^*(\lambda)^2 - p_{j+1}^*(\lambda_{l,k})^2}{p_{j+1}(\lambda)} &
  \omega_{l,k}^{-1}\frac{\lambda - \lambda_{l,k}}{\lambda - \conj{\lambda_{l,k}}}\frac{p_{j+1}^*(\lambda_{l,k})^2}{p_{j+1}(\lambda)p_{j+1}^*(\lambda)}\ee^{2\ii\theta(\lambda;t,z)} \\
  -\omega_{l,k}\frac{p_{j+1}(\lambda)}{p_{j+1}^*(\lambda)}\frac{\lambda - \conj{\lambda_{l,k}}}{\lambda - \lambda_{l,k}}\ee^{-2\ii\theta(\lambda;t,z)} &
  \frac{1}{\lambda - \conj{\lambda_{l,k}}}\frac{p_{j+1}(\lambda)}{p_{j+1}^*(\lambda)}
 \epm\,.
\end{equation}
Similarly to Section~\ref{s:N-DSG-localization},
it is important to notice that the matrices
$\Z_{j,l,k}(\lambda;t,z)$ are analytic in $D_{\lambda_{l,k}}^\epsilon$,
and the matrices $\Z_{j,l,k}(\conj\lambda;t,z)^\dagger$ are analytic in $D_{\conj{\lambda_{l,k}}}^\epsilon$.
Therefore,
the new matrix $\M_j^{(3)}(\lambda;t,z)$ are still analytic in
$D_{\lambda_{l,k}}^\epsilon\bigcup D_{\conj{\lambda_{l,k}}}^\epsilon$
for all $1\le l\le J$ and $1\le k\le N_l$.
One can then verify that $\M_j^{(3)}(\lambda;t,z)$ solves the following RHP.
\begin{rhp}
\label{rhp:Mj3}
\myblue{
Seek a $2\times2$
matrix function $\M_j^{(3)}(\lambda;t,z)$ analytic on
$\Complex\setminus\partial D^\epsilon$ with continuous boundary values,
where $D^\epsilon$ is defined in RHP~\ref{e:Mj2}.
It has the asymptotics $\M_j^{(3)}(\lambda;t,z)\to\I$ as $\lambda\to\infty$ and jumps,
for $1\le l\le j$ and $1\le k\le N_l$
}
\begin{equation}
\begin{aligned}
\M_j^{(3)+}(\lambda;t,z)
 & = \M_j^{(3)-}(\lambda;t,z) \delta_{j+1}(\lambda)^{\sigma_3}\V_{l,k}(\lambda;t,z)^{-1}\delta_{j+1}(\lambda)^{-\sigma_3}\,,\quad && \lambda\in\partial D_{\lambda_{l,k}}^\epsilon\,,\\
\M_j^{(3)+}(\lambda;t,z)
 & = \M_j^{(3)-}(\lambda;t,z) \delta_{j+1}(\lambda)^{\sigma_3}\V_{l,k}(\conj\lambda;t,z)^\dagger\delta_{j+1}(\lambda)^{-\sigma_3}\,,\quad && \lambda\in\partial D_{\conj{\lambda_{l,k}}}^\epsilon\,,
\end{aligned}
\end{equation}
where the matrices $\V_{l,k}(\lambda;t,z)$ are given in RHP~\ref{rhp:N-soliton-jump-form}, and for $j+1\le l\le J$ and $1\le k\le N_l$
\begin{equation}
\begin{aligned}
\M_j^{(3)+}(\lambda;t,z)
 & = \M_j^{(3)-}(\lambda;t,z)\W_{j,l,k}^{(1)}(\lambda;t,z)\,,\qquad && \lambda\in\partial D_{\lambda_{l,k}}^\epsilon\,,\\
\M_j^{(3)+}(\lambda;t,z)
 & = \M_j^{(3)-}(\lambda;t,z)\W_{j,l,k}^{(1)}(\conj\lambda;t,z)^{-\dagger}\,,\qquad && \lambda\in\partial D_{\conj{\lambda_{l,k}}}^\epsilon\,,
\end{aligned}
\end{equation}
where $\W_{j,l,k}^{(1)}(\lambda;t,z)$ are given by
\begin{equation}
\everymath{\displaystyle}
\W_{j,l,k}^{(1)}(\lambda;t,z)
    \coloneqq \W_{j,l,k}(\lambda;t,z)\Z_{j,l,k}(\lambda;t,z)^{-1}
    = \bpm
        1 & -\frac{\lambda - \lambda_{l,k}}{\omega_{l,k}}\frac{p_{j+1}^*(\lambda_{l,k})^2}{p_{j+1}(\lambda)^2}\ee^{2\ii\theta(\lambda;t,z)} \\ 0 & 1
    \epm\,.
\end{equation}
\end{rhp}
It is easy to check that
all the jumps in RHP~\ref{rhp:Mj3} are decaying exponentially to the identity matrix $\I$
as $t\to+\infty$,
\myblue{similarly to RHP~\ref{rhp:MDSG}}.
Therefore,
the solution $\M_j^{(3)}(\lambda;t,z)$ to RHP~\ref{rhp:Mj3} can be approximated by
\begin{equation}
\label{e:Mj3-asymptotics}
\myblue{
\M_j^{(3)}(\lambda;t,z) = \I + \O(\ee^{-\aleph t})\,,\qquad t\to+\infty\,,
}
\end{equation}
where $\aleph$ is a positive constant.
Then,
tracing back all the deformations, one concludes that as $t\to+\infty$
and with $|\lambda|\gg1$ or $|\lambda| \ll 1$
\begin{equation}
\label{e:M-asymptotics-stable-tpos-case3}
\myblue{
\M(\lambda;t,z)
 = \M_j^{(2)}(\lambda;t,z)\delta_{j+1}^{\sigma_3}
 = \M_j^{(3)}(\lambda;t,z)\delta_{j+1}^{\sigma_3}
 = \delta_{j+1}(\lambda)^{\sigma_3}(\I + \O(\ee^{-\aleph t}))\,.
}
\end{equation}
Recall that $\delta_{j+1}(\lambda)\to1$ as $\lambda\to\infty$, and $\delta_{j+1}(\lambda)^{\sigma_3}$ commutes with $\sigma_3$.
The reconstruction formula from Lemma~\ref{thm:reconstruction} yields
\begin{equation}
\myblue{
q(t,z) = \O(\ee^{-\aleph t})\,,\qquad
\brho(t,z) = D_-\sigma_3(\I + \O(\ee^{-\aleph t}))\,,\qquad
t\to+\infty\,.
}
\end{equation}
Notice that the calculations in the current subsection are valid for
a fixed $j = 1,2,\dots ,J-1$,
so the Case 3 in Table~\ref{tab:soliton-asymptotics-stable-cases} is finished.

\subsubsection{The case of $\xi = V_j$}
\label{s:N-soliton-asymptotics-stable-tpos-case4}

Fixing the value of $j = 1,2,\dots, J$, we analyze Case~4 in Table~\ref{tab:soliton-asymptotics-stable-cases} in this subsection. The configuration is shown in Figure~\ref{f:soliton-asymptotics-stable-tpos-case234}(right). The difference between the current case and Case~3 in Section~\ref{s:N-soliton-asymptotics-stable-tpos-case3} is that the jumps on the contours $\partial D_{\lambda_{j,k}}^\epsilon$ and $\partial D_{\conj{\lambda_{j,k}}}^\epsilon$ for all $1\le k\le N_j$ do not decay to the identity matrix uniformly. Other than this, the two cases are quite alike, in the sense that, all jumps enclosing eigenvalues $\lambda_{l,k}$ and $\conj{\lambda_{l,k}}$ with $1 \le l\le j-1$ decay to the identity matrix uniformly, and the ones enclosing $\lambda_{l,k}$ and $\conj{\lambda_{l,k}}$ with $l > j$ grow uniformly, as $t\to+\infty$. As a result, both deformations~\eqref{e:Mj2-def} and~\eqref{e:Mj3-def} apply to the current case. Thus, we start the discussion of Case~4 with RHP~\ref{rhp:Mj3}. However, unlike the previous case, one cannot take the limit $t\to+\infty$ of RHP~\ref{rhp:Mj3} and arrive at Equation~\eqref{e:Mj3-asymptotics}, because the jumps surrounding eigenvalues $\lambda_{j,k}$ are oscillatory as $t$ changes. In fact, RHP~\ref{rhp:Mj3} do not have a limit in Case~4.

To continue our calculation, let us first build a parametrix for RHP~\ref{rhp:Mj3} by ignoring all decaying jumps.
\begin{rhp}[Jump form]
\label{rhp:Ps+j-jump-form}
\myblue{
Seek a $2\times2$
matrix function $\P_{\s+,j}(\lambda;t,z)$ analytic on
$\Complex\setminus\partial\bigcup_{k=1}^{N_j}\big(D_{\lambda_{j,k}}^\epsilon\bigcup D_{\conj{\lambda_{j,k}}}^\epsilon\big)$
with continuous boundary values.
It has the asymptotics $\P_{\s+,j}(\lambda;t,z)\to\I$ as $\lambda\to\infty$ and satisfies the following jumps,
}
\begin{equation}
\begin{aligned}
\P_{\s+,j}^{+}(\lambda;t,z)
 & = \P_{\s+,j}^{-}(\lambda;t,z) \delta_{j+1}(\lambda)^{\sigma_3}\V_{j,k}(\lambda;t,z)^{-1}\delta_{j+1}(\lambda)^{-\sigma_3}\,,\quad && \lambda\in\partial D_{\lambda_{j,k}}^\epsilon\,,\quad 1\le k\le N_j\,,\\
\P_{\s+,j}^{+}(\lambda;t,z)
 & = \P_{\s+,j}^{-}(\lambda;t,z) \delta_{j+1}(\lambda)^{\sigma_3}\V_{j,k}(\conj\lambda;t,z)^\dagger\delta_{j+1}(\lambda)^{-\sigma_3}\,,\quad && \lambda\in\partial D_{\conj{\lambda_{j,k}}}^\epsilon\,,\quad 1\le k\le N_j\,,
\end{aligned}
\end{equation}
where $\V_{j,k}(\lambda;t,z)$ is still given in RHP~\ref{rhp:N-soliton-jump-form}.
\end{rhp}
The subscript ``$\s+$" denotes the case in a stable medium with asymptotics $t\to+\infty$. For convenience, we explicitly express the jump matrices of RHP~\ref{rhp:Ps+j-jump-form} below
\begin{equation}
\everymath{\displaystyle}
\begin{aligned}
\delta_{j+1}(\lambda)^{\sigma_3}\V_{j,k}(\lambda;t,z)^{-1}\delta_{j+1}(\lambda)^{-\sigma_3}
    & = \bpm
        1 & 0 \\ -\frac{\omega_{j,k}\delta_{j+1}(\lambda)^{-2}}{\lambda - \lambda_{j,k}}\ee^{-2\ii\theta(\lambda;t,z)} & 1
    \epm\,,\\
\delta_{j+1}(\lambda)^{\sigma_3}\V_{j,k}(\conj\lambda;t,z)^\dagger\delta_{j+1}(\lambda)^{-\sigma_3}
    & = \bpm
        1 & \frac{\conj{\omega_{j,k}}\delta_{j+1}(\lambda)^2}{\lambda - \conj{\lambda_{j,k}}}\ee^{2\ii\theta(\lambda;t,z)} \\ 0 & 1
    \epm\,.
\end{aligned}
\end{equation}
Note that the $(2,1)$ entry of
$\delta_{j+1}(\lambda)^{\sigma_3}\V_{j,k}(\lambda;t,z)^{-1}\delta_{j+1}(\lambda)^{-\sigma_3}$
and the $(1,2)$ entry of
$\delta_{j+1}(\lambda)^{\sigma_3}\V_{j,k}(\conj\lambda;t,z)^\dagger\delta_{j+1}(\lambda)^{-\sigma_3}$
have simple poles at $\lambda = \lambda_{j,k}$ and $\lambda = \conj{\lambda_{j,k}}$,
respectively.
Hence,
one can convert the above RHP with jumps into an equivalent RHP with residue conditions,
just like the equivalence between RHPs~\ref{rhp:N-soliton-residue-form} and~\ref{rhp:N-soliton-jump-form}.
Comparing the residue form of RHP~\ref{rhp:Ps+j-jump-form} with RHP~\ref{rhp:N-soliton-residue-form}, one immediately realizes that RHP~\ref{rhp:Ps+j-jump-form} describes a DSG containing $N_j$ solitons. The eigenvalues are still the same $\Lambda_j$ in the upper half plane, but the corresponding norming constants are modified as the set $\Omega_j^{(+)}$ containing $\omega_{j,k}^{(+)}$. Therefore, the parametrix $\P_{\s+,j}(\lambda;t,z)$ is solvable, and is bounded as $t\to+\infty$ via Theorem~\ref{thm:N-DSG} with $z = V_j t$.

One looks at the difference between the original matrix function $\M^{(3)}(\lambda;t,z)$ from RHP~\ref{rhp:Mj3} and its parametrix $\P_{\s+,j}(\lambda;t,z)$ from RHP~\ref{rhp:Ps+j-jump-form}, by defining an error function
\begin{equation}
\label{e:Es+-def}
\E_{\s+,j}(\lambda;t,z) \coloneqq \M_j^{(3)}(\lambda;t,z) \P_{\s+,j}(\lambda;t,z)^{-1}\,,\qquad
\mbox{for a fixed } j = 1,2,\dots, J\,.
\end{equation}
Then, RHPs~\ref{rhp:Mj3} and~\ref{rhp:Ps+j-jump-form} yield a RHP for the error function $\E_{\s+,j}(\lambda;t,z)$.
\begin{rhp}
\label{rhp:Es+j}
\myblue{
Seek a $2\times2$
matrix function $\E_{\s+,j}(\lambda;t,z)$ analytic on
$\Complex\setminus\partial D^\epsilon$ with continuous boundary values,
where
$D^\epsilon\coloneqq \bigcup_{l\ne j}\bigcup_{k=1}^{N_l}\big(D_{\lambda_{l,k}}^\epsilon\bigcup D_{\conj{\lambda_{l,k}}}^\epsilon\big)$.
It has asymptotics $\E_{\s+,j}(\lambda;t,z)\to\I$ as $\lambda\to\infty$ and satisfies jumps,
for $1\le l< j$ and $1\le k\le N_l$
}
\begin{equation}
\begin{aligned}
\E_{\s+,j}^{+}(\lambda;t,z)
 & = \E_{\s+,j}^{-}(\lambda;t,z)\P_{\s+,j} \delta_{j+1}(\lambda)^{\sigma_3}\V_{l,k}(\lambda;t,z)^{-1}\delta_{j+1}(\lambda)^{-\sigma_3}\P_{\s+,j}^{-1}\,,\quad && \lambda\in\partial D_{\lambda_{l,k}}^\epsilon\,,\\
\E_{\s+,j}^{+}(\lambda;t,z)
 & = \E_{\s+,j}^{-}(\lambda;t,z)\P_{\s+,j} \delta_{j+1}(\lambda)^{\sigma_3}\V_{l,k}(\conj\lambda;t,z)^\dagger\delta_{j+1}(\lambda)^{-\sigma_3}\P_{\s+,j}^{-1}\,,\quad && \lambda\in\partial D_{\conj{\lambda_{l,k}}}^\epsilon\,,
\end{aligned}
\end{equation}
and for $j+1\le l\le J$ and $1\le k\le N_l$
\begin{equation}
\begin{aligned}
\E_{\s+,j}^{+}(\lambda;t,z)
 & = \E_{\s+,j}^{-}(\lambda;t,z)\P_{\s+,j}\W_{j,l,k}^{(1)}(\lambda;t,z)\P_{\s+,j}^{-1}\,,\qquad && \lambda\in\partial D_{\lambda_{l,k}}^\epsilon\,,\\
\E_{\s+,j}^{+}(\lambda;t,z)
 & = \E_{\s+,j}^{-}(\lambda;t,z)\P_{\s+,j}\W_{j,l,k}^{(1)}(\conj\lambda;t,z)^{-\dagger}\P_{\s+,j}^{-1}\,,\qquad && \lambda\in\partial D_{\conj{\lambda_{l,k}}}^\epsilon\,.
\end{aligned}
\end{equation}
\end{rhp}
Note that $\P_{\s+,j}(\lambda;t,z)$ is bounded as $t\to+\infty$.
Thus,
all jump matrices for $\E_{\s+,j}(\lambda;t,z)$ in RHP~\ref{rhp:Es+j}
decay to the identity matrix exponentially and uniformly as $t\to+\infty$,
\myblue{just like what happens to RHP~\ref{rhp:MDSG}}.
Consequently,
one writes
\begin{equation}
\myblue{
\E_{\s+,j}(\lambda;t,z) = \I + \\O(\ee^{-\aleph t})\,,\qquad t\to+\infty\,,
}
\end{equation}
for every fixed $j = 1,2,\dots,J$ and with a positive constant $\aleph$.
Because of the boundedness of $\P_{\s+,j}(\lambda;t,z)$, the definition~\eqref{e:Es+-def} yields
\begin{equation}
\myblue{
\M_j^{(3)}(\lambda;t,z) = (\I + \O(\ee^{-\aleph t}))\P_{\s+,j}(\lambda;t,z)\,,\qquad t\to+\infty\,.
}
\end{equation}
Therefore, the solution to RHP~\ref{rhp:N-soliton-jump-form} can be reconstructed similarly to Equation~\eqref{e:M-asymptotics-stable-tpos-case3}, as follows,
\begin{equation}
\label{e:M-asymptotics-stable-tpos-case4}
\myblue{
\M(\lambda;t,z)
 = (\I + \O(\ee^{-\aleph t}))\P_{\s+,j}(\lambda;t,z)\delta_{j+1}(\lambda)^{\sigma_3}\,,\qquad
  |\lambda|\gg1\mbox{ or } |\lambda| \ll 1\,.
}
\end{equation}
Recall that $\delta_{j+1}(\lambda)\to1$ as $\lambda\to\infty$, and $\delta_{j+1}(\lambda)^{\sigma_3}$ commutes with $\sigma_3$. So, by Lemma~\ref{thm:reconstruction} and Theorem~\ref{thm:N-DSG} one concludes
\begin{equation}
\myblue{
\begin{aligned}
q(t,z;\Lambda,\Omega) & = q(t,z;\Lambda_j,\Omega_j^{(+)}) + \O(\ee^{-\aleph t})\,,\\
D(t,z;\Lambda,\Omega) & = D(t,z;\Lambda_j,\Omega_j^{(+)}) + \O(\ee^{-\aleph t})\,,\\
P(t,z;\Lambda,\Omega) & = P(t,z;\Lambda_j,\Omega_j^{(+)}) + \O(\ee^{-\aleph t})\,,
\end{aligned}
}
\end{equation}
where the leading-order term
\myblue{
$\{q(t,z;\Lambda_j,\Omega_j^{(+)}), D(t,z;\Lambda_j,\Omega_j^{(+)}), P(t,z;\Lambda_j,\Omega_j^{(+)})\}$
}
is an $N_j$-DSG derived from RHP~\ref{rhp:Ps+j-jump-form} with eigenvalues $\Lambda_j$
and modified norming constants $\Omega_j^{(+)}$
\myblue{
via Theorem~\ref{thm:N-DSG},
and the quantity $V_j$ is nothing else but the velocity of the DSG.
}

\subsection{Soliton asymptotics as $t\to-\infty$ in a stable medium}
\label{s:N-soliton-asymptotics-stable-tneg}

\begin{figure}[tp]
\centering
\includegraphics[width=0.31\textwidth]{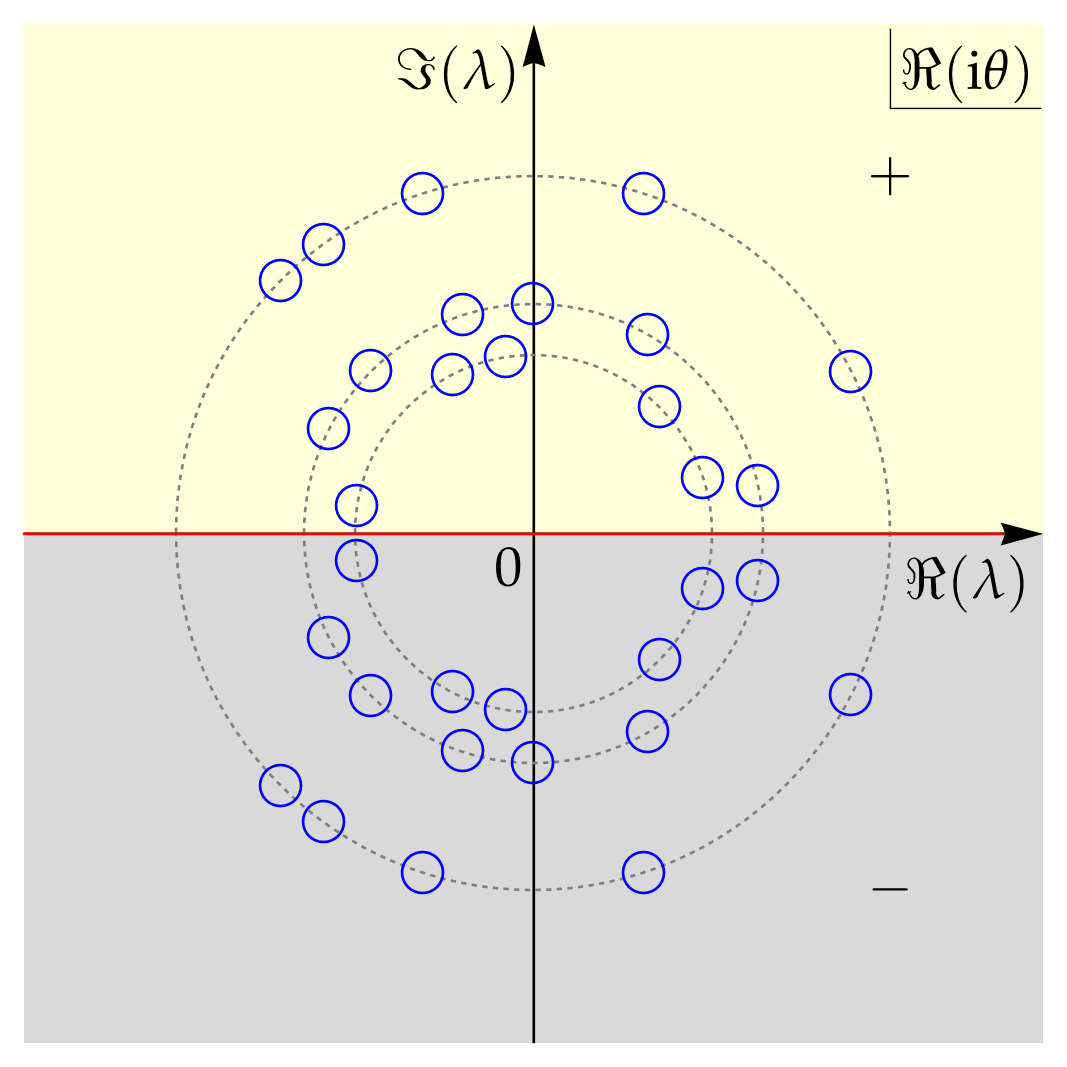}\qquad
\includegraphics[width=0.31\textwidth]{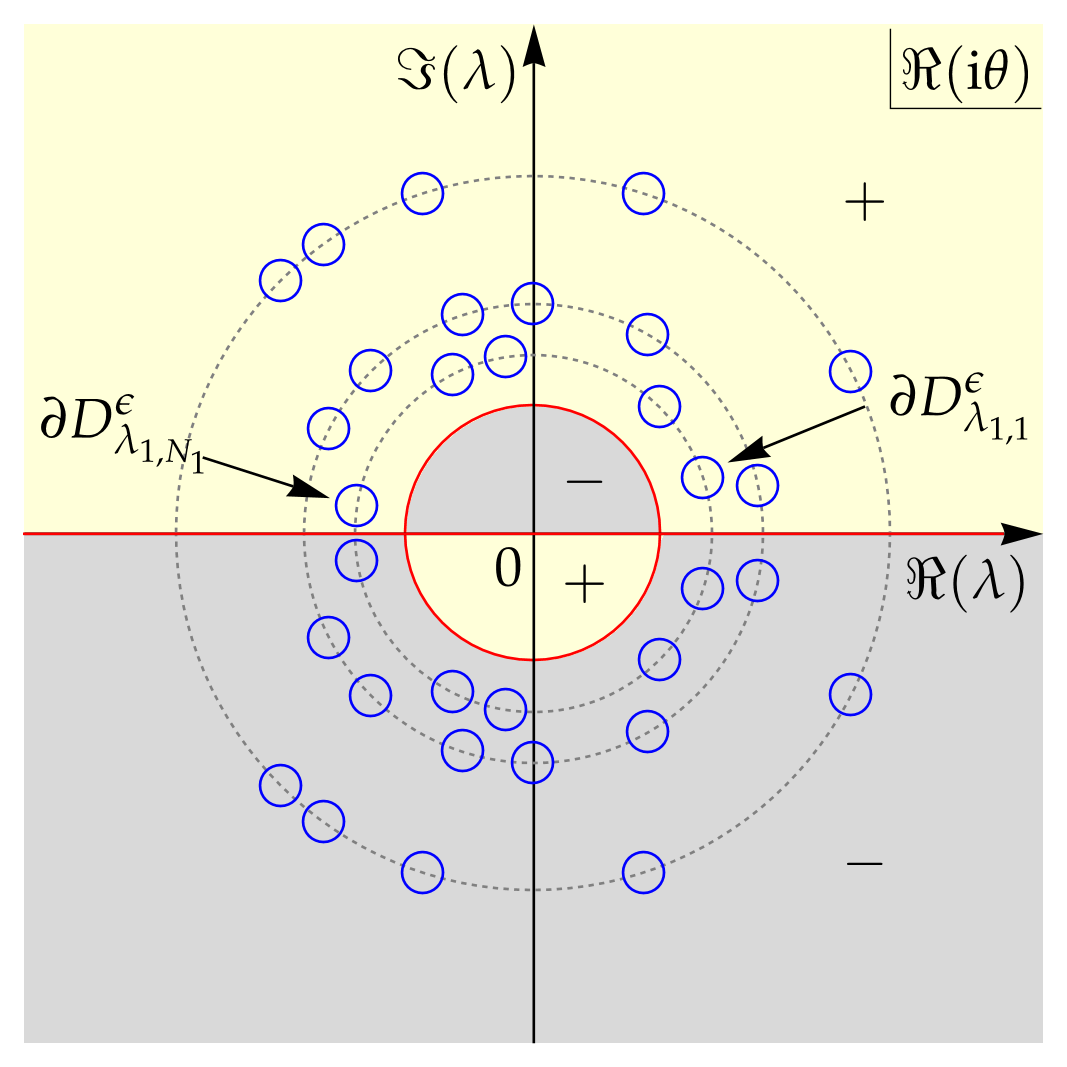}\\[1ex]
\includegraphics[width=0.31\textwidth]{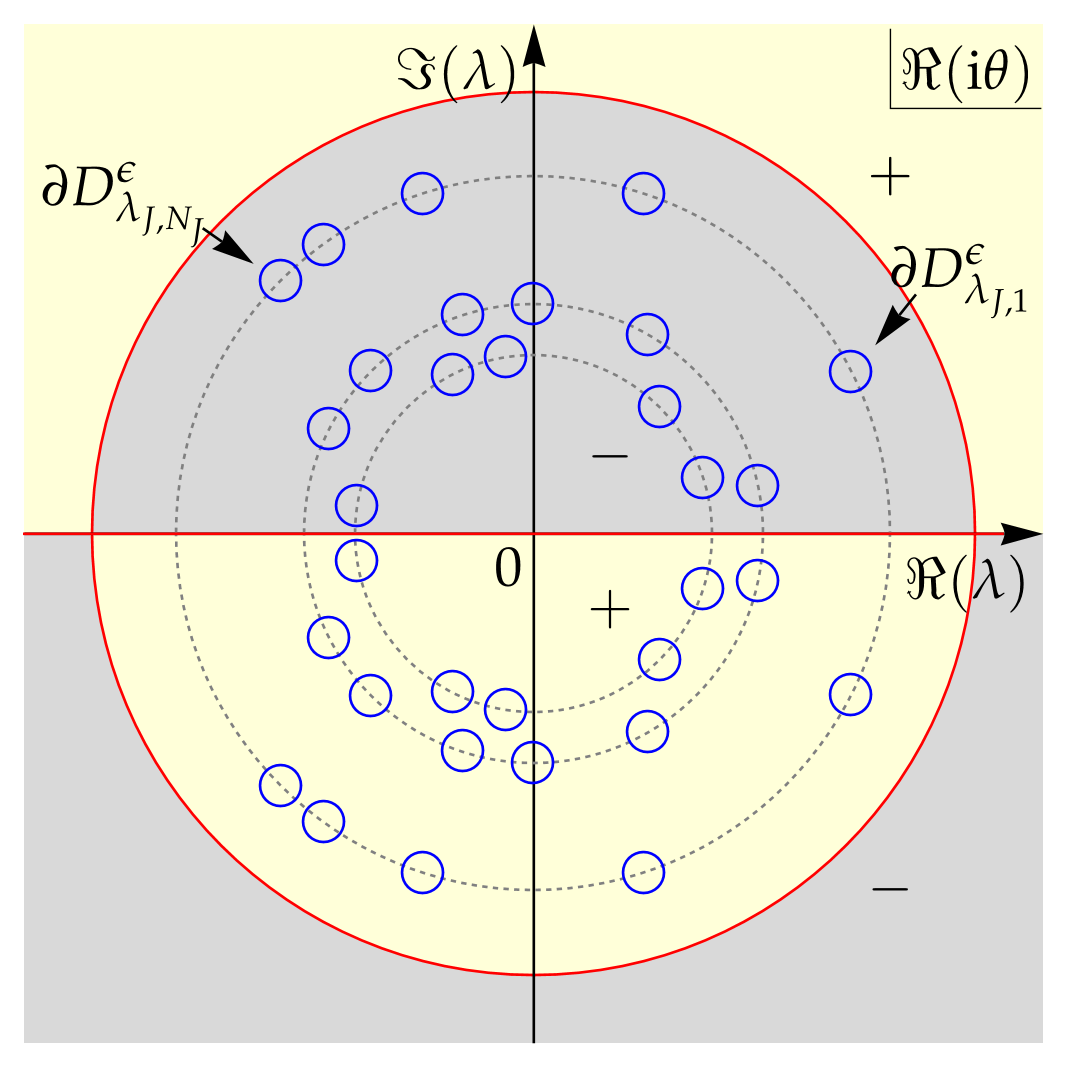}\quad
\includegraphics[width=0.31\textwidth]{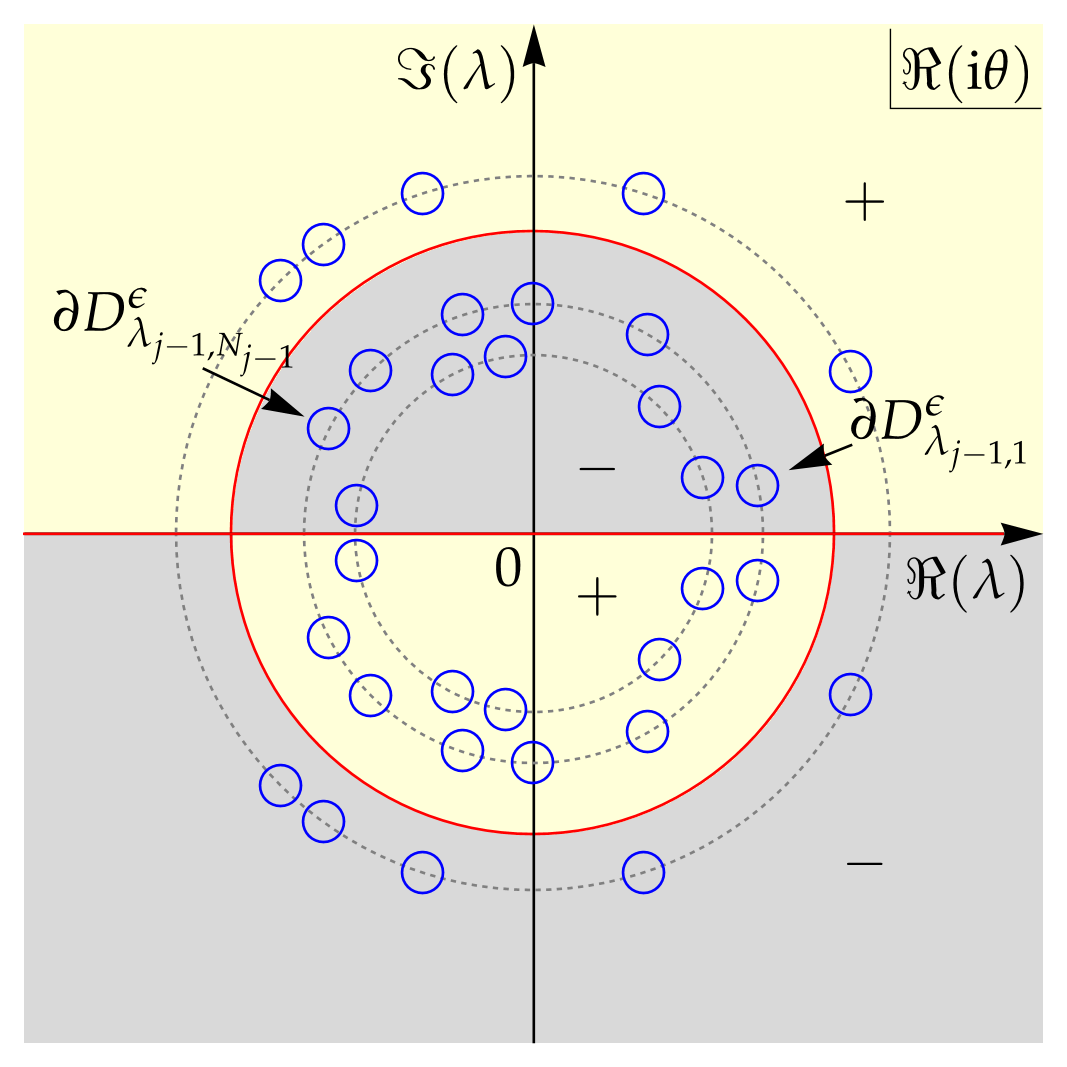}\quad
\includegraphics[width=0.31\textwidth]{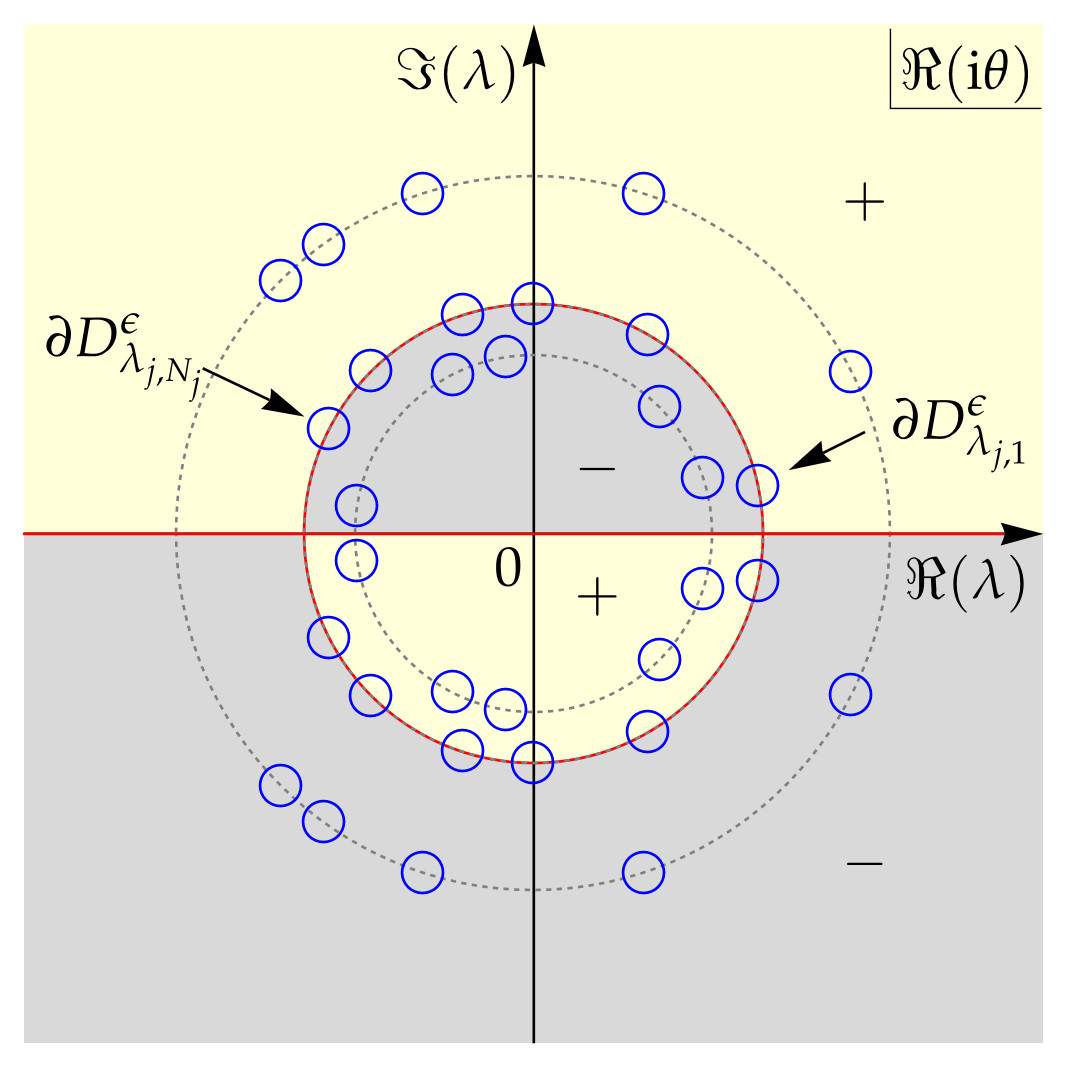}
\caption{Similarly to Figures~\ref{f:soliton-asymptotics-stable-tpos-case1} and~\ref{f:soliton-asymptotics-stable-tpos-case234}, but with $t < 0$ representing the four cases in Table~\ref{tab:soliton-asymptotics-stable-cases} .
Top row: situations $\xi \le 0 < V_1$ (left) and $0 < \xi <V_1$ (right) corresponding to Case 1.
Bottom left: $\xi > V_J$ corresponding to Case 2.
Bottom center: $V_{j-1} < \xi < V_j$ slightly different but equivalent to Case 3.
Bottom right: $\xi = V_j$ corresponding to Case 4.
}
\label{f:soliton-asymptotics-stable-tneg-case1234}
\end{figure}
One needs to consider the four cases in Table~\ref{tab:soliton-asymptotics-stable-cases},
but with the limit $t\to-\infty$.
Because $t < 0$ instead of $t > 0$ from Section~\ref{s:N-soliton-asymptotics-stable-tpos},
and because of $\theta(\lambda; t,z) = \lambda t + \xi t/(2\lambda)$,
the quantity $\Re(\ii\theta(t,z))$ flips its sign comparing to
the previous $t\to+\infty$ discussion.

The illustrative plots are given in Figure~\ref{f:soliton-asymptotics-stable-tneg-case1234},
which should be compared with
Figures~\ref{f:soliton-asymptotics-stable-tpos-case1} and~\ref{f:soliton-asymptotics-stable-tpos-case234}.
It is then necessary to notice the similarities and differences
between the situations $t\to-\infty$ and $t\to+\infty$.
Built on the detailed discussion on the case of $t\to+\infty$
in Section~\ref{s:N-soliton-asymptotics-stable-tpos},
the results of the $t\to-\infty$ case can be easily obtained,
in similar treatments.
As such, most details will be omitted in order to avoid repetitions.
\myblue{
In fact, we omit the calculations of Cases 1--3 completely,
and only show essential steps in Case 4
from Table~\ref{tab:soliton-asymptotics-stable-cases}.
}

\subsubsection{The case of $\xi \ne V_j$ for all $1\le j\le J$}

\myblue{
Mimicking similarly calculations from
Sections~\ref{s:N-soliton-asymptotics-stable-tpos-case1}--\ref{s:N-soliton-asymptotics-stable-tpos-case3},
one obtains
\begin{equation}
q(t,z) = \O(\ee^{-\aleph |t|})\,,\qquad
\brho(t,z) = D_-\sigma_3 + \O(\ee^{-\aleph |t|})\,,\qquad
t\to-\infty\,,
\end{equation}
with a positive constant $\aleph$.
}

\subsubsection{The case of $\xi = V_j$}
\label{s:N-soliton-asymptotics-stable-tneg-case4}

First, let us fix the value of $j = 1,2,\dots, J$.
According to Figure~\ref{f:soliton-asymptotics-stable-tneg-case1234}(bottom right),
this case is analogous to the one in Section~\ref{s:N-soliton-asymptotics-stable-tpos-case4},
in the sense that all the jumps of RHP~\ref{rhp:N-soliton-jump-form} with indices $l < j$ are growing as $t\to-\infty$
and the ones with $l > j$ are decaying to the identity matrix in the asymptotics.
\myblue{
With similar treatments for the growing jumps in Section~\ref{s:N-soliton-asymptotics-stable-tpos-case4},
one arrives at a RHP with all jumps indexed $l \ne j$ decaying exponentially to the identity matrix as $t\to-\infty$.
Ignoring all decaying jumps yields the following leading-order problem.
}
\begin{rhp}[Jump form]
\label{rhp:Ps-j-jump}
\myblue{
Let $j = 1,2,\dots, J$ be fixed.
Seek a $2\times2$
matrix function $\P_{\s-,j}(\lambda;t,z)$ analytic on
$\Complex\setminus\partial \bigcup_{k=1}^{N_j}\big(D_{\lambda_{j,k}}^\epsilon\bigcup D_{\conj{\lambda_{j,k}}}^\epsilon\big)$,
with continuous boundary values.
It has the asymptotics $\P_{\s-,j}(\lambda;t,z)\to\I$ as $\lambda\to\infty$,
and satisfies jumps
}
\begin{equation}
\everymath{\displaystyle}
\begin{aligned}
\P_{\s-,j}^{+}
 & = \P_{\s-,j}^{-}\delta_1^{\sigma_3}\delta_j^{-\sigma_3}\V_{j,k}(\lambda;t,z)^{-1}\delta_j^{\sigma_3}\delta_1^{-\sigma_3}
 = \P_{\s-,j}^{-}\bpm 1 & 0 \\ -\frac{\omega_{j,k}}{\lambda - \lambda_{j,k}}\frac{\delta_j(\lambda)^{2}}{\delta_1(\lambda)^{2}}\ee^{-2\ii\theta} & 1 \epm \,,\qquad
 \lambda\in\partial D_{\lambda_{j,k}}^\epsilon\,,\\
\P_{\s-,j}^{+}
 & = \P_{\s-,j}^{-}\delta_1^{\sigma_3}\delta_j^{-\sigma_3}\V_{j,k}(\conj\lambda;t,z)^\dagger\delta_j^{\sigma_3}\delta_1^{-\sigma_3}
 = \P_{\s-,j}^{-}\bpm 1 & \frac{\conj{\omega_{j,k}}}{\lambda - \conj{\lambda_{j,k}}}\frac{\delta_1(\lambda)^2}{\delta_j(\lambda)^2}\ee^{2\ii\theta} \\ 0 & 1 \epm\,,\qquad
 \lambda\in\partial D_{\conj{\lambda_{j,k}}}^\epsilon\,,\\
\end{aligned}
\end{equation}
\end{rhp}
The subscript ``$\s-$" denote the asymptotics in a stable medium as $t\to-\infty$.
If needed, the above RHP can be turned in to a residue form.
%
\myblue{
Again,
omitting details,
RHP~\ref{rhp:Ps-j-jump} together with
Lemma~\ref{thm:reconstruction} and Theorem~\ref{thm:N-DSG} yields
}
\begin{equation}
\myblue{
\begin{aligned}
q(t,z;\Lambda,\Omega) & = q(t,z;\Lambda_j,\Omega_j^{(-)}) + \O(\ee^{-\aleph |t|})\,,\\
D(t,z;\Lambda,\Omega) & = D(t,z;\Lambda_j,\Omega_j^{(-)}) + \O(\ee^{-\aleph |t|})\,,\\
P(t,z;\Lambda,\Omega) & = P(t,z;\Lambda_j,\Omega_j^{(-)}) + \O(\ee^{-\aleph |t|})\,,
\end{aligned}
}
\end{equation}
where
\myblue{the leading-order term
$\{q(t,z;\Lambda_j,\Omega_j^{(-)}), D(t,z;\Lambda_j,\Omega_j^{(-)}), P(t,z;\Lambda_j,\Omega_j^{(-)})\}$ is an $N_j$-DSG derived from RHP~\ref{rhp:Ps-j-jump} with eigenvalues $\Lambda_j$ and modified norming constants $\Omega_j^{(-)} \coloneqq \{\omega_{j,k}^{(-)}\}$
}
defined in Theorem~\ref{thm:soliton-asymptotics}.
\myblue{The quantity $V_j$ is the DSG velocity.}

\subsection{Soliton asymptotics in an unstable medium}

It is time to discuss the soliton asymptotics in an unstable medium with $J \ge 2$.
\myblue{
We do so by comparison between the stable medium case $D_- = -1$ and unstable medium case $D_-=1$.
To distinguish similar quantities,
let us denote $V_{\s, j} = 2r_j^2$ in a stable medium and $V_{\u,j} = -2r_j^2 = -V_{\s,j}$ in an unstable medium,
as the two realizations of $V_j$ defined in Equation~\eqref{e:Vj-def}.
}
The crucial component in the Deift-Zhou's nonlinear steepest descent method is the sign of $\Re(\ii\theta(\lambda_{j,k};t,z))$, which are given in the stable and unstable case, respectively, as follows,

\begin{equation}
\begin{aligned}
\Re(\ii\theta_{\s,j,k}(\xi,t))
 = \frac{y_{j,k}}{2r_j^2}(\xi - V_{\s,j})t\,,
\qquad
\Re(\ii\theta_{\u,j,k}(\xi,t))
 = -\frac{y_{j,k}}{2r_j^2}(\xi - V_{\u,j})t\,,
\end{aligned}
\end{equation}
where $y_{j,k} = \Im(\lambda_{j,k}) > 0$ and $z = \xi t$.
As a result, it can be seen
\begin{equation}
\label{e:reitheta-unstable-stable-relation}
\Re(\ii\theta_{\u,j,k}(\xi,t))
 = \Re(\ii\theta_{\s,j,k}(-\xi,t))\,.
\end{equation}
Therefore, the asymptotics in an unstable medium along the line $z = \xi t$ can be regarded as asymptotics in a stable medium along a different line $z = -\xi t$. This relation allows us to compute the asymptotics in the unstable case easily from known results of the stable case.

Clearly, the quantities $V_{\u,j} = -2r_j^2$ defined in Theorem~\ref{thm:N-DSG} are ordered below,
\begin{equation}
V_J < V_{J-1} < \dots < V_1 < 0\,.
\end{equation}
Similarly to what happens in a stable case,
there are four fundamental cases given in Table~\ref{tab:soliton-asymptotics-unstable-cases}.
\begin{table}[t]
\caption{Four cases in the long-time asymptotics in an unstable medium.}
\begin{tabular}{|c|c|c|}
\hline
Unstable Cases & $\vphantom{\Big|}$\quad Relations between $\xi$ and $V_{\u,j}$ \quad & \hspace{5em} Description \hspace{5em} \\
\hline\hline
Case 1 & $\vphantom{\Big|}\xi < V_{\u,J} < \dots < V_{\u,1}$ & \myblue{$\xi$ is the smallest} \\
\hline
Case 2 & $\vphantom{\Big|}V_{\u,J} < \dots < V_{\u,1} < \xi$ & \myblue{$\xi$ is the largest} \\
\hline
Case 3 & $\vphantom{\Big|}V_{\u,j+1} < \xi < V_{\u,j}$ with $j = 1,\dots,J-1$ & \myblue{$\xi$ is in-between $V_{\u,j+1}$ and $V_{\u,j}$} \\
\hline
Case 4 & $\vphantom{\Big|}V_{\u,j+1} < \xi = V_{\u,j} < V_{\u,j-1}$ with $j = 1,\dots,J$ & \myblue{$\xi$ coincides with some $V_{\u,j}$} \\
\hline
\end{tabular}
\label{tab:soliton-asymptotics-unstable-cases}
\end{table}

In the asymptotics $t\to+\infty$, one can relate the cases between stable medium (cf. Table~\ref{tab:soliton-asymptotics-stable-cases}) and unstable medium (cf. Table~\ref{tab:soliton-asymptotics-unstable-cases}) below.
\begin{enumerate}[align=left]
\item [\bf Unstable Case 1]
Equation~\eqref{e:reitheta-unstable-stable-relation} implies that
the unstable Case 1 ($\xi < V_{\u,J} < \dots < V_{\u,1}$)
can be calculated equivalently to the stable Case 2 ($V_{\s, 1} < \dots < V_{\s,J} < -\xi$)
in the same asymptotics $t\to+\infty$,
where $z = \xi t$ with $\xi < 0$.
Hence,
\myblue{
one arrives at $q(t,z) = \O(\ee^{-\aleph t})$ and $\brho(t,z) = D_-\sigma_3(\I+\O(\ee^{-\aleph t}))$,
with a positive constant $\aleph$,
}
following Section~\ref{s:N-soliton-asymptotics-stable-tpos-case2}.
\item [\bf Unstable Case 2]
Equation~\eqref{e:reitheta-unstable-stable-relation} implies that
the unstable Case 2 ($V_{\u,J} < \dots < V_{\u,1} < \xi$)
is equivalent to the stable Case 1 ($-\xi < V_{\s,1} < \dots < V_{\s, J}$).
Therefore,
one arrives at the same asymptotic result
\myblue{$q(t,z) = \O(\ee^{-\aleph t})$ and $\brho(t,z) = D_-\sigma_3(\I+\O(\ee^{-\aleph t}))$}
as $t\to+\infty$ following Section~\ref{s:N-soliton-asymptotics-stable-tpos-case1}.
\item [\bf Unstable Case 3]
Equation~\eqref{e:reitheta-unstable-stable-relation} implies that
the unstable Case 3 ($V_{\u,j+1} < \xi < V_{\u,j}$)
is equivalent to the stable Case 3 ($V_{\s,j} < -\xi < V_{\s,j+1}$).
Therefore,
one arrives at the same asymptotic result
\myblue{$q(t,z) = \O(\ee^{-\aleph t})$ and $\brho(t,z) = D_-\sigma_3(\I+\O(\ee^{-\aleph t}))$}
as $t\to+\infty$ following Section~\ref{s:N-soliton-asymptotics-stable-tpos-case3}.
\item [\bf Unstable Case 4]
Equation~\eqref{e:reitheta-unstable-stable-relation} implies that the unstable
Case 4 ($\xi = V_{\u,j}$) is equivalent to the stable Case 4 ($-\xi = V_{\s,j}$). Therefore, by repeating similar calculations as in Section~\ref{s:N-soliton-asymptotics-stable-tpos-case4}, one arrives at the formula for the asymptotic solution with $\Lambda_j$ and $\Omega_j^{(+)}$ in Theorem~\ref{thm:soliton-asymptotics}.
\end{enumerate}
The other asymptotics $t\to-\infty$ between the stable and unstable cases can also be related. The details are omitted for brevity. The asymptotic DSG is given in Theorem~\ref{thm:soliton-asymptotics}.
Finally, this completes the proof of Theorem~\ref{thm:soliton-asymptotics} for all cases.

\section{High-order solitons}
\label{s:N-order-soliton}

In this section, we discuss various aspects about the high-order solitons. First, we show how to fuse $N\ge2$ simple poles of RHP~\ref{rhp:N-soliton-residue-form} in order to get an $N$th order pole. As a result, the corresponding $N$-soliton solution merge into a single $N$th-order soliton.

\subsection{Derivation of an $N$th order soliton from fusion}
\label{s:N-order-soliton:fusion}

Suppose RHP~\ref{rhp:N-order-soliton:jump-form} with the eigenvalue $\lambda_\circ$ and norming constants $\{\omega_{\circ,k}\}_{k=0}^{N-1}$ given. We would like to find an explicit limiting procedure, such that when merging all eigenvalues $\Lambda$ from RHP~\ref{rhp:N-soliton-jump-form}, the $N$ simple poles $\omega_{j,k}/(\lambda-\lambda_{j,k})$ fuse into the $N$th-order pole $p_\circ(\lambda)/(\lambda - \lambda_\circ)^N$ appearing in RHP~\ref{rhp:N-order-soliton:jump-form}.

Therefore, we start with RHP~\ref{rhp:N-soliton-jump-form} for the $N$-soliton solutions. To better take limits, we rewrite this RHP in an equivalent form, in which the jump contours $\partial D_{\lambda_{j,k}}^\epsilon$ are modified into $\partial D_{\Lambda}^\epsilon$, which is a single contour surrounding all eigenvalues in the upper half plane. Moreover, we would like to relabel all eigenvalues as $\lambda_{j,k}\mapsto\lambda_s$ for $s = 1,2,\dots,N$, and same for the norming constants $\omega_{j,k}\mapsto\omega_s$. As a result, the new RHP is given by below.
\begin{rhp}[Equivalent jump-form of RHP~\ref{rhp:N-soliton-jump-form}]
\label{rhp:N-soliton-jump-form-bigcircle}
\myblue{
Seek a $2\times2$ matrix function $\M(\lambda;t,z)$ analytic on
$\Complex\setminus\partial\big(D_{\Lambda}^\epsilon\bigcup D_{\conj{\Lambda}}^\epsilon\big)$
with continuous boundary values.
It has the asymptotics $\M(\lambda;t,z)\to\I$ as $\lambda\to\infty$ and jumps
}
\begin{equation}
\begin{aligned}
\M^+(\lambda;t,z) & = \M^-(\lambda;t,z)\V(\lambda;t,z)^{-1}\,,\qquad & \partial D_\Lambda^\epsilon\,,\\
\M^+(\lambda;t,z) & = \M^-(\lambda;t,z)\V(\conj\lambda;t,z)^\dagger\,,\qquad & \partial D_{\conj{\Lambda}}^\epsilon\,,
\end{aligned}
\end{equation}
where
\begin{equation}
\everymath{\displaystyle}
\myblue{
\V(\lambda;t,z)
\coloneqq \bpm 1 & 0 \\ \sum\limits_{s=1}^N \frac{\omega_s}{\lambda - \lambda_s}\ee^{-2\ii\theta(\lambda;t,z)} & 1 \epm\,.
}
\end{equation}
\end{rhp}
Recall that we would like to take limits $\lambda_s\to\lambda_\circ$, so it is natural to write
\begin{equation}
\lambda_s = \lambda_\circ + \epsilon_s\,,
\end{equation}
where the target eigenvalue $\lambda_\circ$ is a fixed complex number in the upper half plane.
\myblue{
The distinct eigenvalues $\{\lambda_s\}$ require that the complex quantities $\{\epsilon_s\}$ are distinct in the following calculations.
}
As a result, the relevant quantity in the above RHP becomes
\begin{equation}
\label{e:N-order-soliton:pole-sum}
\sum_{s = 1}^N \frac{\omega_s}{\lambda - \lambda_s}
 = \sum_{s = 1}^N \frac{\omega_s}{\lambda - \lambda_\circ - \epsilon_s}\,.
\end{equation}

We next show that the fusion of eigenvalues must be done in a special way. Otherwise, the result is quite trivial.
\begin{remark}[A naive limit]
\label{thm:N-order-soliton:naive-limit}
We demonstrate here that taking the limit $\epsilon_s\to0$ for all $s = 1,2,\dots,N$ directly will yield a trivial result. Note that this naive limit does not rescale the norming constants $\omega_s$. Then,
\begin{equation}
\sum_s \frac{\omega_s}{\lambda - \lambda_s}
    = \sum_s \frac{\omega_s}{\lambda - \lambda_\circ - \epsilon_s}
    \to \sum_s\frac{\omega_s}{\lambda - \lambda_\circ}
     = \frac{\sum_s \omega_s}{\lambda - \lambda_\circ}\,,
\end{equation}
as $\epsilon_s\to0$ for all $s$.
So, the $N$ simple poles $\frac{\omega_s}{\lambda - \lambda_s}$ fuse into one simple pole at $\lambda = \lambda_\circ$. Correspondingly, the $N$-soliton solution becomes an one-soiton solution, with its eigenvalue $\lambda = \lambda_\circ$ and norming constant $\sum_s\omega_s$.
Thus, in order to get the $N$th-order soliton, one has to rescale the norming constants $\omega_s$, i.e., assuming $\omega_s = \omega_s(\epsilon_1,\dots\epsilon_N)$. It remains to show that such rescaling exists.
\end{remark}
Let us calculate the sum from Equation~\eqref{e:N-order-soliton:pole-sum} more carefully
\begin{equation}
\label{e:N-order-soliton:pole-sum1}
\sum_{s=1}^N \frac{\omega_s}{\lambda - \lambda_s}
    = \frac{\sum_{s=1}^N\omega_s \myblue{g_s(\lambda)}}{\prod_{s=1}^N(\lambda - \lambda_s)}\,,\quad \mbox{where we define}\quad
\myblue{g_s(\lambda)} \coloneqq \frac{\prod_{l=1}^N ( \lambda - \lambda_l)}{\lambda - \lambda_s}\,.
\end{equation}
Note that the quantity \myblue{$g_s(\lambda)$} is a polynomial of $\lambda$ of degree $N - 1$.
Thus,
\myblue{$\sum_s \omega_s g_s(\lambda)$} is also a polynomial of degree at most $N - 1$.
The leading term of \myblue{$g_s(\lambda)$} is simply $\lambda^{N-1}$,
and the leading term of \myblue{$\sum_s \omega_s g_s(\lambda)$} is
$\lambda^{N-1}\sum_s \omega_s$.
For now,
let us assume $\sum_s\omega_s \in\Complex\setminus\{0\}$,
so that the degree of \myblue{$\sum_s \omega_s g_s(\lambda)$} is exactly $N-1$.
\begin{lemma}
\label{thm:N-order-soliton:p-sol}
For the given polynomial $p_\circ(\lambda)$ from \myblue{Equation~\eqref{e:p-circ-def}},
the following equation for unknowns $\{\omega_s\}_{s=1}^N$
\myblue{
\begin{equation}
\label{e:N-order-soliton:omega-solutions}
\sum_{s=1}^N\omega_s g_s(\lambda)
    = p_\circ(\lambda)
\end{equation}
}
with $\lambda\in\Complex$ has a unique solution.
\end{lemma}
\begin{proof}
Recall that the given polynomial $p_\circ(\lambda)$ from RHP~\ref{rhp:N-order-soliton:jump-form}
has the coefficients $\{\omega_{\circ,k}\}_{k=0}^{N-1}$ with
$\omega_{\circ,k}\in\Complex$ and $\omega_{\circ,0}\ne0$.
We show that the unknowns $\{\omega_l\}_{l = 1}^N$
can be solved from Equation~\eqref{e:N-order-soliton:omega-solutions} explicitly.
It is important to notice that the polynomial \myblue{$g_s(\lambda)$} of degree $N-1$
has roots $\{\lambda_1,\dots,\lambda_{s-1},\lambda_{s +1},\dots,\lambda_N\}$.
Thus,
substituting $\lambda = \lambda_s$ into Equation~\eqref{e:N-order-soliton:omega-solutions}
yields \myblue{$\omega_s g_s(\lambda_s) = p_\circ(\lambda_s)$}, which implies
\myblue{
\begin{equation}
\label{e:N-order-soliton:omega-rescaling}
\omega_s
 = \frac{p_\circ(\lambda_s)}{g_s(\lambda_s)}
 = \frac{\sum_{k = 0}^{N-1}\omega_{\circ,k}\epsilon_s^k}{\prod_{k\ne s}(\epsilon_s - \epsilon_k)}\,,\qquad
s = 1,2,\dots, N-1\,.
\end{equation}
}
Hence, the unknowns $\{\omega_1,\dots, \omega_N\}$ are solved explicitly and uniquely.
\end{proof}
Lemma~\ref{thm:N-order-soliton:p-sol} implies that the rescaling
$\omega_s = \omega_s(\epsilon_1,\dots,\epsilon_N)$ exists.
\myblue{
Note that the limit of each $\omega_s$ from Equation~\eqref{e:N-order-soliton:omega-rescaling}
may not exist as $\epsilon_k\to0$ for all $k$.
However, what appears in RHP~\ref{rhp:N-soliton-jump-form-bigcircle}
is the sum $\sum_{s=1}^N\frac{\omega_s}{\lambda-\lambda_s}$,
whose limit exists.
}
We next show that the $N$-soliton solution can be transformed into an \myblue{$N$th-order} soliton.
Lemma~\ref{thm:N-order-soliton:p-sol} implies that Equation~\eqref{e:N-order-soliton:pole-sum1} becomes
$\sum_{s=1}^N \frac{\omega_s}{\lambda - \lambda_s} = p_\circ(\lambda)/\prod_{s=1}^N(\lambda - \lambda_s)$.
Then, taking limits $\lambda_s\to\lambda_\circ$ for all $s = 1,2,\dots,N$ with the particular $\omega_s$ from Equation~\eqref{e:N-order-soliton:omega-rescaling}, the sum becomes
$\sum_{s=1}^N \frac{\omega_s}{\lambda - \lambda_s} \to p_\circ(\lambda)/(\lambda - \lambda_\circ)^N$.
One can treat the poles in the lower half plane in a similar way, what we omit here.
Hence, RHP~\ref{rhp:N-soliton-jump-form-bigcircle} immediately becomes RHP~\ref{rhp:N-order-soliton:jump-form}, where the contour $\partial D_\Lambda^\epsilon$ becomes $\partial D_{\lambda_\circ}^\epsilon$ in the upper half plane.

Finally, we conclude that by taking the limit $\lambda_s\to\lambda_\circ$ for all $s = 1,2,\dots,N$ of the $N$-soliton solution of MBEs, with particular rescaling of the norming constants $\{\omega_s\}_{s=1}^N$ described in Equation~\eqref{e:N-order-soliton:omega-rescaling}, the $N$-soliton solution becomes an $N$th-order soliton, described by RHP~\ref{rhp:N-order-soliton:jump-form}.

\subsection{Derivation of the pole form of RHP~\ref{rhp:N-order-soliton:jump-form}}
\label{s:N-order-soliton:pole-form}

In this section, we rewrite RHP~\ref{rhp:N-order-soliton:jump-form} for the $N$th-order soliton in an equivalent form, where the jump conditions become appropriate pole conditions, mimicking the residue conditions in RHP~\ref{rhp:N-soliton-residue-form} for the $N$-soliton solutions. To start, let use define a new matrix function
\begin{equation}
\N_\circ(\lambda;t,z)
    \coloneqq \begin{cases}
        \M_\circ(\lambda;t,z)\,,& \lambda\in\Complex\setminus\{D_{\lambda_\circ}^\epsilon, D_{\conj{\lambda_\circ}}^\epsilon\}\,,\\
        \M_\circ(\lambda;t,z)\V_\circ(\lambda;t,z)\,, & \lambda\in D_{\lambda_\circ}^\epsilon\,,\\
        \M_\circ(\lambda;t,z) \V_\circ(\conj\lambda;t,z)^{-\dagger}\,, & \lambda \in D_{\conj{\lambda_\circ}}^\epsilon\,,
\end{cases}
\end{equation}
where $\M_\circ(\lambda;t,z)$ and $\V_\circ(\lambda;t,z)$
are defined in RHP~\ref{rhp:N-order-soliton:jump-form}.
It is easy to verify that:
\myblue{$\N_\circ(\lambda;t,z)\to\I$} as $\lambda\to\infty$;
$\N_\circ(\lambda;t,z)$ does not admit any jumps on
$\partial D_{\lambda_\circ}^\epsilon$ and $\partial D_{\conj{\lambda_\circ}}^\epsilon$;
and $\N_\circ(\lambda;t,z)$ is meromorphic on $\Complex$,
with poles at $\lambda = \lambda_\circ$ and $\lambda = \conj{\lambda_\circ}$.
The only thing left is to derive the pole conditions for $\N_\circ(\lambda;t,z)$.
Recall the pole operator $\sP{n}$ from Definition~\ref{def:pole-operator},
which can be used to compute the Laurent coefficient of $\N_\circ(\lambda;t,z)$ at $\lambda = \lambda_\circ$ and $\lambda = \conj{\lambda_\circ}$. Here, we show detailed steps for the case of $\lambda = \lambda_\circ$. The other point can be addressed similarly.
\begin{equation}
\begin{aligned}
\sP{-n}_{\lambda = \lambda_\circ}\N_\circ(\lambda;t,z)
    = \sP{-n}_{\lambda = \lambda_\circ}\bpm \M_{\circ,1} + \frac{p_\circ(\lambda)}{(\lambda - \lambda_\circ)^N}\ee^{-2\ii\theta}\M_{\circ,2} & \M_{\circ,2}\epm\,,
\end{aligned}
\end{equation}
where the subscripts $1$ and $2$ denote the first and second columns of the matrix, respectively, and $n = 1,2,\dots, N$. Also, we suppress the variable dependence of $\M_\circ$ in the above and below equations for simplicity. Recall that $\M_\circ$ is analytic at $\lambda = \lambda_\circ$, so by distributing $\sP{-n}$ in each column, one gets
\begin{equation}
\label{e:N-order-soliton:N-circ-pole-condition}
\begin{aligned}
\sP{-n}_{\lambda = \lambda_\circ}\N_\circ(\lambda;t,z)
    & = \sP{-n}_{\lambda = \lambda_\circ}\M_\circ\bpm 0 & 0 \\ \frac{p_0(\lambda)}{(\lambda - \lambda_\circ)^N}\ee^{-2\ii\theta} & 0 \epm\,,\qquad
n = 1,2,\dots, N\,.
\end{aligned}
\end{equation}
It is necessary to calculate $\sP{-n}\limits_{\lambda = \lambda_\circ}\M_\circ \frac{p_\circ(\lambda)}{(\lambda - \lambda_\circ)^N}\ee^{-2\ii\theta}$. The analyticity of $p_0(\lambda)\ee^{-2\ii\theta(\lambda;t,z)}\M_\circ(\lambda;t,z)$ at $\lambda = \lambda_\circ$ yields
$\sP{-n}\limits_{\lambda = \lambda_\circ}\frac{p_\circ(\lambda)\ee^{-2\ii\theta}}{(\lambda - \lambda_\circ)^N}\M_\circ = \sP{N-n}\limits_{\lambda=\lambda_\circ}p_\circ(\lambda)\ee^{-2\ii\theta}\M_\circ$, for $n=1,2,\dots,N$.
Using the expression for $p_\circ(\lambda)$ in Equation~\eqref{e:p-circ-def}, the computation continues
\begin{equation}
\begin{aligned}
\sP{-n}_{\lambda = \lambda_\circ}\frac{p_\circ(\lambda)\ee^{-2\ii\theta}}{(\lambda - \lambda_\circ)^N}\M_\circ
    & = \sum_{i = 0}^{N-n}\sP{N-n}_{\lambda = \lambda_\circ} \omega_{\circ,i}(\lambda - \lambda_\circ)^i\ee^{-2\ii\theta}\M_\circ + \sum_{i = N-n+1}^{N-1}\sP{N-n}_{\lambda = \lambda_\circ} \omega_{\circ,i}(\lambda - \lambda_\circ)^i\ee^{-2\ii\theta}\M_\circ\\
\end{aligned}
\end{equation}
Clearly, terms in the second summation are all zeros. Thus, the equation continues
\begin{equation}
\label{e:N-order-soliton:N-circ-pole}
\begin{aligned}
\sP{-n}_{\lambda = \lambda_\circ}\frac{p_\circ(\lambda)\ee^{-2\ii\theta}}{(\lambda - \lambda_\circ)^N}\M_\circ
    = \sum_{i = 0}^{N-n}\sP{N-n}_{\lambda = \lambda_\circ} \omega_{\circ,i}(\lambda - \lambda_\circ)^i\ee^{-2\ii\theta}\M_\circ
    = \sum_{i = 0}^{N-n}\omega_{\circ,i}\sP{N-n-i}_{\lambda = \lambda_\circ}\ee^{-2\ii\theta}\M_\circ\,,
\end{aligned}
\end{equation}
for $n = 1,2,\dots,N$. Due to analyticity, the Taylor expansion of the product $\ee^{-2\ii\theta}\M_\circ$ at $\lambda = \lambda_\circ$ is
\begin{equation}
\ee^{-2\ii\theta}\M_\circ
    = \sum_{k = 0}^\infty \frac{(\lambda - \lambda_\circ)^k}{k!}\frac{\partial^k(\ee^{-2\ii\theta}\M_\circ)}{\partial \lambda^k}(\lambda_\circ)\,.
\end{equation}
Thus, each Laurent coefficient in Equation~\eqref{e:N-order-soliton:N-circ-pole} can be calculated as (with help of the product rule)
\begin{equation}
\sP{N-n-i}_{\lambda = \lambda_\circ}\ee^{-2\ii\theta}\M_\circ
    = \frac{1}{(N-n-i)!}\sum_{k=0}^{N-n-i}\binom{N-n-i}{k}\frac{\partial^{N-n-i-k}\ee^{-2\ii\theta}}{\partial\lambda^{N-n-i-k}}(\lambda_\circ)\frac{\partial^k\M_\circ}{\partial\lambda^k}(\lambda_\circ)\,.
\end{equation}
Therefore, one can continue the calculation of the pole condition from Equation~\eqref{e:N-order-soliton:N-circ-pole-condition}
\begin{equation}
\label{e:N-order-soliton:N-circ-pole-condition1}
\begin{aligned}
\sP{-n}_{\lambda = \lambda_\circ}\N_\circ
    & = \sP{-n}_{\lambda=\lambda_\circ}\M_\circ \frac{p_\circ(\lambda)}{(\lambda - \lambda_\circ)^N}\ee^{-2\ii\theta}\bpm 0 & 0 \\ 1 & 0 \epm\\
    & = \sum_{i=0}^{N-n}\omega_{\circ,i}\sP{N-n-i}_{\lambda = \lambda_\circ}\ee^{-2\ii\theta}\M_\circ\bpm 0 & 0 \\ 1 & 0 \epm\\
    & = \sum_{i=0}^{N-n}\frac{\omega_{\circ,i}}{(N-n-i)!}\sum_{k=0}^{N-n-i}\binom{N-n-i}{k}\frac{\partial^{N-n-i-k}\ee^{-2\ii\theta}}{\partial\lambda^{N-n-i-k}}(\lambda_\circ)\frac{\partial^k\M_\circ}{\partial\lambda^k}(\lambda_\circ)\bpm 0 & 0 \\ 1 & 0 \epm\,.
\end{aligned}
\end{equation}
Next, we need to convert the derivatives of $\M_\circ$ back to the ones of $\N_\circ$ on the right-hand side of the above equation. Recall that near $\lambda = \lambda_\circ$, it is $\M_\circ = \N_\circ\V_\circ^{-1}$, so
\begin{equation}
\frac{\partial^k\M_\circ}{\partial\lambda^k}
    = \frac{\partial^k(\N_\circ\V_\circ^{-1})}{\partial\lambda^k}
    = \sum_{j=0}^k\binom{k}{j}\frac{\partial^j\N_\circ}{\partial\lambda^j}\frac{\partial^{k-j}(\V_\circ^{-1})}{\partial\lambda^{k-j}}\,,\qquad \lambda\in D_{\lambda_\circ}^\epsilon\setminus\{\lambda_\circ\}\,.
\end{equation}
Recall also that $\V_\circ^{-1} = \bpm 1 & 0 \\ -\frac{p_\circ\ee^{-2\ii\theta}}{(\lambda - \lambda_\circ)^N} & 1 \epm$, so that $\frac{\partial^{k-j}(\V_\circ^{-1})}{\partial\lambda^{k-j}} = \bpm 0 & 0 \\ -\frac{\partial^{k-j}}{\partial\lambda^{k-j}}\big(\frac{p_\circ\ee^{-2\ii\theta}}{(\lambda - \lambda_\circ)^N}\big) & 0 \epm$ when $k-j > 0$. Such lower triangular matrices with zero diagonals when multiplying by $\bpm 0 & 0 \\ 1 & 0 \epm$ yields the zero matrix. As a result, the product in Equation~\eqref{e:N-order-soliton:N-circ-pole-condition1} becomes
\begin{equation}
\label{e:N-order-soliton:N-circ-pole-condition2}
\begin{aligned}
\frac{\partial^k\M_\circ}{\partial\lambda^k}\bpm 0 & 0 \\ 1 & 0 \epm
    & = \frac{\partial^k\N_\circ}{\partial\lambda^k}\bpm 0 & 0 \\ 1 & 0 \epm\,,\qquad \lambda\in D_{\lambda_\circ}^\epsilon\setminus\{\lambda_\circ\}\,.
\end{aligned}
\end{equation}
Then, the combination of Equations~\eqref{e:N-order-soliton:N-circ-pole-condition1} and~\eqref{e:N-order-soliton:N-circ-pole-condition2} yields
\begin{equation}
\begin{aligned}
\sP{-n}_{\lambda = \lambda_\circ}\N_\circ
    & = \lim_{\lambda\to\lambda_\circ}\sum_{k=0}^{N-n}\sum_{i=0}^{N-n-k}\frac{\omega_{\circ,i}}{(N-n-i-k)!\,k!}\frac{\partial^{N-n-i-k}\ee^{-2\ii\theta}}{\partial\lambda^{N-n-i-k}}\frac{\partial^k\N_\circ}{\partial\lambda^k}\bpm 0 & 0 \\ 1 & 0 \epm\\
\end{aligned}
\end{equation}
Similarly, one can derive the pole conditions for $N_\circ(\lambda;t,z)$ at $\lambda = \conj{\lambda_\circ}$
\begin{equation}
\sP{-n}_{\lambda = \conj{\lambda_\circ}}\N_\circ
    = -\lim_{\lambda\to\lambda_\circ}\sum_{k=0}^{N-n}\sum_{i=0}^{N-n-k}\frac{\conj{\omega_{\circ,i}}}{(N-n-i-k)!\,k!}\frac{\partial^{N-n-i-k}\ee^{2\ii\theta}}{\partial\lambda^{N-n-i-k}}\frac{\partial^k\N_\circ}{\partial\lambda^k}\bpm 0 & 1 \\ 0 & 0 \epm\,.
\end{equation}
Finally, let us rename the matrix $\N_\circ(\lambda;t,z)$ back to $\M_\circ(\lambda;t,z)$. Then, RHP~\ref{rhp:N-order-soliton:pole-form} is obtained.

\subsection{$N$th-order soliton solution formula}
\label{s:N-order-soliton:solution-formula}

In this section, we show how to derive the $N$th-order soliton solution formula from RHP~\ref{rhp:N-order-soliton:pole-form}. For brevity, we suppress the variable dependence $(\lambda;t,z)$ in this section. Due to the particular form of the RHP, we assume the following ansatz for the matrix $\M_\circ(\lambda;t,z)$
\begin{equation}
\label{e:N-order-soliton:M-circ-solution0}
\M_\circ = \I + \sum_{n=1}^N\frac{1}{(\lambda - \lambda_\circ)^n}\sP{-n}\limits_{\lambda = \lambda_\circ}\M_\circ + \sum_{n=1}^N\frac{1}{(\lambda - \conj{\lambda_\circ})^n}\sP{-n}_{\lambda = \conj{\lambda_\circ}}\M_\circ\,.
\end{equation}
As discussed in Remark~\ref{thm:M-symmetry}, one only needs to reconstruct the top row, so let us denote
\begin{equation}
X \coloneqq M_{\circ,1,1}\,,\qquad Y \coloneqq M_{\circ,1,2}\,.
\end{equation}
Thus, the top row of Equation~\eqref{e:N-order-soliton:M-circ-solution0} can be expressed explicitly as
\begin{equation}
\bpm X & Y \epm = \bpm 1 & 0 \epm + \sum_{n=1}^N\frac{1}{(\lambda - \lambda_\circ)^n}\sP{-n}_{\lambda = \lambda_\circ}\bpm X & Y \epm + \sum_{n=1}^N\frac{1}{(\lambda - \conj{\lambda_\circ})^n}\sP{-n}_{\lambda = \conj{\lambda_\circ}}\bpm X & Y \epm\,.
\end{equation}
Substituting both pole conditions into the above equation, one gets
\begin{equation}
\begin{aligned}
X
    & = 1 + \sum_{n=1}^N\frac{1}{(\lambda - \lambda_\circ)^n}\lim_{\lambda\to\lambda_\circ}\sum_{k=0}^{N-n}\sum_{i=0}^{N-n-k}\frac{\omega_{\circ,i}}{(N-n-i-k)!\,k!}\frac{\partial^{N-n-i-k}\ee^{-2\ii\theta}}{\partial\lambda^{N-n-i-k}}\frac{\partial^k Y}{\partial\lambda^k}\,,\\
Y
    & = -\sum_{n=1}^N\frac{1}{(\lambda - \conj{\lambda_\circ})^n}\lim_{\lambda\to\conj{\lambda_\circ}}\sum_{k=0}^{N-n}\sum_{i=0}^{N-n-k}\frac{\conj{\omega_{\circ,i}}}{(N-n-i-k)!\,k!}\frac{\partial^{N-n-i-k}\ee^{2\ii\theta}}{\partial\lambda^{N-n-i-k}}\frac{\partial^k X}{\partial\lambda^k}\,.
\end{aligned}
\end{equation}
Clearly, $X$ is analytic at $\lambda = \conj{\lambda_\circ}$ and $Y$ is analytic at $\lambda = \lambda_\circ$, so the limits can be computed as
\begin{equation}
\begin{aligned}
X
    & = 1 + \sum_{n=1}^N\frac{1}{(\lambda - \lambda_\circ)^n}\sum_{k=0}^{N-n}\sum_{i=0}^{N-n-k}\frac{\omega_{\circ,i}}{(N-n-i-k)!\,k!}\frac{\partial^{N-n-i-k} \ee^{-2\ii\theta}}{\partial\lambda^{N-n-i-k}}(\lambda_\circ)\frac{\partial^k Y}{\partial\lambda^k}(\lambda_\circ)\,,\\
Y
    & = -\sum_{n=1}^N\frac{1}{(\lambda - \conj{\lambda_\circ})^n}\sum_{k=0}^{N-n}\sum_{i=0}^{N-n-k}\frac{\conj{\omega_{\circ,i}}}{(N-n-i-k)!\,k!}\frac{\partial^{N-n-i-k} \ee^{2\ii\theta}}{\partial\lambda^k}(\conj{\lambda_\circ})\frac{\partial^k X}{\partial\lambda^k}(\conj{\lambda_\circ})\,.
\end{aligned}
\end{equation}
We next change the order of sums as $\sum_{n=1}^{N}\sum_{k=0}^{N-n} = \sum_{k=0}^{N-1}\sum_{n=1}^{N-k}$ and use the quantity $\gamma_k(\lambda)$ defined in Theorem~\ref{thm:N-order-soliton:solution-formula} to simplify the equations,
\begin{equation}
\label{e:N-order-soliton:XY-system}
X(\lambda) = 1 + \sum_{k=0}^{N-1}\gamma_k(\lambda)\frac{\partial^k Y}{\partial\lambda^k}(\lambda_\circ)\,,\qquad
Y(\lambda) = - \sum_{k=0}^{N-1}\gamma_k^*(\lambda)\frac{\partial^k X}{\partial\lambda^k}(\conj{\lambda_\circ})\,.
\end{equation}
Let us define more quantities in order to solve the above system
\begin{equation}
X_s \coloneqq X^{(s)}(\conj{\lambda_\circ})\,,\qquad
Y_s \coloneqq Y^{(s)}(\lambda_\circ)\,,\qquad
s = 0,1,\dots,N-1\,.
\end{equation}
Substituting $\lambda = \lambda_\circ$ and $\lambda = \conj{\lambda_\circ}$ in Equation~\eqref{e:N-order-soliton:XY-system}, one gets
\begin{equation}
\label{e:N-order-soliton:XY-system1}
X_0 = 1 + \sum_{k=0}^{N-1}\gamma_k(\conj{\lambda_\circ})Y_k\,,\qquad
Y_0 = -\sum_{k=0}^{N-1}\gamma_k^*(\lambda_\circ)X_k\,.
\end{equation}
Moreover, by taking derivatives of the system~\eqref{e:N-order-soliton:XY-system} $s$ times and substituting $\lambda = \lambda_\circ$ and $\lambda = \conj{\lambda_\circ}$, one gets more equations
\begin{equation}
\label{e:N-order-soliton:XY-system2}
X_s = \sum_{k=0}^{N-1}\gamma_k^{(s)}(\conj{\lambda_\circ})Y_k\,,\qquad
Y_s = -\sum_{k=0}^{N-1}\big(\gamma_k^{(s)}\big)^*(\lambda_\circ)X_k\,,\qquad
s = 1,2,\dots,N-1\,.
\end{equation}
Equations~\eqref{e:N-order-soliton:XY-system1} and~\eqref{e:N-order-soliton:XY-system2} form a linear system for $\{X_s,Y_s\}_{s=0}^{N-1}$, which can be solved in the same manner as in Section~\ref{s:derivation-Nsoliton}. The final solution formula is given in Theorem~\ref{thm:N-order-soliton:solution-formula}.

\section*{Acknowledgments}

\myblue{
We are grateful to the referees for their valuable comments and suggestions.
}
This project was partially supported by
the National Natural Science Foundation of China (No.~12201526),
the Fundamental Research Funds for the Central Universities (No.~20720220040),
and the Natural Science Foundation of Fujian Province of China (No.~2022J01032).

\section*{Data availability}
Data sharing not applicable to this article as no datasets were generated or analyzed during
the current study.

\begin{appendices}
\section{Soliton gas}
\label{s:soliton-gas}

\myblue{
In this appendix,
we explain how one can take the formal limit $N\to+\infty$ of an $N$-DSG and derive an
$\infty$-DSG,
i.e., \myblue{a bound state soliton gas}.
The derivation is similar to~\cite{ggjm2021}.
It is easy to verify that RHP~\ref{rhp:N-soliton-jump-form} with $J = 1$ can be equivalently turned into the following one with an enlarged jump contour.
}
\begin{rhp}
\label{rhp:N-DSG-jumpcombined}
\myblue{
Suppose $J=1$ and $N \ge 1$,
and the sets $\Lambda$,
$\Omega$ are given according to Definition~\ref{def:LambdaOmega}.
Seek a
$2\times2$ matrix function $\lambda\mapsto\M^{(4)}(\lambda;t,z)$ analytic on
$\Complex\setminus\partial\big(D_{\Lambda}^\epsilon\bigcup D_{\conj{\Lambda}}^\epsilon\big)$ with continuous boundary values.
It has the asymptotics $\M^{(4)}(\lambda;t,z)\to\I$ as $\lambda\to\infty$ and satisfies jumps
}
\begin{equation}
\begin{aligned}
\M^{(4)+}(\lambda;t,z)
 & = \M^{(4)-}(\lambda;t,z)\V^{(2)}(\lambda;t,z)^{-1}\,,\qquad&& \lambda\in\partial D_{\Lambda}^\epsilon\,,\\
\M^{(4)+}(\lambda;t,z)
 & = \M^{(4)-}(\lambda;t,z)\V^{(2)}(\conj\lambda;t,z)^\dagger\,,\qquad&& \lambda\in\partial D_{\conj{\Lambda}}^\epsilon\,,
\end{aligned}
\end{equation}
where the jump matrix $\V^{(2)}(\lambda;t,z)$ is given by
\begin{equation}
\V^{(2)}(\lambda;t,z)
 \coloneqq \bpm1 & 0 \\ \ee^{-2\ii\theta(\lambda;t,z)}\sum_{k=1}^N\frac{\omega_{1,k}}{\lambda-\lambda_{1,k}} & 1 \epm\,,
\end{equation}
and the jump contours $\partial D_\Lambda^\epsilon$ and $\partial D_{\conj{\Lambda}}^\epsilon$ are illustrated in Figure~\ref{f:N-DSG-jumpcombined}(center) with counterclockwise orientation.
\end{rhp}
The region $D_{\Lambda}^\epsilon\subset\Complex^+$
encloses all the discrete eigenvalues $\Lambda$ in the upper half plane,
and similarly,
$D_{\conj{\Lambda}}^\epsilon\subset\Complex^-$ encloses $\conj{\Lambda}$
in the lower half plane.
Here, $\conj{\Lambda}$ denote the complex conjugate of the set $\Lambda$,
not its closure.
Therefore,
the two jump contours $\partial D_{\Lambda}^\epsilon$ and
$\partial D_{\conj{\Lambda}}^\epsilon$ never intersect.
Furthermore,
we assume that $\partial D_\Lambda^\epsilon$ is sufficiently close to $\Lambda$
and contains an arc of the circle with radius $r_1$,
as illustrated in Figure~\ref{f:N-DSG-jumpcombined}(left and center).
\begin{figure}[tp]
\centering
\includegraphics[width=0.62\textwidth]{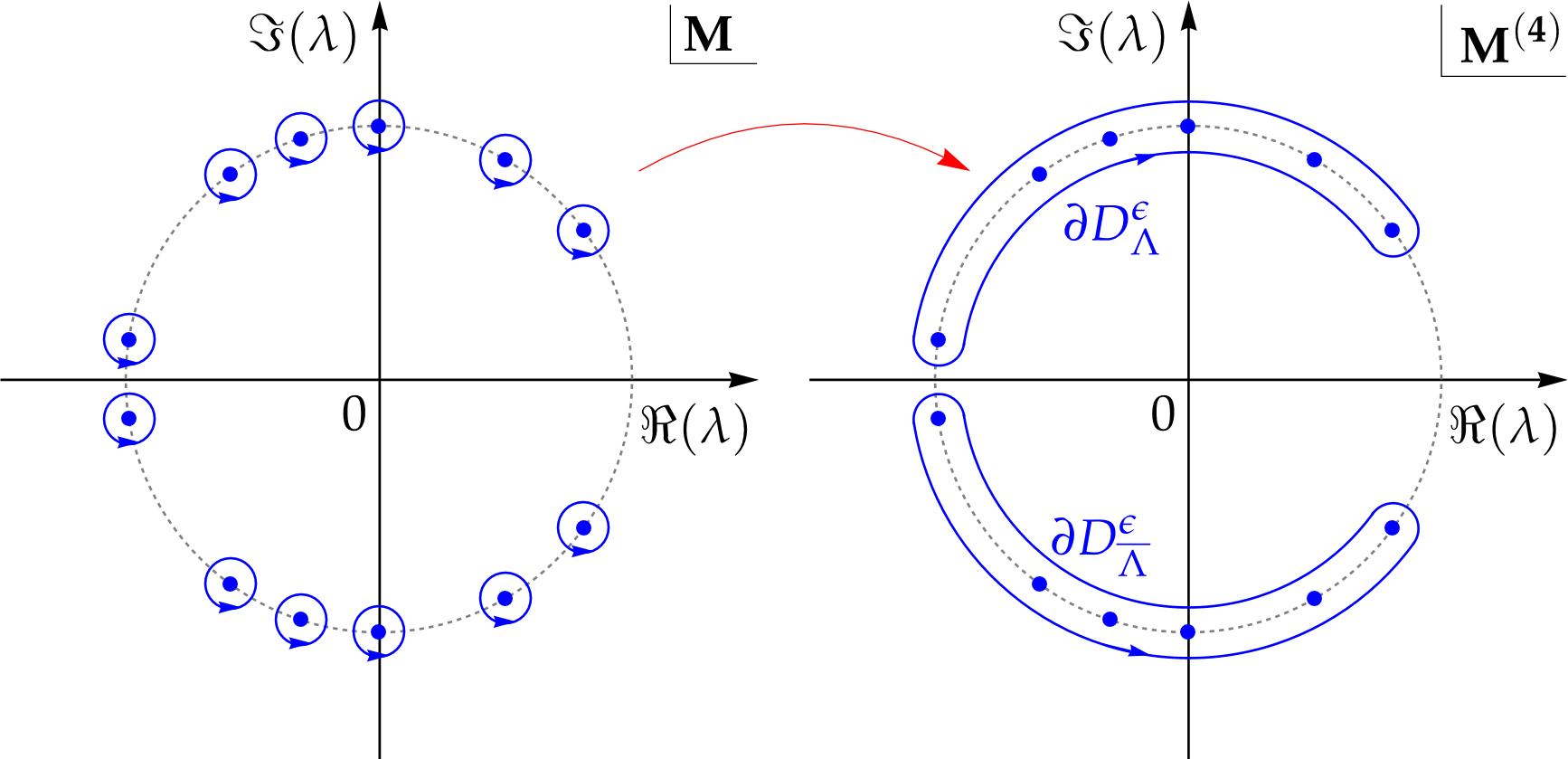}\quad
\includegraphics[width=0.3\textwidth]{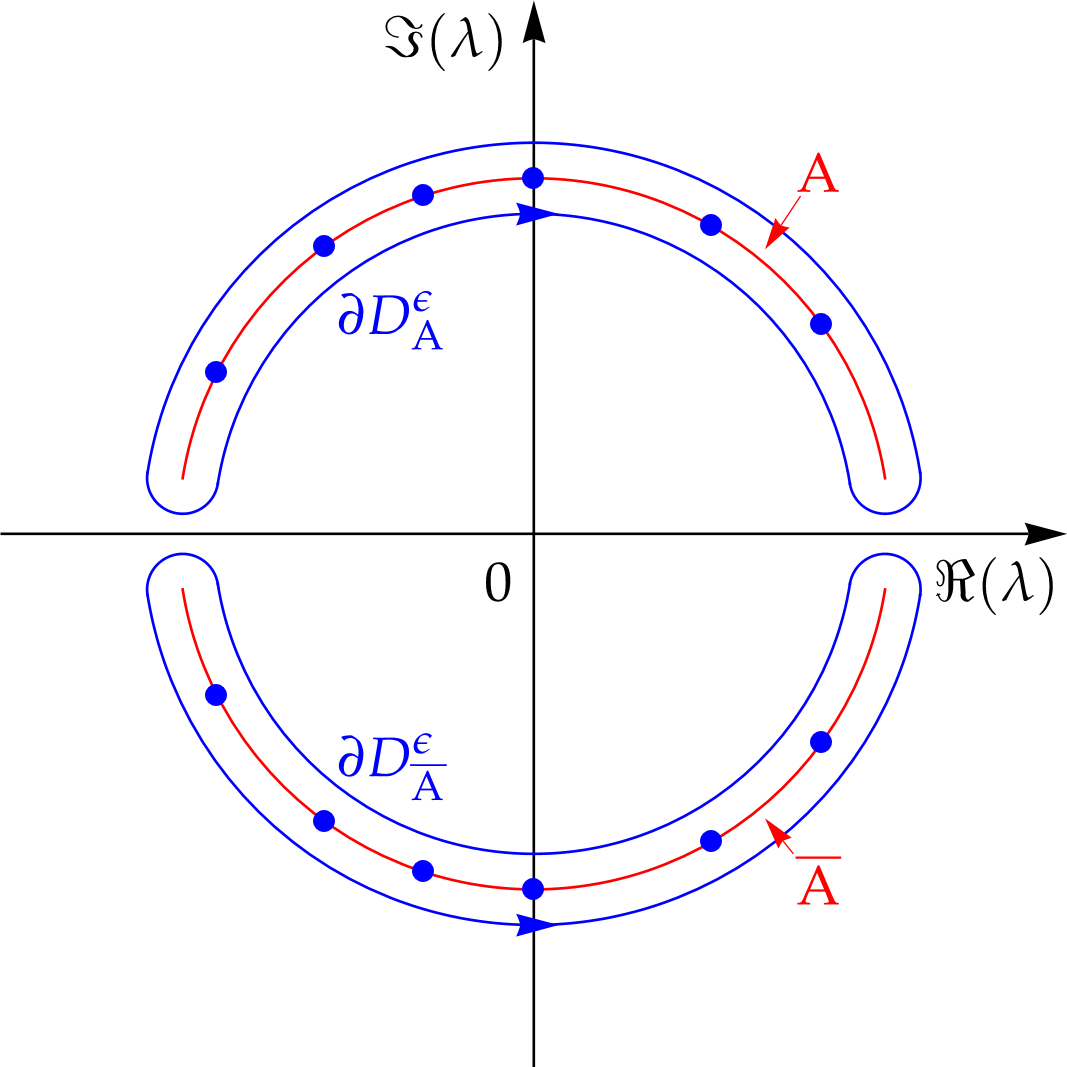}
\caption{
Left and center: The transformation from the $N$-DSG RHP~\ref{rhp:N-soliton-jump-form} with function $\M(\lambda;t,z)$ and $J = 1$ to an equivalent RHP~\ref{rhp:N-DSG-jumpcombined} with combined jump contours.
Right:
An illustration of:
(i) the contour $\rA$ and $\conj{\rA}$ from Equation~\eqref{e:soliton-gas:A_def} shown in red;
(ii) the regions $D_\rA^\epsilon$ and $D_{\conj{\rA}}^\epsilon$ with boundaries shown in blue,
oriented counterclockwise.
The blue dot are representations of discrete eigenvalues $\Lambda$ and $\conj{\Lambda}$.
}
\label{f:N-DSG-jumpcombined}
\end{figure}

\myblue{
We would like to take the formal limit of RHP~\ref{rhp:N-DSG-jumpcombined}
as $N\to+\infty$ in a controllable way,
in the sense that all discrete eigenvalues in $\Lambda_1 = \Lambda$
are accumulating on an arc centered at the origin.
Recall that the arc $\rA$ is defined in Definition~\ref{def:generalized-eigenvalue-norming-constant},
so we know that $\Lambda\subset\rA$.
Correspondingly,
one defines a new region
$D_{\rA}^\epsilon\subset\Complex^+$
which contains the whole arc $\rA$ and does not intersects with the real axis.
Please see an illustration of the contour $\rA$
and corresponding region $D_{\rA}^\epsilon$
shown in Figure~\ref{f:N-DSG-jumpcombined}(right).
}

\myblue{
The poles in RHP~\ref{rhp:N-DSG-jumpcombined} are encoded
in the jump matrix $\V^{(2)}(\lambda;t,z)$ via the sum in the $(1,2)$ component.
For convenience,
we rewrite the relevant function as follows,
}
\begin{equation}
\label{e:soliton-gas:pole_sum1}
\sum_{k=1}^N\frac{\omega_{1,k}}{\lambda-\lambda_{1,k}}
 = -\int_{\rA}\omega_1(s;N)\frac{\dd s}{s - \lambda}\,,\qquad
\omega_1(\lambda;N) \coloneqq \sum_{k = 1}^N \omega_{1,k}\delta(\lambda - \lambda_{1,k})\,,\qquad \lambda\not\in\rA\,,
\end{equation}
where $\delta(\cdot)$ is the Dirac delta distribution defined on $\rA$,
oriented left-to-right, i.e.,
$\int_{\rA}f(s)\delta(s-s_0)\dd s = f(s_0)$,
provided that $f(z)$ is continuous at every interior point $z = s_0\in \rA$.
In Equation~\eqref{e:soliton-gas:pole_sum1},
the norming constants $\omega_{1,k}$ can be recognized as weights of each Dirac delta,
and the sum is \textit{not} normalized,
$\int_{\rA} \omega_1(s;N)\dd s \neq 1$.
Hence,
the function $\omega_1(s;N)$ can be considered as an
(unnormalized) complex distribution on the arc $\rA$.
\myblue{
Assuming $\omega_1(\lambda;N)$ interpolates a continuous
(unnormalized) complex distribution $\omega(\lambda)/(2\pi\ii)$ on $\rA$,
with $\omega(\lambda)$ square integrable,
one formally gets
\begin{equation}
\label{e:soliton-gas:pole_sum_to_integral}
\begin{aligned}
\lim_{N\to+\infty}\sum_{k=1}^N\frac{\omega_{1,k}}{\lambda-\lambda_{1,k}}
 & = -\lim_{N\to+\infty}\int_\rA\omega_1(s;N)\frac{\dd s}{s- \lambda}
 = -\frac{1}{2\pi\ii}\int_{\rA}\frac{\omega(s)\dd s}{s - \lambda}\,, &&\quad
\lambda\not\in\rA\,,\\
\lim_{N\to+\infty}\sum_{k=1}^N\frac{\conj{\omega_{1,k}}}{\lambda - \conj{\lambda_{1,k}}}
 & = -\lim_{N\to+\infty}\int_{\conj{\rA}} \omega_1^*(s;N)\frac{\dd s}{s- \lambda}
 = -\frac{1}{2\pi\ii}\int_{\conj{\rA}}\frac{\omega^*(s)\dd s}{s - \lambda}\,, &&\quad
\lambda\not\in\conj{\rA}\,.
\end{aligned}
\end{equation}
where the integration contour is oriented left-to-right.
Please see Figure~\ref{f:soliton-gas:soliton_gas_limit} for demonstration.
We point out that
Equation~\eqref{e:soliton-gas:pole_sum_to_integral}
can be realized using the definition of Riemann integrals,
assuming special forms of the norming constants~\cite{ggjm2021,ggjmm2023}.
}

\begin{figure}[tp]
\centering
\includegraphics[width=0.95\textwidth]{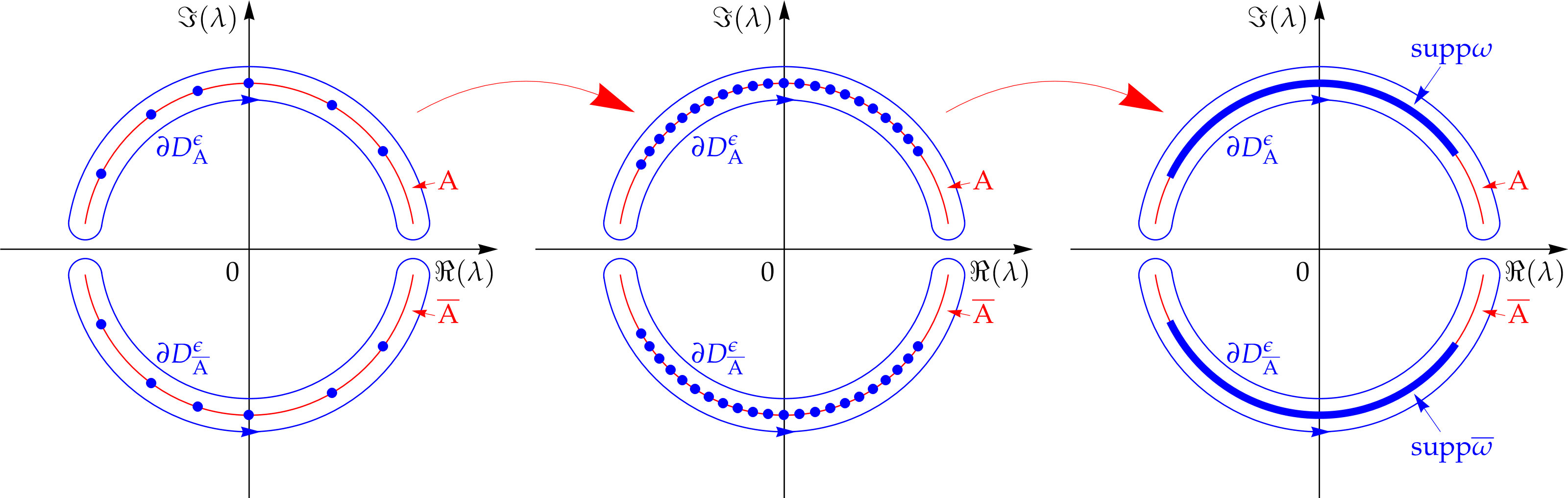}
\caption{
An illustration of the eigenvalue distribution in the limit $N\to+\infty$. In the right panel, the thick blue arcs represent $\supp\omega\subset\rA$ and $\supp\conj{\omega}\subset\conj{\rA}$.
}
\label{f:soliton-gas:soliton_gas_limit}
\end{figure}
%

%
%

\myblue{
To simplify the above integrals, one defines the Cauchy transforms
\begin{equation}
\label{e:Cauchy_Transform}
\sC_\omega(\lambda) \coloneqq \frac{1}{2\pi\ii}\int_{\rA}\frac{\omega(s)\dd s}{s - \lambda}.
\end{equation}
Then, the above limits can be rewritten as
$\lim\limits_{N\to+\infty}\sum_{k=1}^N\frac{\omega_{1,k}}{\lambda-\lambda_{1,k}}
 = -\sC_\omega(\lambda)$, and
$\lim\limits_{N\to+\infty}\sum_{k=1}^N\frac{\conj{\omega_{1,k}}}{\lambda - \conj{\lambda_{1,k}}}
 = -\sC_{\conj{\omega}}(\lambda)$.
It is easy to verify that
$\sC_\omega^*(\lambda) = - \sC_{\conj\omega}(\lambda)$.
Consequently, the jumps in RHP~\ref{rhp:N-DSG-jumpcombined} becomes
\begin{equation}
\begin{aligned}
\lim_{N\to+\infty}\V^{(2)}(\lambda;t,z)^{-1}
 = \V^{(3)}(\lambda;t,z)^{-1}\,,\quad
\lim_{N\to+\infty}\V^{(2)}(\conj\lambda;t,z)^\dagger
 = \V^{(3)}(\conj{\lambda};t,z)^\dagger\,,
\end{aligned}
\end{equation}
with
\begin{equation}
\label{e:V3-def}
\V^{(3)}(\lambda;t,z)
 \coloneqq \bpm1 & 0 \\ -\ee^{-2\ii\theta(\lambda;t,z)}\sC_\omega(\lambda) & 1 \epm\,,
\end{equation}
Following these calculations,
RHP~\ref{rhp:N-DSG-jumpcombined} as $N\to+\infty$ becomes the following one.
}
\begin{rhp}[Cauchy-form of soliton gas]
\label{rhp:soliton-gas:inftyDSG_cauchy_form}
\myblue{
Let an arc $\rA$ and a function $\omega(\lambda)$ be given according to Definition~\ref{def:generalized-eigenvalue-norming-constant}.
Seek a $2\times2$
matrix function $\M^{(5)}(\lambda;t,z)$ analytic on
$\Complex\setminus\partial\big(D_\rA^\epsilon\bigcup D_{\conj{\rA}}^\epsilon\big)$
with continuous boundary values.
It has the asymptotics $\M^{(5)}(\lambda;t,z)\to\I$ as $\lambda\to\infty$ and satisfies jumps
}
\begin{equation}
\begin{aligned}
\M^{(5)+}(\lambda;t,z)
 & = \M^{(5)-}(\lambda;t,z)\V^{(3)}(\lambda;t,z)^{-1}\,,\qquad&& \lambda\in\partial D_{\rA}^\epsilon\,,\\
\M^{(5)+}(\lambda;t,z)
 & = \M^{(5)-}(\lambda;t,z)\V^{(3)}(\conj\lambda;t,z)^{\dagger}\,,\qquad&& \lambda\in\partial D_{\conj{\rA}}^\epsilon\,,
\end{aligned}
\end{equation}
where the jump matrix $\V^{(3)}(\lambda;t,z)$ is given in Equation~\eqref{e:V3-def},
and the jump contours $\partial D_\rA^\epsilon$ and $\partial D_{\conj{\rA}}^\epsilon$ are illustrated in
Figure~\ref{f:soliton-gas:soliton_gas_limit}(right) with counterclockwise orientation.
\end{rhp}
\myblue{
Finally, one can convert RHP~\ref{rhp:soliton-gas:inftyDSG_cauchy_form}
to RHP~\ref{rhp:soliton-gas:inftyDSG_jump_form},
by shrinking the jump contour $\partial D_\rA^\epsilon$ onto the arc $\rA$ and
$\partial D_{\conj{\rA}}^\epsilon$ onto $\conj{\rA}$, respectively.
}

\end{appendices}


\end{document}